\documentclass[11pt, a4paper]{article}
\usepackage[T1]{fontenc}
\usepackage{a4wide, url}
\usepackage[english]{babel}
\usepackage{graphicx}
\usepackage{natbib}
\usepackage{hyperref}
\usepackage[textfont=sc, labelfont=bf]{caption}
\usepackage[blocks]{authblk}
\usepackage{setspace}
\usepackage{multirow}
\usepackage{amsmath, amsthm, amssymb, amsfonts}
\usepackage{pgf} 
\usepackage{bm}
\usepackage{paralist}
\usepackage{multirow}
\usepackage{color}     
\usepackage{changepage}
\usepackage{enumitem}
\usepackage{ulem}

\DeclareMathAlphabet\mathbfcal{OMS}{cmsy}{b}{n}

\newcommand{\mbf}{\mathbf}
\newcommand{\mc}{\mathcal}
\newcommand{\beq}{\begin{equation}}
\newcommand{\eeq}{\end{equation}}
\newcommand{\bea}{\begin{eqnarray}}
\newcommand{\eea}{\end{eqnarray}}
\newcommand{\beas}{\begin{eqnarray*}}
\newcommand{\eeas}{\end{eqnarray*}}
\newcommand{\bit}{\begin{itemize}}
\newcommand{\eit}{\end{itemize}}
\newcommand{\ben}{\begin{enumerate}} 
\newcommand{\een}{\end{enumerate}}
\newcommand{\bpm}{\begin{pmatrix}}
\newcommand{\epm}{\end{pmatrix}}
\newcommand{\bbm}{\begin{bmatrix}}
\newcommand{\ebm}{\end{bmatrix}}
\newcommand{\eps}{\epsilon}
\newcommand{\vep}{\varepsilon}

\renewcommand{\l}{\left}
\renewcommand{\r}{\right}

\def\wh{\widehat}
\def\wt{\widetilde}

\newcommand{\E}[0]{\mathsf{E}}
\newcommand{\Var}[0]{\mathsf{Var}}

\newcommand{\tr}[0]{\mathsf{tr}}

\newcommand{\p}{\mathsf{P}}

\newcommand{\Z}{\mathbb{Z}}

\newcommand{\ubar}{\underline}
\newcommand{\iid}{\mbox{\scriptsize{iid}}}

\newcommand{\nn}{\nonumber}

\newcommand{\cB}{\mathcal{B}}

\newcommand{\cD}{\mathcal{D}}

\newcommand{\cI}{\mathcal{I}}
\newcommand{\cL}{\mathcal{L}}
\newcommand{\cN}{\mathcal{N}}
\newcommand{\cR}{\mathcal{R}}
\newcommand{\cS}{\mathcal{S}}
\newcommand{\cT}{\mathcal{T}}
\newcommand{\cU}{\mathcal{U}}
\newcommand{\cY}{\mathcal{Y}}

\newcommand{\bbI}{\mathbb{I}}

\newcommand{\jts}{J^*_T}
\newcommand{\heta}{\wh\eta}

\theoremstyle{definition}
\newtheorem{thm}{Theorem}
\theoremstyle{definition}

\theoremstyle{definition}
\newtheorem{lem}{Lemma}
\theoremstyle{definition}
\newtheorem{prop}{Proposition}
\theoremstyle{definition}
\newtheorem{assum}{Assumption}
\theoremstyle{definition}
\newtheorem{rem}{Remark}
\theoremstyle{definition}

\theoremstyle{definition}

\title{Simultaneous multiple change-point and factor analysis for high-dimensional time series}
\author{Matteo Barigozzi$^1$ \hskip 1cm Haeran Cho$^2$ \hskip 1cm Piotr Fryzlewicz$^3$}
\date{\small{\today}}

\begin{document}

\maketitle

\begin{abstract}
We propose the first comprehensive treatment of high-dimensional time series factor models 
with multiple change-points in their second-order structure. 
We operate under the most flexible definition of piecewise stationarity,
and estimate the number and locations of change-points consistently 
as well as identifying whether they originate in the common or idiosyncratic components. 
Through the use of wavelets, we transform the problem of change-point detection in the second-order structure 
of a high-dimensional time series, into the (relatively easier) problem of change-point detection 
in the means of high-dimensional panel data. 
Also, our methodology circumvents the difficult issue of the accurate estimation of 
the true number of factors in the presence of multiple change-points by adopting a screening procedure. 
We further show that consistent factor analysis is achieved over each segment
defined by the change-points estimated by the proposed methodology.
In extensive simulation studies, we observe that factor analysis 
prior to change-point detection improves the detectability of change-points, 
and identify and describe an interesting `spillover' effect in which 
substantial breaks in the idiosyncratic components get, 
naturally enough, identified as change-points in the common components, which prompts us to regard 
the corresponding change-points as also acting as a form of `factors'. 
Our methodology is implemented in the R package {\tt factorcpt}, available from CRAN.
\end{abstract}

\thispagestyle{empty}

 \footnotetext[1]{Department of Statistics, London School of Economics, Houghton Street, London WC2A 2AE, UK.\\
Email: \url{m.barigozzi@lse.ac.uk}.} 

\footnotetext[2]{School of Mathematics, University of Bristol, University Walk, Bristol, BS8 1TW, UK.\\
Email: \url{haeran.cho@bristol.ac.uk}.}

 \footnotetext[3]{Department of Statistics, London School of Economics, Houghton Street, London WC2A 2AE, UK.\\
Email: \url{p.fryzlewicz@lse.ac.uk}.} 


\section{Introduction}
\label{sec:intro}

High-dimensional time series data abound in modern data science, 
including finance (e.g., simultaneously measured returns on a large number of assets \citep{fan2011, BH16}), 
economics (e.g., country-level macroeconomic data \citep{stock2008} or retail price index data \citep{groen2013}),
neuroimaging (e.g.,~measurements of brain activity \citep{schroder2015, barnett2016})
and biology (e.g., transcriptomics data \citep{omranian2015} or Hi-C data matrices \citep{brault2016}).

Factor modelling, in which the individual elements of a high-dimensional time series are modelled as sums
of a common component (a linear combination of a small number of possibly unknown factors),
plus each individual element's own idiosyncratic noise, 
is a well-established technique for dimension reduction in time series.
Time series factor models are classified, in relation to the effect of factors on the observed time series, 
into `static' (only a contemporaneous effect, see e.g., \cite{stockwatson02JASA, baing02, bai2003}) or 
`dynamic' (lagged factors may also have an effect, see e.g., \cite{fornilippi01, hallinlippi13}) factor models. 

It is increasingly recognised that in several important application areas, 
such as those mentioned at the beginning of this section,
nonstationary time series data are commonly observed. 
Arguably the simplest realistic departure from stationarity,
which also leads to sparse and interpretable time series modelling, 
is piecewise-stationarity, in which the time series is modelled as approximately stationary 
between neighbouring change-points, 
and changing its distribution (e.g., the mean or covariance structure) at each change-point. 

The main aim of this work is to provide the first comprehensive framework for the estimation
of time series factor models with multiple change-points in their second-order structure.
The existing literature on time series factor modelling has only partially embraced nonstationarity.
One way in which a change-point is typically handled in the literature is via 
the assumption that the structural break in the loadings is `moderate' and it affects only a limited number of series,
so that it does not adversely impact the quality of traditional stationary 
principal component analysis (PCA)-type estimation
\citep{stockwatson02JASA, stock2008, bates2013}.

However, opinions diverge on the empirically relevant degrees of temporal instability in the factor loadings,
and several authors observe that `large' changes in the stochastic data-generating mechanism have the potential
to severely distort the estimation of the factor structure.
Investigations into the effect of a {\it single} break in the loadings or the number of factors on the factor structure,
with accompanying change-point tests and estimators for its location,
can be found in \cite{breitung2011}, \cite{chen2014}, \cite{han2014}, 
\cite{corradi2014}, \cite{yamamoto2015}, \cite{baltagi2015}, \cite{bai2017} and \cite{massacci2017}.
Lasso-type estimation is considered for change-point analysis under factor modelling 
in \cite{cheng2016} and \cite{ma2016}:
the former concerns single change-point detection in the loadings and the number of factors, 
while the latter considers multiple change-point detection in loadings only. 
Note that the $\ell_1$-penalty of the Lasso is not optimal for change-point detection,
as investigated in \cite{brodsky1993} and \cite{cho2011}.
In summary, apart from \cite{ma2016} and \cite{sun2016}
(the latter considers factor models with multiple change-points but for a small number of time series only),
the existing change-point methods proposed for factor models focus on detecting a single break
of a particular type, namely a break in the loadings or the number of factors.

We now describe in detail the contribution and findings of this work, at the same time giving an overview of the organisation of the paper.
\begin{itemize}
\item[(a)] \textit{We propose a comprehensive methodology for the consistent estimation of 
multiple change-points in the second-order structure of 
a high-dimensional time series governed by a factor model.}
This is in contrast to the substantial time series factor model literature,
which is overwhelmingly concerned with testing for a single change-point. 
In practice, the possibility of the presence of multiple change-points cannot be ruled out
from a dataset consisting of observations over a long stretch of time,
as illustrated in our applications to financial and macroeconomic time series data in Section \ref{sec:real}.
Our estimators are `interpretable' in the sense that they enable the identification of whether 
each change-point originates from the common or idiosyncratic components.
Through simulation studies (Section \ref{sec:sim}),
it is demonstrated that in high-dimensional time series segmentation, 
factor analysis prior to change-point detection
improves the detectability of change-points that appear only in either of the common or the idiosyncratic components.

\item[(b)] We operate under the most flexible definition of piecewise-stationarity, 
embracing all possible structural instabilities under factor modelling:
it allows factors and idiosyncratic components to have unrestricted 
piecewise-stationary second-order structures, 
including changes in their autocorrelation structures (Section \ref{sec:model}).

\item[(c)] We derive a uniform convergence rate for
the PCA-based estimator of the common components under the factor model with multiple change-points
in Theorems \ref{thm:common} and \ref{thm:overestimation} (Section \ref{sec:sep}).  
A key to the derivation of the theoretical results is the introduction of 
the `capped' PCA estimator of the common components,
which controls for the possible contribution of spurious factors to individual common components
even when the number of factors is over-specified.

\item[(d)] Through the use of wavelets,
we transform the problem of change-point detection in the second-order structure 
of a high-dimensional time series, 
into the (relatively easier) problem of change-point detection in the means of high-dimensional panel data. 
More specifically, in the first stage of the proposed methodology, we
decompose the observed time series into common and idiosyncratic components,
and we compute wavelet transforms for each component (separately), 
see e.g., \cite{nason2000}, \cite{fryzlewicz2006} and \cite{van2008} 
for the use of locally stationary wavelet models for time series data. 
In this way, any change-point in the second-order structure of common or idiosyncratic components
is detectable as a change-point in the means of the wavelet-transformed series (Sections \ref{sec:wave} and \ref{sec:choice}).

\item [(e)] Each of the panels of transformed common and idiosyncratic components serves 
as an input to the second stage of our methodology,
which requires an algorithm for multiple change-point detection 
in the means of high-dimensional panel data.
For this, a number of methodologies have been investigated in the literature,
such as \cite{horvath2012}, \cite{enikeeva2014}, \cite{jirak2014}, \cite{cho2015} and \cite{wang2016}.
In Section \ref{sec:dcbs}, we adopt the Double CUSUM Binary Segmentation procedure 
proposed in \cite{cho2016}, which achieves consistency in 
estimating the total number and locations of the multiple change-points
while permitting both within-series and cross-sectional correlations.
In Section \ref{sec:consistency}, we prove that this consistency result 
carries over to the consistency in multiple change-point detection
in the common and idiosyncratic components, as the dimensions of the data, $n$ and $T$, diverge (Theorem \ref{thm:dcbs}).

\item[(f)] 
Motivated by the theoretical finding noted in (c), 
our methodology is equipped with a step that screens
the results of change-point analysis over a range of factor numbers 
employed for the estimation of common components (Section \ref{sec:factor:number}).
It enables us to circumvent the challenging problem of accurately estimating the true number of factors
in contrast to much of the existing literature, and plays the key role for the proposed methodology to achieve consistent change-point estimation.

\item [(g)] Once all the change-points are estimated, we show that 
the common components are consistently estimated via PCA on each estimated segment (Section \ref{sec:post}).
Such result is new to the literature on factor models with structural breaks.

\item[(h)] We identify and describe an interesting `spillover' effect in which change-points 
attributed to large breaks in the second-order structure of the idiosyncratic components 
get (seemingly falsely) identified as change-points in the common components (Section \ref{sec:sim}). 
We argue that this phenomenon is only too natural and expected, 
and can be ascribed to the prominent change-points 
playing the role of `common factors', regardless of whether 
they originate in the common or idiosyncratic components.
We refer to this new point of view as {\it regarding change-points as factors}.

\item[(i)] We provide an R package named {\tt factorcpt}, which implements our methodology. 
The package is available from CRAN \citep{factorcpt}.


\end{itemize}


\subsection*{An overview of the proposed methodology}

We provide an overview of the change-point detection methodology 
with an accompanying flowchart in Figure \ref{fig:flowchart}. 
For each step, we provide a reference to the relevant section in the paper.
\begin{description}[leftmargin=0pt]
\item[Input:] time series $\{x_{it}, \, i = 1, \ldots, n, \, t = 1, \ldots, T\}$ 
and a set of factor number candidates $\cR = \{\underline{r}, \underline{r}+1, \ldots, \bar{r}\}$.
\item[Iteration:] repeat Stages 1--2 for all $k \in \cR$.
\item[Stage 1: Factor analysis and wavelet transformation.] \hfill
\begin{description}
\item[PCA:]  $x_{it}$ is decomposed into the common component estimated with $k$ factors ($\wh\chi^k_{it}$) 
and the idiosyncratic component ($\wh{\eps}^k_{it}$) via PCA  (Section \ref{sec:sep}).
\item[Wavelet transformation:] $\wh\chi^k_{it}$ and $\wh{\eps}^k_{it}$ are separately transformed into panels with (almost) piecewise constant signals 
via wavelet-based transformations $g_j(\cdot)$ and $h_j(\cdot, \cdot)$ (Section \ref{sec:choice}).
\end{description}
\item[Stage 2: High-dimensional panel data segmentation.] The joint application of the Double CUSUM Binary Segmentation (Section \ref{sec:dcbs}) 
and stationary bootstrap (Section \ref{sec:bootstrap}) algorithms on the panels obtained in Stage 1,
returns the sets of change-points detected for the common ($\wh{\cB}^\chi(k)$) 
and idiosyncratic ($\wh{\cB}^\eps(k)$) components.
\item[Screening over $k \in \cR$:] 
the sets $\wh{\cB}^\chi(k)$ are screened
to obtain the final estimates $\wh{\cB}^\chi$ and $\wh{\cB}^\eps$ (Section \ref{sec:factor:number}).
\end{description}

\begin{figure}[t!]
\centering
\includegraphics[scale=1.2]{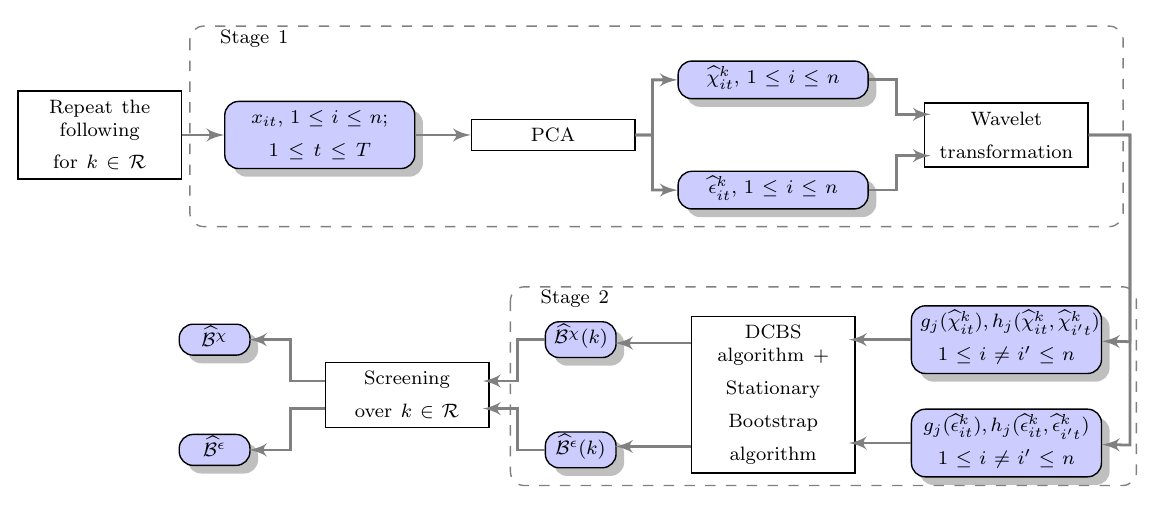}
\caption{\footnotesize Flowchart of the proposed methodology. }
\label{fig:flowchart}
\end{figure}

\subsection*{Notation}

For a given $m \times n$ matrix $\mbf B$ with $b_{i j} = [\mbf B]_{i, j}$ denoting its $(i, j)$ element, 
its spectral norm is defined as $\Vert\mbf B\Vert= \sqrt{\mu_1(\mbf B\mbf B^\top)}$,
where $\mu_k(\mbf C)$ denotes the $k$th largest eigenvalue of $\mbf C$,
and its Frobenius norm as $\Vert\mbf B\Vert_F=\sqrt{\sum_{i=1}^m\sum_{j=1}^n b_{ij}^2}$. 
For a given set $\Pi$, we denote its cardinality by $|\Pi|$.
The indicator function is denoted by $\bbI(\cdot)$.
Also, we use the notations $a \vee b = \max(a, b)$ and $a \wedge b = \min(a, b)$.
Besides, $a \sim b$ indicates that $a$ is of the order of $b$, and $a \gg b$ that $a^{-1} b \to 0$.
We denote an $n \times n$-matrix of zeros by $\mbf O_n$.

\section{Piecewise stationary factor model}
\label{sec:model}

In this section, we define a piecewise stationary factor model which provides a framework 
for developing our change-point detection methodology. 
Throughout, we assume that we observe an $n$-dimensional vector of time series 
$\mbf x_t = (x_{1t}, \ldots, x_{nt})^\top$, following a factor model and undergoing an unknown number of change-points in its second-order structure. The location of the change-points is also unknown.
Each of element $x_{it}$ can be written as:
\begin{eqnarray}
x_{it} = \chi_{it} + \eps_{it} = \bm\lambda_i^\top \mbf f_t + \eps_{it}, \qquad i=1,\ldots, n,\quad t = 1, \ldots, T, \label{ps:fm}
\end{eqnarray}
where $\bm\chi_t = (\chi_{1t}, \ldots, \chi_{nt})^\top$  
and $\bm\eps_t = (\eps_{1t}, \ldots, \eps_{nt})^\top$ denote 
the common and the idiosyncratic components of $\mbf x_t$,
with $\E(\chi_{it}) = \E(\eps_{it}) = 0$ for all $i, t$. We refer to $r$ as the number of factors in model \eqref{ps:fm}. Then the common components are driven by the $r$ factors $\mbf f_t = (f_{1t}, \ldots, f_{rt})^\top$ with $\E(\mbf f_t) = \mbf 0$,
and $\bm\lambda_i = (\lambda_{i1}, \ldots, \lambda_{ir})^\top$ denotes the $r$-dimensional vector of loadings. We denote the $n \times r$ matrix of factor loadings by 
$\bm\Lambda = [\bm\lambda_1, \ldots, \bm\lambda_n]^\top$.
The change-points in the second-order structure of $\mbf x_t$
are classified into $\cB^\chi = \{\eta^\chi_1, \ldots, \eta^\chi_{B_\chi}\}$, 
those in the common components,
and $\cB^\eps = \{\eta^\eps_1, \ldots, \eta^\eps_{B_\eps}\}$, those in the idiosyncratic components.


\begin{rem}
\label{rem:representation}
In the factor model \eqref{ps:fm}, 
both the loadings and the factor number are time-invariant.
However, it is well established in the literature
on factor models with structural breaks, that
any change in the number of factors or loadings
can be represented by a factor model with stable loadings and time-varying factors.
In doing so, $r$, the dimension of the factor space under \eqref{ps:fm},
satisfies $r \ge r_b$, where $r_b$ denotes
the minimum number of factors of each segment
$[\eta^\chi_b+1, \eta^\chi_{b+1}]$, $b = 0, \ldots, B_\chi$,
over which the common components remain stationary.
We provide a comprehensive example illustrating this point in Appendix \ref{sec:ex}.
Many tests and estimation methods for changes in the loadings rely on 
this equivalence between the two representations of time-varying factor models,
see e.g., \cite{han2014} and \cite{chen2014}.
Therefore, we work with the representation in \eqref{ps:fm}, where
all change-points in the common components are imposed on the
second-order structure of $\mbf f_t$ as detailed in
Assumption \ref{assum:one} below.
At the same time, we occasionally 
speak of changes in the loadings or the number of factors,
referring to those changes in $\mbf f_t$ that can be ascribed back to such changes.
\end{rem}

We denote by $\beta(t) = \max\{0 \le b \le B_\chi: \, \eta^\chi_b+1 \le t\}$ 
the index of the change-point in $\chi_{it}$ that is nearest to, and strictly to the left of $t$,
and by $\ubar{\eta}^\chi(t) = \max\{\eta^\chi_b: \, \eta^\chi_b+1 \le t, \ b=0, \ldots, B_\chi\}$ 
the latest change-point in $\chi_{it}$ strictly before $t$,
with the notational convention $\eta^\chi_0 = 0$ and $\eta^\chi_{B_\chi+1} = T$.
Similarly, $\ubar{\eta}^\eps(t)$ is defined with respect to $\eta^\eps_b \in \cB^\eps$, and
$\gamma(t)$ denotes the index of the change-point in $\eps_{it}$ 
that is nearest to $t$ while being strictly to its left.
Then, we impose the following conditions on $\mbf f_t$ and $\bm\eps_t$.
\begin{assum}
\label{assum:one}
\hfill
\bit
\item[(i)] There exist weakly stationary processes $\mbf f^b_t = (f^b_{1t}, \ldots, f^b_{rt})^\top$ 
associated with the intervals $[\eta^\chi_b+1, \eta^\chi_{b+1}]$ for $b = 0, \ldots, B_\chi$ such that 
$\E(f^b_{jt}) = 0$ for all $j, t$, and
\beas
\max_{1 \le j \le r} \E\Big\{f_{jt} - f^{\beta(t)}_{jt}\Big\}^2 = O\Big(\rho_f^{t - \ubar{\eta}^\chi(t)}\Big)
\eeas
for some fixed $\rho_f \in [0, 1)$.
\item[(ii)] Let $\bm\chi^b_t = (\chi^b_{1t}, \ldots, \chi^b_{nt})^\top = \bm\Lambda\mbf f^b_t$ and
denote the (auto)covariance matrix of $\bm\chi^b_t$ by 
$\wt{\bm\Gamma}^b_\chi(\tau) = \E\{\bm\chi^b_{t+\tau}(\bm\chi^b_t)^\top\} = 
\bm\Lambda\E\{\mbf f^b_{t+\tau}(\mbf f^b_t)^\top\}\bm\Lambda^\top$.
Then, there exists fixed $\bar{\tau}_\chi < \infty$ such that
for any $b = 1, \ldots, B_\chi$, we have
$\wt{\bm\Gamma}_\chi^b(\tau) - \wt{\bm\Gamma}_\chi^{b-1}(\tau) \ne \mbf O_n$ 
for some $|\tau| \le \bar{\tau}_\chi$.
\eit
\bit
\item[(iii)] There exist weakly stationary processes $\bm\eps^b_t = (\eps^b_{1t}, \ldots, \eps^b_{nt})^\top$ 
associated with the intervals $[\eta^\eps_b+1, \eta^\eps_{b+1}]$ for $b = 0, \ldots, B_\eps$ such that 
$\E(\eps^b_{it}) = 0$ for all $i, t$, and
\beas
\max_{1 \le i \le n} \E\Big\{\eps_{it} - \eps^{\gamma(t)}_{it}\Big\}^2 
= O\Big(\rho_\eps^{t - \ubar{\eta}^\eps(t)}\Big)
\eeas
for some fixed $\rho_\eps \in [0, 1)$.
\item[(iv)] Denote the (auto)covariance matrix of $\bm\eps^b_t$ by 
$\wt{\bm\Gamma}^b_\eps(\tau) = \E\{\bm\eps^b_{t+\tau}(\bm\eps^b_t)^\top\}$.
Then, there exists fixed $\bar{\tau}_\eps < \infty$ such that
for any $b = 1, \ldots, B_\eps$, we have
$\wt{\bm\Gamma}_\eps^b(\tau) - \wt{\bm\Gamma}_\eps^{b-1}(\tau) \ne \mbf O_n$ 
for some $|\tau| \le \bar{\tau}_\eps$.
\eit
\end{assum}
Assumption \ref{assum:one} (i)--(ii) indicates that for each $b = 0, \ldots, B_\chi$, 
the common component $\chi_{it}$ is `close' to a stationary process $\chi^b_{it}$ 
over the segment $[\eta^\chi_b+1, \eta^\chi_{b+1}]$
such that the effect of transition from one segment to another
diminishes at a geometric rate as $t$ moves away from $\eta^\chi_b$, and 
any $\eta^\chi_b \in \cB^\chi$ coincides with a change-point
in the autocovariance or cross-covariance matrices of $\bm\chi^{\beta(t)}_t$.
The same arguments apply to $\eps_{it}$ under Assumption \ref{assum:one} (iii)--(iv).
The treatment of such approximately piecewise stationary $\mbf f_t$ and $\bm\eps_t$
is similar to that in \cite{fryzlewicz2014b}.

Note that the literature on factor models with structural breaks have primarily focused
on the case of a single change-point in the loadings or factor number
see e.g., \cite{breitung2011}, \cite{chen2014} and \cite{han2014}.
We emphasise that to the best of our knowledge, 
the model \eqref{ps:fm} equipped with Assumption \ref{assum:one}
is the first one to offer a comprehensive framework
that allows for multiple change-points that are not confined to breaks in the loadings or 
the emergence of a new factor, 
but also includes breaks in the second-order structure of the factors and idiosyncratic components.

We now list and motivate the assumptions imposed on the piecewise stationary factor model; 
see e.g., \cite{stockwatson02JASA}, \cite{baing02}, \cite{bai2003}, \cite{FGLR09} and \cite{fan13}
for similar conditions on stationary factor models.
\begin{assum}
\label{assum:two} \hfill
\bit
\item[(i)] $T^{-1}\sum_{t=1}^T \E(\mbf f_t\mbf f_t^\top) = \mbf I_r$.
\item[(ii)] There exists some fixed $c_f, \beta_f \in (0, \infty)$
such that, for any $z > 0$ and $j = 1, \ldots, r$,
\beas
\p(|f_{jt}| > z) \le \exp\{1-(z/c_f)^{\beta_f}\}.
\eeas
\eit
\end{assum}

\begin{assum}
\label{assum:three} \hfill
\bit
\item[(i)] There exists a positive definite $r\times r$ matrix $\mbf H$ with distinct eigenvalues, such that $n^{-1}\bm\Lambda^\top\bm\Lambda \to \mbf H$ as $n\to\infty$. 
\item[(ii)] There exists $\bar{\lambda} \in (0, \infty)$ such that
$|\lambda_{ij}| < \bar{\lambda}$ for all $i, j$.
\eit
\end{assum}

\begin{assum}
\label{assum:four} \hfill
\bit
\item[(i)] There exists $C_\eps \in (0, \infty)$ such that
\begin{eqnarray*}
\max_{1 \le s \le e \le T} \bigg\{\frac{1}{e-s+1}\sum_{i, i'=1}^n\sum_{t, t' = s}^e
a_ia_{i'}\E(\eps_{it}\eps_{i't'})\bigg\} < C_\eps,
\end{eqnarray*}
for any sequence of coefficients $\{a_i\}_{i=1}^n$ satisfying $\sum_{i=1}^n a_i^2 = 1$.
\item[(ii)] $\eps_{it}$ are normally distributed.
\eit
\end{assum}

\begin{assum}
\label{assum:five} \hfill
\bit
\item[(i)] $\{\mbf f_t, 1\le t \le T\}$ and $\{\bm\eps_t, 1\le t \le T\}$ are independent.
\item[(ii)] Denoting the $\sigma$-algebra generated by $\{(\mbf f_t, \bm\eps_t), \, s \le t \le e\}$ 
by $\mc F_s^e$,
let
\beas
\alpha(k) = \max_{1 \le t \le T} \sup_{\substack{A \in \mc F_{-\infty}^t \\ B \in \mc F_{t+k}^\infty}}
|\p(A)\p(B) - \p(A \cap B)|.
\eeas
Then, there exists some fixed $c_\alpha, \beta \in (0, \infty)$ 
satisfying $3\beta_f^{-1} + \beta^{-1} > 1$,
such that, for all $k \in \Z^+$, 
we have $\alpha(k) \le \exp(-c_\alpha k^\beta)$.
\eit
\end{assum}
We adopt the normalisation given in Assumption \ref{assum:two} (i) 
for the purpose of identification; in general, factors and loadings are 
recoverable up to a linear invertible transformation only.
Similar assumptions are found in the factor model literature, see e.g., equation (2.1) of \cite{fan13}.
In order to motivate Assumptions \ref{assum:two} (i), \ref{assum:three} and \ref{assum:four} (i), 
we introduce the notations
\begin{align*}
& \bm\Gamma^{b}_f = \frac{1}{\eta^\chi_{b+1} - \eta^\chi_b}\sum_{t=\eta^\chi_b+1}^{\eta^\chi_{b+1}}\E(\mbf f_t\mbf f_t^\top),
\quad
\bm\Gamma^{b}_\eps = \frac{1}{\eta^\eps_{b+1} - \eta^\eps_b}\sum_{t=\eta^\eps_b+1}^{\eta^\eps_{b+1}}\E(\bm\eps_t\bm\eps_t^\top),
\quad
\bm\Gamma^{b}_\chi = \bm\Lambda \bm\Gamma^{b}_f  \bm\Lambda^\top,
\quad
\\
& \bm\Gamma_{\chi} = \frac{1}{T}\sum_{b=0}^{B_\chi} (\eta^\chi_{b+1}-\eta^\chi_b)\bm\Gamma^{b}_\chi, \quad
\bm\Gamma_\eps = \frac{1}{T}\sum_{b=0}^{B_\eps} (\eta^\eps_{b+1}-\eta^\eps_b)\bm\Gamma^{b}_\eps \quad \mbox{and} \quad
\bm\Gamma_x = \bm\Gamma_\chi  + \bm\Gamma_\eps.
\end{align*}
We denote the eigenvalues (in non-increasing order) of 
$\bm\Gamma^b_\chi$, $\bm\Gamma^b_\eps$, $\bm\Gamma_x$, $\bm\Gamma_\chi$ and $\bm\Gamma_\eps$ by 
$\mu_{\chi, j}^{b}$, $\mu_{\eps, j}^{b}$, $\mu_{x, j}$, $\mu_{\chi, j}$ and $\mu_{\eps, j}$, respectively.
Then, Assumptions \ref{assum:two} (i) and \ref{assum:three} imply that $\mu_{\chi, j}^{b}, \, j=1, \ldots, r_b$ 
are diverging as $n \to \infty$ and, in particular, are of order $n$ for all $b = 0, \ldots, B_\chi$. 
Assumption \ref{assum:three} implicitly rules out change-points which are due to weak changes in the loadings, 
in the sense that the magnitudes of the changes in the loadings are small, 
or only a small fraction of $\chi_{it}$ undergoes the change.
A similar requirement can be found in e.g.,
Assumption 1 of \cite{chen2014} and Assumption 10 of \cite{han2014}.
We note that \cite{bai2017} studies the consistency of a least squares estimator 
for a single change-point attributed to a possibly weak break in the loadings in the above sense.
In Section \ref{sec:sim}, we provide numerical results 
on the effect of the size of the break on our proposed methodology.

Assumption \ref{assum:four} (i) guarantees that, 
when $\mbf a = (a_1, \ldots, a_n)^\top$ is a normalised eigenvector of $\bm\Gamma_\eps^b$, then
the largest eigenvalue of $\bm\Gamma_\eps^{b}$ 
is bounded for all $b = 0, \ldots, B_\eps$ and $n$, 
that is $\mu_{\eps, 1}^{b} < C_\eps$. 
This is the same assumption as those made in \cite{chamberlain1983} and \cite{FGLR09}
and comparable to Assumption C.4 of \cite{bai2003} and 
Assumption 2.1 of \cite{fan2011b} in the stationary case.
Note that Assumption \ref{assum:four} (i) is sufficient in guaranteeing
the \textit{commonness} of $\chi_{it}$ and the \textit{idiosyncrasy} of $\eps_{it}$
according to Definitions 2.1 and 2.2 of \cite{hallinlippi13}.  Assumption \ref{assum:four} (ii) may be relaxed to allow
for $\epsilon_{it}$ of exponential-type tails, provided that the tail behaviour
carries over to the cross-sectional sums of $\epsilon_{it}$.
Note that the normality of the idiosyncratic component does not necessarily 
imply the normality of the data since the factors are allowed to be non-normal.


Assumption \ref{assum:five} is commonly found in the factor model literature 
(see e.g., Assumptions 3.1--3.2 of \cite{fan2011b}). 
In particular, the exponential-type tail conditions in Assumptions \ref{assum:two} (ii) and \ref{assum:four} (ii),
along with the mixing condition in Assumption \ref{assum:five} (ii), allow us to control the deviation of sample covariance estimates from their population counterparts,
via Bernstein-type inequality (see e.g., Theorem 1.4 of \cite{bosq1998} and Theorem 1 of \cite{merlevede2011}). 

We also assume the following on the minimum distance between two adjacent change-points.
\begin{assum}
\label{assum:six}
There exists a fixed $c_\eta \in (0, 1)$ such that 
\beas
\min\Big\{\min_{0 \le b \le B_\chi} |\eta^\chi_{b+1}-\eta^\chi_b|, 
\min_{0 \le b \le B_\eps} |\eta^\eps_{b+1}-\eta^\eps_b|\Big\} \ge c_\eta T.
\eeas
\end{assum}
Under Assumptions \ref{assum:two}--\ref{assum:six}, 
the eigenvalues of $\bm\Gamma_{\chi}$ and $\bm\Gamma_\eps$ satisfy:
\ben
\item [(C1)] there exist some fixed $\underline c_j, \bar{c}_j$ such that for $j = 1, \ldots, r$,
\beas
0 < \underline c_j <\lim_{n\to\infty}\!\!\inf \frac {\mu_{\chi, j}}{n} \le 
\lim_{n\to\infty}\!\!\sup \frac{\mu_{\chi, j}}{n} < \bar{c}_j < \infty,
\eeas
and $\bar{c}_{j+1} < \underline c_j$ for $j = 1, \ldots, r-1$;
\item [(C2)] $\mu_{\eps, 1} < C_{\eps}$, for any $n$.
\een
Moreover, by Weyl's inequality, 
we have the following asymptotic behaviour of the eigenvalues of $\bm\Gamma_x$:
\ben
\item [(C3)] the $r$ largest eigenvalues, $\mu_{x, 1}, \ldots, \mu_{x, r}$, 
are diverging linearly in $n$ as $n \to \infty$;
\item [(C4)] the $(r+1)$th largest eigenvalue, $\mu_{x, r+1}$, stays bounded for any $n$.
\een
From (C1)--(C4) above, it is clear that for identification of the common and idiosyncratic components, 
factor models need to be studied in the asymptotic limit where $n \to \infty$ and $T \to \infty$. 
This requirement, in practice, recommends the use of large cross-sections to apply PCA for factor analysis. 
In particular, we require
\begin{assum}
\label{assum:seven}
$n \to \infty$ as $T \to \infty$, with $n = O(T^\kappa)$ for some $\kappa \in (0, \infty)$.
\end{assum}
Under Assumption \ref{assum:seven}, we are able to establish 
the consistency in estimation of the common components 
as well as that of the change-points in high-dimensional settings where $n >T$,
even when the factor number $r$ is unknown.

\begin{rem}
\label{rem:beta}
The multiple change-point detection algorithm adopted in Stage 2 of our methodology
still achieves consistency when Assumption \ref{assum:six} is relaxed to allow 
$\min\{|\eta^\chi_{b+1}-\eta^\chi_b|, \, |\eta^\eps_{b+1}-\eta^\eps_b|\} \ge c_\eta T^\nu$ 
for some fixed $c > 0$ and $\nu \in (6/7, 1]$,
with $B \equiv B_T$ increasing in $T$ such that $B_T = O(\log^2 T)$.
However, under these relaxed conditions, it is no longer guaranteed that $\bm\Gamma_\chi$ and,
consequently, $\bm\Gamma_x$ have $r$ diverging eigenvalues.
Therefore, in this paper we limit the scope of our theoretical results to the more restricted setting 
of Assumption \ref{assum:six}.
\end{rem}

\section{Factor analysis and wavelet transformation}
\label{sec:stage:one}

\subsection{Estimation of the common and idiosyncratic components}
\label{sec:sep}

Decomposing $x_{it}$ into common and idiosyncratic components
is an essential step for the separate treatment of the change-points in $\chi_{it}$ and $\eps_{it}$,
such that we can identify the origins of detected change-points. 
Therefore, in this section, we establish the asymptotic bounds on the estimation error 
of the common and idiosyncratic components, when using PCA
under the piecewise stationary factor model in \eqref{ps:fm}. 

Let $\wh{\mbf w}_{x, j}$ denote the $n$-dimensional normalised eigenvector corresponding 
to the $j$th largest eigenvalue of the sample covariance matrix, $\wh{\bm\Gamma}_{x}$, 
with its entries $\wh{w}_{x, ij}$, $i=1,\ldots, n$. 
When the number of factors $r$ is known,
the PCA estimator of $\chi_{it}$ is defined as
$\wh{\chi}_{it} = \sum_{j=1}^r \wh{w}_{x,ij} \wh{\mbf w}^{\top}_{x,j} \mbf x_t$,
for which the following theorem holds.
\begin{thm}
\label{thm:common}
{\it
Under Assumptions \ref{assum:one}--\ref{assum:seven}, 
the PCA estimator of $\chi_{it}$ with $r$ known satisfies
\beas
\max_{1 \le i \le n} \max_{1 \le t \le T}
\vert \wh{\chi}_{it}-\chi_{it}\vert = O_p\l\{\Big(\sqrt{\frac{\log\,n}{T}} \vee \frac{1}{\sqrt n}\Big)\log^\theta T\r\},
\eeas
for some fixed $\theta \ge 1 + (\beta_f^{-1} \vee 1/2)$.}
\end{thm}
Proofs of Theorem \ref{thm:common} and all other theoretical results are provided in the Appendix. 

Despite the presence of multiple change-points, allowed both in the variance and autocorrelations of $\mbf f_t$, 
the rate of convergence for $\wh{\chi}_{it}$ is almost as fast as the one derived for the stationary case, 
e.g., Theorem 3 of \cite{bai2003}, and 
Theorem 1 of \cite{pelger2016} and Theorem 5 of \cite{ait2016} 
in the context of factor modelling high-frequency data. 
We highlight that Theorem \ref{thm:common} derives 
a uniform bound on the estimation error over $i = 1, \ldots, n$ and $t = 1, \ldots, T$;
similar results are found in Theorem 4 of \cite{fan13}.


However, the true number of factors $r$ is typically unknown and, 
since the seminal paper by \cite{baing02}, 
estimation of the number of factors has been one of the most researched problems in the factor model literature;
see also \cite{ABC10}, \cite{onatski10} and \cite{ahnhorenstein13}.
Although, \cite{han2014} (in their Proposition 1) and \cite{chen2014} (in their Proposition 2)
showed that the information criteria of \cite{baing02}
achieve asymptotic consistency in estimating $r$ in the presence of a single break in the loadings,
it has been observed that such an approach tends to exhibit poor finite sample performance
when the idiosyncratic components are both serially and cross-sectionally correlated,
or when $\Var(\eps_{it})$ is large compared to $\Var(\chi_{it})$.
Also, it was noted that the factor number estimates heavily 
depend on the relative magnitude of $n$ and $T$ and the choice of penalty terms
as demonstrated in the numerical studies of \cite{baing02} and \cite{ahnhorenstein13}.
While the majority of change-point methods for factor models 
rely on the consistent estimation of the number of factors
\citep{breitung2011, han2014, chen2014, corradi2014},
our empirical study on simulated data, available in the supplementary document, indicates that 
this dependence on a single estimate may lead to failure in change-point detection.

To remedy this, we propose to consider a range of factor number candidates in our change-point analysis, 
in particular, allowing for the over-specification of the number of factors. 
In high dimensions, the estimated eigenvectors $\wh{\mbf w}_{x, j}$ for $j > r$ 
are in general not consistent, as implied by Theorem 2 of \cite{yu15}.
Thus, estimates of common components based on more than $r$ principal components 
may be subject to a non-negligible overestimation error.
In order to control the effect of over-specifying the number of factors,
we propose a modified PCA estimator of $\chi_{it}$ defined as
$\wh{\chi}^k_{it} = \sum_{j=1}^k \wt{w}_{x,ij} \wt{\mbf w}^{\top}_{x,j} \mbf x_t$,
where each element of $\wt{\mbf w}_{x,j}$ is `capped' according to the rule
\beas
\wt{w}_{x, ij} = \wh{w}_{x,ij} \, \bbI\Big(|\wh{w}_{x,ij}| \le \frac{c_w}{\sqrt n}\Big)
+ \mbox{sign}(\wh{w}_{x,ij}) \cdot \frac{c_w}{\sqrt n} \, \bbI\Big(|\wh{w}_{x,ij}| > \frac{c_w}{\sqrt n}\Big),
\eeas
for some fixed $c_w>0$. 
The estimated loadings are then given by $\wh\lambda_{ij} = \sqrt{n}\wt{w}_{x, ij}$ and the factors by
$\wh{f}_{jt} = n^{-1/2}\wt{\mbf w}^{\top}_{x,j} \mbf x_t$.
A practical way of choosing $c_w$ is discussed in Section \ref{sec:tuningp}.

Note that, thanks to the result in (C3) and Lemmas \ref{lem:evals} and \ref{lem:evecs} in the Appendix, 
we have $\max_{1 \le j \le r} \max_{1 \le i \le n} |\wh{w}_{x,ij}| = O_p(1/\sqrt{n})$.
In other words, asymptotically, the capping does not alter the contribution
from the $r$ leading eigenvectors of $\wh{\bm\Gamma}_x$ to $\wh{\chi}^k_{it}$,
even when the capping is applied without the knowledge of $r$.
On the other hand, by means of this capping, 
we control the contribution from the `spurious' eigenvectors when $k > r$,
which allows us to establish the following bound on the partial sums of the estimation errors
even when the factor number is over-specified.

\begin{thm}
\label{thm:overestimation}
{\it Suppose that Assumptions \ref{assum:one}--\ref{assum:seven} hold,
and let $\theta$ be defined as in Theorem \ref{thm:common}.
Then, the capped estimator $\wh\chi^k_{it}$ satisfies the following:
\bit
\item[(i)] when $k = r$,
\bea
\max_{1 \le i \le n} \, \max_{1 \le s < e \le T}\frac 1 {\sqrt {e-s+1}} 
\Big\vert \sum_{t = s}^e (\wh{\chi}^r_{it}-\chi_{it}) \Big\vert 
= O_p\l\{\Big(\sqrt{\frac{\log\,n}{T}} \vee \frac{1}{\sqrt n}\Big)\log^\theta T\r\};
\label{eq:thm:overestimation:one}
\eea
\item[(ii)] when $k > r$,
\bea
\max_{1 \le i \le n} \, \max_{1 \le s < e \le T}\frac 1 {\sqrt {e-s+1}} 
\Big\vert \sum_{t = s}^e (\wh{\chi}^k_{it}-\chi_{it}) \Big\vert 
= O_p(\log^\theta T);
\label{eq:thm:overestimation:two}
\eea
\eit
with $\theta$ given in Theorem \ref{thm:common}.}
\end{thm}

The bound in \eqref{eq:thm:overestimation:one} concerns 
the case in which we correctly specify the number of factors and is in agreement with Theorem \ref{thm:common}. 
Turning to when the factor number is over-specified, \cite{FHLR00} (in their Corollary 2) 
and \cite{onatski15} (in his Proposition 1) reported similar results for the stationary case.
However, the uniform consistency of $\wh\chi^k_{it}, \, k  \ge r$, in the presence of multiple change-points, 
has not been spelled out before in the factor model literature to the best of our knowledge;
we achieve this via the proposed capping. 
On the other hand, it is possible to show that with $k < r$, 
the estimation error in $\wh{\chi}^k_{it}$ is non-negligible. 
Lastly, thanks to Lemma \ref{lem:x:n} in Appendix,
it is straightforward to show that analogous bounds hold for the idiosyncratic component 
$\wh{\eps}^k_{it} = x_{it} - \wh{\chi}^k_{it}$.

Although the over-specification of $k > r$ brings the bound in \eqref{eq:thm:overestimation:one}
to increase by $\sqrt{n} \wedge \sqrt{T}$,
we can still guarantee that all change-points in $\chi_{it}$ ($\eps_{it}$) are detectable from
$\wh\chi^k_{it}$ ($\wh\eps^k_{it}$) provided that $k \ge r$,
as shown in Proposition \ref{prop:chi:additive} and Theorem \ref{thm:dcbs} below.
In what follows, we continue describing our methodology
by supposing that $k \ge r$ is given, and we refer to Section \ref{sec:factor:number} 
for the complete description of the proposed `screening' procedure that considers 
a range of values for $k$. 

\subsection{Wavelet periodograms and cross-periodograms}
\label{sec:wave}



As the first stage of the proposed methodology, we construct
a wavelet-based transformation (WT) of the estimated common and idiosyncratic components 
$\wh\chi^k_{it}$ and $\wh\eps^k_{it}$ from Section \ref{sec:sep}, 
which serves as an input to the algorithm for high-dimensional change-point analysis in Stage 2. 
In order to motivate the WT, which will be formally introduced in Section \ref{sec:choice},
we limit our discussion in this section to the (unobservable) common component $\chi_{it}$;
the same arguments hold verbatim for the idiosyncratic one. 
In practice, the WT is performed on the estimated common component, $\wh{\chi}^k_{it}$, 
and the effect of considering estimated quantities rather than the true ones is studied in Section \ref{sec:choice}. 
 
\cite{nason2000} have proposed the use of wavelets as building blocks in nonstationary time series 
analogous to Fourier exponentials in the classical Cram\'{e}r representation for stationary processes.
The simplest example of a wavelet system, Haar wavelets, are defined as
\beas
\psi_{j, l} = 2^{j/2}\bbI(0 \le l \le 2^{-j-1}-1) - 2^{j/2}\bbI(2^{-j-1} \le l \le 2^{-j}-1),
\eeas
with $j \in \{-1, -2, \ldots\}$ denoting the wavelet scale, and $l \in \Z$ denoting the location.
Small negative values of the scale parameter $j$ denote fine scales 
where the wavelet vectors are more localised and oscillatory, 
while large negative values denote coarser scales with longer, less oscillatory wavelet vectors. 

Recall the notation $\beta(t) = \max\{0 \le b \le B_\chi: \, \eta^\chi_b+1 \le t\}$. 
Wavelet coefficients of $\chi^{\beta(t)}_{it}$ (introduced in Assumption \ref{assum:one}) are defined as 
$d_{j, it} = \sum_{l=0}^{\cL_j-1} \chi^{\beta(t)}_{i, t-l}\psi_{j, l}$,
with respect to $\bm\psi_j = (\psi_{j, 0}, \ldots, \psi_{j, \cL_j-1})^\top$, a vector of discrete wavelets at scale $j$.
Note that the support of $\bm\psi_j$ is of length $\cL_j = M2^{-j}$ for a fixed $M> 0$ 
(which depends on the choice of wavelet family),
so that we have access to wavelet coefficients from at most $\lfloor \log_2\,T \rfloor$ scales 
for a time series of length $T$.
In other words, wavelet coefficients are obtained by filtering $\chi^{\beta(t)}_{it}$ 
with respect to wavelet vectors of finite lengths.

Wavelet periodogram and cross-periodogram sequences of $\chi^{\beta(t)}_{it}$ are defined as
$I_{j, it} = |d_{j, it}|^2$ and $I_{j, ii't} = d_{i, it}d_{j, i't}$.
It has been shown that the expectations of these sequences 
have a one-to-one correspondence with the second-order structure of the input time series, 
see \cite{cho2012} for the case of univariate time series 
and \cite{cho2015} for the high-dimensional case.
To illustrate, suppose that $t- \cL_j +1 \ge \ubar{\eta}^\chi(t) + 1$.
Then,
\bea
\E\vert d_{j, it} \vert^2 
&=& \sum_{l, l' = 0}^{\cL_j-1} \E(\chi^{\beta(t-l)}_{i, t-l}\chi^{\beta(t-l')}_{i, t-l'})\psi_{j, l}\psi_{j, l'}
= \sum_{l = 0}^{\cL_j-1}\sum_{|\tau| < \cL_j} \E(\chi^{\beta(t-l)}_{i, t-l}\chi^{\beta(t-l-\tau)}_{i, t-l-\tau})\psi_{j, l}\psi_{j, l+\tau} \nn
\\
&=& \sum_{|\tau| < \cL_j} [\wt{\bm\Gamma}^{\beta(t)}_\chi(\tau)]_{i, i} \Psi_j(\tau),
\label{eq:chi:wp}
\eea
where $\Psi_j(\tau) = \sum_{l \in \Z}\psi_{j, l}\psi_{j, l+\tau}$ (with $\psi_{j, l} = 0$ unless $0 \le l \le \cL_j-1$)
denotes the autocorrelation wavelets \citep{nason2000}
and $\wt{\bm\Gamma}^b_\chi(\tau) = \E\{\bm\chi^b_{t+\tau}(\bm\chi^b_t)^\top\}$.
That is, provided that $t$ is sufficiently distanced (by $\cL_j$) from any change-points to its left, 
$\E\vert d_{j, it} \vert^2$ is a wavelet transformation of the (auto)covariances
$\E(\chi^{\beta(t)}_{it}\chi^{\beta(t)}_{i, t+\tau}), \, |\tau| \le \cL_j-1$.
Similar arguments hold between wavelet cross-periodograms and cross-covariances of $\chi^{\beta(t)}_{it}$. 
Following \cite{cho2015}, we conclude that under Assumption \ref{assum:one} (ii),
any jump in the autocovariance and cross-covariance structures of $\bm\chi^{\beta(t)}_t$
is translated to a jump in the means of its wavelet periodogram and cross-periodogram sequences 
at some wavelet scale $j$, in the following sense:
$\{\E\vert d_{j, it} \vert^2, \,  1 \le i \le n, \, \E(d_{j, it}d_{j, i't}), \, 1 \le i < i' \le n\}$ 
are `almost' piecewise constant with their change-points coinciding with those in $\cB^\chi$,
apart from intervals of length $\cL_j$ around the change-points.

It is reasonable to limit our consideration to wavelets at the first $\jts$ 
finest scales $j = -1, \ldots, -\jts$ (with $\jts \le \lfloor \log_2\,T \rfloor$),
in order to control the possible bias in change-point estimation that 
arises from the transition intervals of length $\cL_j$.
On the other hand, due to the compactness of the support of $\bm\psi_j$,
a change in $\wt{\bm\Gamma}_\chi^{\beta(t)}(\tau)$ that appears 
only at some large lags ($|\tau| \ge \cL_{-\jts} = M2^{\jts}$), 
is not detectable as a change-point in the wavelet periodogram and cross-periodogram sequences
at the few finest scales ($j \ge -\jts$), see \eqref{eq:chi:wp}.
To strike a balance between the above quantities related to the choice of $\jts$, 
we set $\jts = \lfloor C\log_2\log^\upsilon T \rfloor$ for some $\upsilon \in (0, 1]$ and some fixed $C > 0$. 
We refer to Section \ref{sec:tuningp} for the choice of $C$. Then, 
\begin{itemize}
\item[(a)] any change in $\bm\Gamma_\chi^{\beta(t)}(\tau)$ 
that occurs at (at least) one out of an increasing number of lags ($|\tau| \le M\log^\upsilon T$),
is registered as a change-point in the expectations of the wavelet (cross-)periodograms of $\chi^{\beta(t)}_{it}$;
\item[(b)] the possible bias in the registered locations of the change-points 
is controlled to be at most $O(\log^\upsilon T)$.
\end{itemize}

\subsection{Wavelet-based transformation for change-point analysis}
\label{sec:choice}

In this section, we propose the WT of estimated common and idiosyncratic components,
and show that the change-points in the complex (autocovariance and cross-covariance) 
structure of (unobservable) $\chi_{it}$ and $\epsilon_{it}$, are made detectable
as the change-points in the relatively simple structure (means) of the panel of the 
wavelet transformed $\wh\chi^k_{it}$ and $\wh\epsilon^k_{it}$. 
This panel then serves as an input to Stage 2 of our methodology, as described in Section \ref{sec:stage:two}.
As in Section \ref{sec:wave}, we limit the discussion of the WT and its properties 
when applied to the change-point analysis of the common components;
the same arguments are applicable to that of the idiosyncratic components.

Let $\wh{d}_{j, it} = \sum_{l=0}^{\cL_j-1} \wh\chi^k_{i, t-l}\psi_{j, l}$ denote 
the wavelet coefficients of $\wh\chi^k_{it}$. 
Then, for each $j = -1, -2, \ldots, -\jts$, we propose the following transformation which takes $\wh\chi^k_{it}$
and produces a panel of $n(n+1)/2$-dimensional sequences with elements:
\begin{eqnarray*}
g_j(\wh\chi^k_{it}) &\equiv& g_j(\wh\chi^k_{it}, \ldots, \wh\chi^k_{i, t-\cL_j+1}) 
= \vert \wh d_{j, it} \vert = \Big\vert \sum_{l=0}^{\cL_j-1} \wh\chi^k_{i, t-l}\psi_{j, l} \Big\vert, \quad 1 \le i \le n,
\\
h_j(\wh\chi^k_{it}, \wh\chi^k_{i't}) &\equiv& h_j(\wh\chi^k_{it}, \ldots, \wh\chi^k_{i, t-\cL_j+1}, \wh\chi^k_{i't}, \ldots, \wh\chi^k_{i', t-\cL_j+1}) 
= \vert \wh d_{j, it} + s_{ii'}\wh d_{j, i't} \vert 
\\
&=& \Big\vert \sum_{l=0}^{\cL_j-1} \wh\chi^k_{i, t-l}\psi_{j, l}+s_{ii'}\sum_{l'=0}^{\cL_j-1} \wh\chi^k_{i', t-l'}\psi_{j, l'} \Big\vert, \quad 1 \le i < i' \le n,
\end{eqnarray*}
where $s_{ii'} \in \{-1, 1\}$.
For example, with Haar wavelets at scale $j = -1$,
\beas
g_{-1}(\wh\chi^k_{it}) = \frac{1}{\sqrt 2}|\wh\chi^k_{it} - \wh\chi^k_{i, t-1}| \quad \mbox{and} \quad
h_{-1}(\wh\chi^k_{it}, \wh\chi^k_{i't}) = \frac{1}{\sqrt 2}|\wh\chi^k_{it} - \wh\chi^k_{i, t-1} + s_{ii'}(\wh\chi^k_{i't} - \wh\chi^k_{i', t-1})|.
\eeas
\begin{rem} Notice that $g_j(\wh\chi^k_{it})$ is simply the squared root of 
the wavelet periodogram of $\wh\chi^k_{it}$ at scale $j$. 
Arguments supporting the choice of $h_j$ in place of wavelet cross-periodogram sequences 
can be found in Section 3.1.2 of \cite{cho2015},
where it is guaranteed that any change detectable from $\wh d_{j, it} \cdot \wh d_{j, i't}$
can also be detected from $h_j(\wh\chi^k_{it}, \wh\chi^k_{i't})$ with either of $s_{ii'} \in \{1, -1\}$.
As per the recommendation made therein, 
we select $s_{ii'} = -\mbox{sign}\{\wh{\mbox{cor}}(\wh\chi^k_{it}, \wh\chi^k_{i't})\}$,
where $\wh{\mbox{cor}}(\wh\chi^k_{it}, \wh\chi^k_{i't})$ denotes the sample correlation 
between $\wh\chi^k_{it}$ and $\wh\chi^k_{i't}$ over $t = 1, \ldots, T$.
This is in an empirical effort to select $s_{ii'}$ that better brings out any change in 
$\E\{h_j(\chi_{it}, \chi_{i't})\}$ for the given data.
\end{rem}

As with wavelet periodograms and cross-periodograms discussed in Section \ref{sec:wave}, 
the transformed series $g_j(\wh\chi^k_{it})$ and $h_j(\wh\chi^k_{it}, \wh\chi^k_{i't})$ 
contain the change-points in the second-order structure of $\chi^{\beta(t)}_{it}$ as change-points in their 'signals'. 
This is formalised in the following proposition.

\begin{prop}
\label{prop:chi:additive}
{\it Suppose that all the conditions in Theorem \ref{thm:overestimation} hold. 
For some fixed $k \ge r$ and $\jts = \lfloor C\log_2\log^\upsilon T\rfloor$, 
consider the $N = \jts\, n(n+1)/2$-dimensional panel
\bea
\l\{g_j(\wh\chi^k_{it}), \, 1 \le i \le n, \, h_j(\wh\chi^k_{it}, \wh\chi^k_{i't}), \, 1 \le i < i' \le n; \, -\jts \le j \le -1; \, 1 \le t \le T\r\} \label{eq:chi:panel}
\eea
and denote as $y_{\ell t}$ a generic element of \eqref{eq:chi:panel}. Then, we have the following decomposition:
\begin{eqnarray}
\label{eq:add}
y_{\ell t} = z_{\ell t} + \vep_{\ell t}, \quad \ell = 1, \ldots, N,  \quad  t=1, \ldots, T.
\end{eqnarray}
\bit
\item[(i)] $z_{\ell t}$ are piecewise constant as the corresponding elements of
\bea
\l\{\E\{g_j(\chi^{\beta(t)}_{it})\}, \, 1 \le i \le n, \, 
\E\{h_j(\chi^{\beta(t)}_{it}, \chi^{\beta(t)}_{i't})\}, \, 1 \le i < i' \le n; \, -\jts \le j \le -1; \, 1 \le t \le T\r\}.
\label{eq:add:z}
\eea
That is, all change-points in $z_{\ell t}$ belong to $\cB^\chi = \{\eta^\chi_1, \ldots, \eta^\chi_{B_\chi}\}$ and
for each $b \in \{1, \ldots, B_\chi\}$, there exists at least a single index $\ell \in \{1, \ldots, N\}$ for which 
$|z_{\ell, \eta^\chi_b+1} - z_{\ell \eta^\chi_b}| \ne 0$.
\item[(ii)]
$\max_{1 \le \ell \le N} \max_{1 \le s < e \le T} (e-s+1)^{-1/2} \;
\vert \sum_{t=s}^e \vep_{\ell t} \vert = O_p(\log^{\theta+\upsilon}T)$.
\eit}
\end{prop}

Therefore, the panel data in (\ref{eq:chi:panel}) contains all change-points in
the second-order structure of $\bm\chi^{\beta(t)}_t$ as the change-points in its 
piecewise constant signals represented by $z_{\ell t}$.

Let us denote by $\wt{y}_{\ell t}$ an element of the panel obtained by transforming 
$\chi_{it}$ in place of $\wh\chi^k_{it}$ in \eqref{eq:chi:panel}. 
Then, the proof of Proposition \ref{prop:chi:additive} is based on the following decomposition
\beas
\vep_{\ell t} = y_{\ell t} - z_{\ell t} 
= \{\E(\wt{y}_{\ell t})-z_{\ell t}\} + \{\wt{y}_{\ell t}-\E(\wt{y}_{\ell t})\} + (y_{\ell t} -\wt{y}_{\ell t})
= I + II + III.
\eeas
Then, $I$ accounts for the discrepancy between $\chi_{it}$ and $\chi^{\beta(t)}_{it}$ and is controlled by 
Assumption \ref{assum:one} (i),
and $II$ follows from the weak dependence structure and the tail behaviour of $\mbf f_t$ assumed in
Assumptions \ref{assum:two} (ii) and \ref{assum:five}.
Term $III$ arises from the estimation error in $\wh{\chi}_{it}^k$ which can be bounded 
as shown in Theorem \ref{thm:overestimation},
which further motivates the WT of $\wh\chi^k_{it}$ via $g_j$ and $h_j$
rather than using its wavelet (cross-)periodograms.

To conclude, note that the WT of $\wh\eps^k_{it}$ are collected into the $N$-dimensional panel
\bea
\l\{g_j(\wh\eps^k_{it}), \, 1 \le i \le n, \, 
h_j(\wh\eps^k_{it}, \wh\eps^k_{i't}), \, 1 \le i < i' \le n; \, -\jts \le j \le -1; \, 1 \le t \le T\r\}
\label{eq:eps:panel0}
\eea
for change-point analysis, where the elements of \eqref{eq:eps:panel0} are also decomposed into 
piecewise constant signals 
\bea
\l\{\E\{g_j(\eps^{\gamma(t)}_{it})\}, \, 1 \le i \le n, \, 
\E\{h_j(\eps^{\gamma(t)}_{it}, \eps^{\gamma(t)}_{i't})\}, \, 1 \le i < i' \le n; \, -\jts \le j \le -1; \, 1 \le t \le T\r\},
\label{eq:eps:z}
\eea
and error terms of bounded partial sums; 
see Appendix \ref{sec:when:y:e} where we present the result analogous to Proposition \ref{prop:chi:additive}
for the idiosyncratic components.

\section{High-dimensional panel data segmentation}
\label{sec:stage:two}

\subsection{Double CUSUM Binary Segmentation}
\label{sec:dcbs}

Cumulative sum (CUSUM) statistics have been widely adopted for change-point detection 
in both univariate and multivariate data.
In order to detect change-points in the $N$-dimensional additive panel data in \eqref{eq:chi:panel}, 
we compute $N$ univariate CUSUM series and aggregate the high-dimensional CUSUM series via
the Double CUSUM statistic proposed by \cite{cho2016},
which achieves this through point-wise, data-driven partitioning of the panel data 
When dealing with multiple change-point detection, 
the Double CUSUM statistic is used jointly with binary segmentation in an algorithm,
which we refer to as the Double CUSUM Binary Segmentation (DCBS) algorithm.
The DCBS algorithm guarantees the consistency in multiple change-point detection in high-dimensional settings,
as shown in Theorem \ref{thm:dcbs},
while allowing for the cross-sectional size of change to decrease with increasing sample size 
(see Assumption \ref{assum:b:two} below).
Further, it admits both serial- and cross-correlations in the data,
which is a case highly relevant for the time series factor model considered in this paper.

Consider the $N$-dimensional input panel $\{y_{\ell t}, \, \ell=1,\ldots, N; \, t=1,\ldots, T\}$,
computed from $\wh\chi^k_{it}$ or $\wh\eps^k_{it}$
for the change-point analysis in either the common or the idiosyncratic components via WT;
see \eqref{eq:chi:panel} and \eqref{eq:eps:panel0} for their definitions.
We define the CUSUM series of $y_{\ell t}$ over a generic segment $[s, e]$ for some $1 \le s < e \le T$, as
\beas
\cY^\ell_{s, b, e} = \frac{1}{\sigma_\ell}\sqrt{\frac{(b-s+1)(e-b)}{e-s+1}}
\Big(\frac{1}{b-s+1}\sum_{t=s}^b y_{\ell t} - \frac{1}{e-b}\sum_{t=b+1}^e y_{\ell t} \Big),
\eeas
for $b=s, \ldots, e-1$, where $\sigma_\ell$ denotes a scaling constant for treating 
all rows of the panel data $y_{\ell t}, \, 1 \le \ell \le N$ on equal footing; 
see Remark \ref{rem:sigma} below for its choice.
We note that if $\vep_{\ell t}$ in \eqref{eq:add} were i.i.d. Gaussian random variables,
the maximum likelihood estimator of the change-point location for $\{y_{\ell t}, s\le t \le e\}$
would coincide with $\arg\max_{b \in [s, e)}|\cY^\ell_{s, b, e}|$.

Proposed in \cite{cho2016}, 
the Double CUSUM (DC) operator aggregates the $N$ series of CUSUM statistics from $y_{\ell t}$
and returns a two-dimensional array of DC statistics:
\beas
\cD_{s, b, e}(m) =
\sqrt{\frac{m(2N-m)}{2N}}
\Big(\frac{1}{m} \sum_{\ell=1}^m |\cY^{(\ell)}_{s, b, e}| -
\frac{1}{2N-m} \sum_{\ell=m+1}^N |\cY^{(\ell)}_{s, b, e}|\Big),
\eeas
for $b = s, \ldots, e-1$ and $m = 1, \ldots, N$,
where $|\cY^{(\ell)}_{s, b, e}|$ denotes the CUSUM statistics at $b$ ordered according to their moduli, i.e., 
$|\cY^{(1)}_{s, b, e}| \ge \ldots \ge |\cY^{(N)}_{s, b, e}|$.
Notice that $\cD_{s, b, e}(m)$ takes the contrast between
the $m$ largest CUSUM values $|\cY^{(\ell)}_{s, b, e}|, \, 1 \le \ell \le m$ and the rest at each $b$,
and thus partitions the coordinates into the $m$ that are the most likely to 
contain a change-point and those which are not in a point-wise manner. 
Then, the test statistic is derived by maximising the two-dimensional array of DC statistics 
over both time and cross-sectional indices, as
\bea
\label{dc:test:stat}
\cT_{s, e} = \max_{b \in [s, e)}\max_{1 \le m \le N} \cD_{s, b, e}(m),
\eea
which is compared against a threshold $\pi_{N, T}$ for determining 
the presence of a change-point over the interval $[s, e]$.
If $\cT_{s, e} > \pi_{N, T}$, the location of the change-point is identified as 
\beas
\heta = \arg\max_{b\in[s, e)}\max_{1 \le m \le N} \cD_{s, b, e}(m).
\eeas

\begin{rem}[Choice of $\sigma_\ell$.] 
\label{rem:sigma}
\cite{cho2016} assumes second-order stationarity on the error term $\vep_{\ell t}$ in (\ref{eq:add}),
which enables the use of its long-run variance estimator as the scaling term $\sigma_\ell$.
However, it is not trivial to define a similar quantity in the problem considered here,
particularly due to the possible nonstationarities in $\vep_{\ell t}$.
Following \cite{cho2015}, in order for the CUSUM series computed on $y_{\ell t}$
not to depend on the level of $\E(y_{\ell t}^2)$, 
we also adopt the choice of $\sigma_\ell = \sqrt{T^{-1}\sum_{t=1}^T y_{\ell t}^2}$.
Note that the asymptotic consistency of the DCBS algorithm in Theorem \ref{thm:dcbs} below
does not depend on the choice of $\sigma_\ell$,
provided that it is bounded away from zero and from the above for all $\ell = 1, \ldots, N$ 
with probability tending to one (see Assumption (A6) of \cite{cho2016}). 
By adopting arguments similar to those in the proof of Lemma 6 in \cite{cho2012}, 
it can be shown that our choice of $\sigma_\ell$ satisfies these properties.
\end{rem}

We now formulate the DCBS algorithm which is equipped with the threshold $\pi_{N, T}$.
The index $u$ is used to denote the level (indicating the progression of the segmentation procedure)
and $v$ to denote the location of the node at each level.

\begin{description}
\item[The Double CUSUM Binary Segmentation (DCBS) algorithm]
\item[Step 0:]
Set $(u, v)=(1, 1)$, $s_{u, v}=1$, $e_{u, v}=T$ and $\wh\cB = \emptyset$.

\item[Step 1:] At the current level $u$, repeat the following for all $v$.
\begin{description}
\item[Step 1.1:] Letting $s=s_{u, v}$ and $e=e_{u, v}$,
obtain the series of CUSUMs $\cY^\ell_{s, b, e}$ for $b \in [s, e)$ and $\ell = 1, \ldots, N$,
on which $\cD_{s, b, e}(m)$ is computed over all $b$ and $m$.

\item[Step 1.2:] Obtain the test statistic $\cT_{s, e} = \max_{b \in [s, e)}\max_{1 \le m \le N} \cD_{s, b, e}(m)$.

\item[Step 1.3:] If $\cT_{s, e} \le \pi_{N, T} $, quit searching for change-points on the interval $[s, e]$.
On the other hand, if $\cT_{s, e} > \pi_{N, T} $, 
locate $\heta = \arg\max_{b\in[s, e)}\max_{1 \le m \le N} \cD_{s, b, e}(m)$,
add it to the set of estimated change-points $\wh\cB$,
and proceed to Step 1.4.

\item[Step 1.4:] Divide the interval $[s_{u, v}, e_{u, v}]$ into two sub-intervals
$[s_{u+1, 2v-1}, e_{u+1, 2v-1}]$ and $[s_{u+1, 2v}, e_{u+1, 2v}]$,
where $s_{u+1, 2v-1} = s_{u, v}$, $e_{u+1, 2v-1} = \heta$, 
$s_{u+1, 2v} = \heta+1$ and $e_{u+1, 2v} = e_{u, v}$.
\end{description}

\item[Step 2:]
Once $[s_{u, v}, e_{u, v}]$ for all $v$ are examined at level $u$,
set $u \leftarrow u+1$ and go to Step 1.
\end{description}

Step 1.3 provides a stopping rule to the DCBS algorithm, 
by which the search for further change-points is terminated 
once $\cT_{s, e} \le \pi_{N, T}$ on every segment defined 
by two adjacent estimated change-points in $\wh{\cB}$.
Depending on the choice of the input panel data $y_{\ell t}$ as in \eqref{eq:chi:panel} or \eqref{eq:eps:panel0},
the DCBS algorithm returns $\wh{\cB}^\chi(k)$ or $\wh{\cB}^\eps(k)$,
the sets of change-points detected from $\wh{\chi}^k_{it}$ and $\wh{\eps}^k_{it}$, respectively. 

\subsection{Consistency in multiple change-point detection}
\label{sec:consistency}

In this section, we show the consistency of change-points estimated 
for the common and idiosyncratic components by the DCBS algorithm, 
in terms of their total number and locations.  
Consider the panel $\{y_{\ell t}, \, \ell=1,\ldots, N; \, t=1,\ldots, T\}$ that represents 
the WTs of $\wh\chi^k_{it}$ and $\wh \eps^k_{it}$ as 
in \eqref{eq:chi:panel} and \eqref{eq:eps:panel0}, respectively.
In either case, let $\eta_b, \, b = 1, \ldots, B$ denote
the change-points in the piecewise constant signals $z_{\ell t}$ underlying $y_{\ell t}$ 
(see \eqref{eq:add:z} and \eqref{eq:eps:z} for the precise definitions of $z_{\ell t}$), 
so that we have either $\eta_b = \eta^\chi_b$ with $B = B_\chi$, 
or $\eta_b = \eta^\eps_b$ with $B = B_\eps$. 
We impose the following assumptions on the signals $z_{\ell t}$ and the change-point structure therein.

\begin{assum}
\label{assum:b:one}
There exists a fixed constant $\bar{g}>0$ such that 
$\max_{1 \le \ell \le N}\max_{1 \le t \le T} |z_{\ell t}| \le \bar{g}$.
\end{assum}

\begin{assum}
\label{assum:b:two}
At each $\eta_b$, define 
$\Pi_b = \{1 \le \ell \le N: \, \delta_{\ell, b} = |z_{\ell, \eta_b+1} - z_{\ell \eta_b}| > 0\}$,
the index set of those $z_{\ell t}$ that undergo a break at $t = \eta_b$,
and denote its cardinality by $m_b = |\Pi_b|$.
Further, let $\wt{\delta}_b = m_b^{-1}\sum_{\ell \in \Pi_b} |\delta_{\ell, b}|$,
the average size of jumps in $z_{\ell t}$ at $t = \eta_b$.
Then $\Delta_{N, T} = \min_{1 \le b \le B} \sqrt{m_b}\wt{\delta}_b$ satisfies
$(\sqrt{N}\log^{\theta+\upsilon}T)^{-1} \Delta_{N, T} T^{1/4} \to \infty$ as $T \to \infty$. 
\end{assum}

Assumption \ref{assum:b:one} requires the expectations of the WTs of $\chi^b_{it}$ and $\eps^b_{it}$
defined in Assumption \ref{assum:one} to be bounded, which in fact holds trivially as
$\sum_{l=0}^{\mc L_j-1} \psi_{j, l}^2 = 1$ for all $j$.

Assumption \ref{assum:b:two} imposes a condition on the minimal cross-sectional size of the changes in $z_{\ell t}$,
represented by $\Delta_{N, T}$.
Under Assumption \ref{assum:three}, 
the change-points which are
due to breaks in the loadings or (dis)appearance of new factors
(see Remark \ref{rem:representation}), are implicitly required to be `dense',
in the sense that the number of coordinates in $\mbf x_t$ 
affected by the changes is of order $n$.
Consequently, such changes appear in a large number (of order $n^2$) 
of elements of $\wt{\bm\Gamma}^{\beta(t)}_\chi(\tau)$ and, therefore, 
that of the WT of the second-order structure, $z_{\ell t}$.
Similarly, noting that
$[\wt{\bm\Gamma}^b_\chi(\tau)]_{i, i'} = \sum_{j, j'=1}^r \lambda_{ij}\lambda_{i'j'} \E(f_{jt}^b f_{j', t+\tau}^b)$,
a break in the autocovariance functions of $\mbf f_t^{\beta(t)}$ results in a dense change-point
that affects a large number of $z_{\ell t}$.
In other words, for change-point analysis in the common components, 
Assumption \ref{assum:b:two} is reduced to requiring
$(\sqrt{\jts}\log^{\theta + \upsilon}T)^{-1} T^{1/4} 
\min_{1 \le b \le B_\chi}\wt{\delta}_b \to \infty$, as $T\to\infty$,
allowing for the average jump size in individual coordinates $z_{\ell t}$ to tend to zero at a rate 
slower than $T^{-1/4}$, and is no longer dependent on $n$.
On the other hand, our model in \eqref{ps:fm} does not impose any assumption 
on the `denseness' of $\eta^\eps_b, \, b = 1, \ldots, B_\eps$, and 
sparse change-points (with $m_b \ll n^2$) in the idiosyncratic components are detectable provided 
that their sparsity is compensated by the size of jumps, $\wt\delta_b$.


Let $\wh{\cB} = \{\heta_b, \ b=1, \ldots, \wh B, \ 1 < \heta_1 < \ldots < \heta_{\wh{B}} < T\}$ 
denote the change-points detected from $\{y_{\ell t}, \,\ell=1,\ldots, N; \, t=1,\ldots, T\}$
by the DCBS algorithm, i.e., 
$\wh{\cB} = \wh{\cB}^\chi(k)$ or $\wh{\cB} = \wh{\cB}^\eps(k)$, 
depending on whether $y_{\ell t}$ is the WT of $\wh{\chi}^k_{it}$ or $\wh{\eps}^k_{it}$ for some fixed $k$.
Accordingly, we have either $\eta_b = \eta^\chi_b$ with $B = B_\chi$, 
or $\eta_b = \eta^\eps_b$ with $B = B_\eps$. 
Then the following theorem establishes that the DCBS algorithm performs
consistent change-point analysis for both the common and the idiosyncratic components.
\begin{thm}
\label{thm:dcbs}
{\it Suppose that Assumptions \ref{assum:one}--\ref{assum:b:two} hold,
and let $\theta$ be defined as in Theorem \ref{thm:common}. 
Also, let the threshold $\pi_{N, T}$ satisfy
$C'N\Delta_{N, T}^{-1}\log^{2\theta+2\upsilon}T < \pi_{N, T}  < C''\Delta_{N, T} T^{1/2}$ 
for some fixed $C', C'' > 0$.
Then, there exists $c_1>0$ such that 
\beas
\p\Big(\wh{B}=B; \,|\heta_b-\eta_b| < c_1\omega_{N, T} \mbox{ for } b=1, \ldots, \wh B\Big) \to 1
\eeas
as $T \to \infty$, where
$\omega_{N, T} = N\Delta_{N, T}^{-2} \log^{2\theta+2\upsilon}T$.}
\end{thm}

Theorem \ref{thm:dcbs} shows both the total number and the locations of 
the change-points in $\mc B^\chi$ and $\mc B^\eps$ are consistently estimated;
in the rescaled time $t/T\in [0, 1]$, the bound on the bias in estimated change-point locations satisfies 
$\omega_{N, T}/T \to 0$ as $T \to \infty$ under Assumption \ref{assum:b:two}.
The optimality in change-point detection may be defined as
when the true and estimated change-points are within the distance of $O_p(1)$ \citep{korostelev1987}.
With our approach, near-optimality in change-point estimation is achieved up to a logarithmic factor
when the change-points are cross-sectionally dense ($m_b \sim N$)
with average size $\wt\delta_b$ bounded away from zero, so that $\Delta_{N, T}\sim \sqrt N$. 

\subsection{Screening over a range of factor number candidates}
\label{sec:factor:number}

In this section, we detail a screening procedure motivated by Theorem \ref{thm:overestimation},
which enables us to bypass the challenging task of estimating the number of factors
in the presence of (possibly) multiple change-points in the factor structure. 

The performance of most methods proposed for change-point analysis under factor modelling, 
such as those listed in Introduction, relies heavily on accurate estimation of the factor number. 
Thanks to Theorem \ref{thm:overestimation}, however,
mis-specifying the factor number in our methodology does not 
influence the theoretical consistency as reported in Theorem \ref{thm:dcbs},
provided that we choose $k$ sufficiently large to satisfy $k \ge r$.
Based on this observation, we propose to screen $\wh{\cB}^\chi(k)$, 
the set of change-points detected from $\wh{\chi}^k_{it}$,
for a range of factor number candidates denoted as $\cR = \{\underline{r}, \underline{r}+1, \ldots, \bar{r}\}$. 
Specifically, we select the $\wh{\cB}^\chi(k)$ with the largest cardinality over $k$ and, 
if there is a tie, we select $\wh{B}^\chi(k)$ with the largest $k$.
Denoting $k^* = \max\{k \in \cR: \, |\wh{\cB}^\chi(k)| = \max_{k'\in\cR} |\wh{\cB}^\chi(k')|\}$,
change-points in the idiosyncratic components are detected from 
$\wh{\eps}_{it}^* \equiv \wh{\eps}_{it}^{k^*} = x_{it} - \wh{\chi}_{it}^*$,
where $\wh \chi^*_{it} \equiv \wh \chi^{k^*}_{it}$. 

As noted below Theorem \ref{thm:overestimation}, 
using $k < r$ factors leads to $\wh\chi^k_{it}$ with non-negligible estimation error
which may not contain all the change-points in $\cB^\chi$ as change-points in its second-order structure.
Moreover, with such $k$, some $\eta^\chi_b$'s may appear as change-points in $\wh{\eps}^k_{it}$,
which we refer to as the spillage of change-points in the common components over to the idiosyncratic components.
This justifies the proposed screening procedure and the choice of $k^*$.

For $\underline{r}$, we use the number of factors estimated by minimising the information criterion of \cite{baing02}
using as penalty $p(n, T) = (n \wedge T)^{-1}\log(n \wedge T)$;
in principle, any other procedure for factor number estimation may be adopted to select $\underline{r}$.
The range $\cR$ also involves the choice of the maximum number of factors, $\bar{r}$.
In the factor modelling literature, this choice is a commonly faced problem,
e.g., when estimating the number of factors using an information criterion-type estimator. 
In the literature on stationary factor models,
the maximum number of factors is usually fixed at a small number
($\bar{r} = 8$ is used in \cite{baing02}) for practical purposes. 
However, the presence of change-points tends to increase the number of factors,
as shown with an example in Appendix \ref{sec:ex}.
Therefore, we set $\bar{r}  = 20\vee \lfloor \sqrt{n\wedge T}\rfloor$, 
where the second term dominates the first when $n$ and $T$ are large.

\section{Factor analysis after change-point detection}
\label{sec:post}

Once the change-points are detected, we can estimate the factor space over each segment
$I^\chi_b = [\wh\eta^\chi_b+1, \wh\eta^\chi_{b+1}], \, b = 0, \ldots, \wh{B}_\chi$,
defined by two consecutive change-points estimated from the common components.
Denoting the sample covariance matrix over $t \in I^\chi_b$ by
$\wh{\bm\Gamma}^b_x = (\wh\eta^\chi_{b+1} - \wh\eta^\chi_b)^{-1} 
\sum_{t \in I^\chi_b} \mbf x_t\mbf x_t^\top$,
let $\wh{\mbf w}^b_{x, j}$ denote the eigenvector of $\wh{\bm\Gamma}^b_x$ corresponding 
to its $j$-th largest eigenvalue. 
Then, for a fixed number of factors $k$, the segment-specific estimator of $\chi_{it}$ is obtained via PCA as
$\wh{\chi}^{(b, k)}_{it} = \sum_{j=1}^{k} \wh{w}^b_{x, ij}(\wh{\mbf w}^b_{x, j})^\top\mbf x_t$. 

The number of factors $r_b$ in the $b$-th segment can be estimated by means of the information criterion proposed in \cite{baing02}:
\bea
\wh{r}_b = \arg\min_{1 \le k \le \bar{r}} \Big[\log\Big\{\frac{1}{n(\wh\eta^\chi_{b+1}-\wh\eta^\chi_b)}
\sum_{i=1}^n\sum_{t=\wh\eta^\chi_b+1}^{\wh\eta^\chi_{b+1}}(x_{it}-\wh{\chi}^{(b,k)}_{it})^2\Big\} 
+ k p(n, T) \Big],
\label{eq:ic}
\eea
where $\bar{r}$ denotes the maximum allowable factor number,
and the penalty function $p(n, T)$ satisfies 
$p(n, T) \to 0$ as well as $(n^2 \wedge \sqrt{T})\log^{-2/\beta_f}T \cdot p(n, T) \to \infty$. 
Motivated by the formulation of penalties in \cite{baing02}, we may use
$p(n, T) = (n + \sqrt T)/(n\sqrt T) \log^a(n \wedge \sqrt T)$ for some $a > 2/\beta_f$.

Let $\bm\Lambda_b$ be the $n\times r_b$ matrix of loadings for the $b$-th segment. 
In order to discuss the theoretical properties of segment-wise factor analysis,
we require the following assumption that extends Assumption \ref{assum:three} (i) to each segment $I^\chi_b$.
\begin{assum}
\label{assum:three:new}
There exists a positive definite 
$r_b \times r_b$ matrix $\mbf H_b$ such that $n^{-1}\bm\Lambda_b^\top\bm\Lambda_b \to \mbf H_b$ 
as $n\to\infty$. 
\end{assum}
Then, we obtain the asymptotic results in Propositions \ref{prop:bn}--\ref{thm:pca}
for the segment-wise estimators of the factor number and the common components.

\begin{prop}
\label{prop:bn} 
{\it Suppose that all the conditions in Theorem \ref{thm:dcbs} and Assumption \ref{assum:three:new} hold.
For all $b = 0, \ldots, \wh B_\chi$, $\wh r_b$ returned by \eqref{eq:ic} satisfies
$\p(\wh r_b = r_b) \to 1$ as $n, T \to \infty$.}
\end{prop}

\begin{prop}
\label{thm:pca}
{\it Suppose that all the conditions in Theorem \ref{thm:dcbs} and Assumption \ref{assum:three:new} hold.
Then, for all $b = 0, \ldots, \wh{B}_\chi$,
\beas
\max_{1 \le i \le n} \max_{t \in I^\chi_b}
\vert \wh\chi^{(b,r_b)}_{it} - \chi_{it}\vert = 
O_p\l\{\Big(\sqrt{\frac{\log\,n}{T}} \vee \frac{1}{\sqrt{n}}\Big) \log^\theta T\r\}.
\eeas}
\end{prop}

From Propositions \ref{prop:bn} and \ref{thm:pca},
we have the guarantee that the factor space is consistently estimated by PCA
over each estimated segment.
We may further refine the post change-point analysis by first determining
whether $\wh\eta^\chi_b$ can be associated with (a) a break in the loadings or factor number, or
(b) that in the autocorrelation structure of the factors only.
This can be accomplished by comparing $\wh r_{b-1}$ and $\wh r_b$ against
the factor number estimated from the pooled segment $I^\chi_{b-1} \cup I^\chi_b$:
a break in the loadings or factor number necessarily brings in a change
in the number of factors from the pooled segment.
However, if (b) is the case, the segments before and after $\wh\eta^\chi_b$ as well as the pooled one
return the identical number of factors, and 
we can perform the joint factor analysis of the two segments.

\section{Computational aspects}

\subsection{Bootstrap procedure for threshold selection}
\label{sec:bootstrap}

The range of theoretical rates supplied for $\pi_{N, T}$ in Theorem \ref{thm:dcbs},
involves typically unattainable knowledge of the minimum cross-sectional size of the changes.
Hence, we propose a bootstrap procedure for the selection of $\pi_{N, T}^\chi$ and $\pi_{N, T}^\eps$,
the thresholds for change-point analysis of $\wh\chi^k_{it}$ and $\wh\eps^k_{it}$, respectively.
We omit $N$ and $T$ from their subscripts for notational convenience when there is no confusion.
Although a formal proof on the validity of the proposed bootstrap algorithm is beyond the scope of the current paper,
simulation studies reported in Section \ref{sec:sim} demonstrate its good performance
when applied jointly with our proposed methodology.
We refer to \cite{trapani2013}, \cite{corradi2014} and \cite{gonccalves2014, gonccalves2016}
for alternative bootstrap methods under factor models
and \cite{jentsch2015} for the linear process bootstrap for multivariate time series in general.


We propose a new bootstrap procedure which is specifically motivated by 
the separate treatment of common and idiosyncratic components in our change-point detection methodology.
Namely, the resampling method produces bootstrap samples from 
the common and idiosyncratic components independently,
relying on the consistency of the estimated components 
with an over-specified factor number as reported in Theorem \ref{thm:overestimation}.

Let $\cT^\chi_{s, e}(k)$ and $\cT^\eps_{s, e}(k)$ 
denote the test statistics $\cT_{s, e}$ computed on the interval $[s, e]$
for the panel data obtained from the WT of $\wh{\chi}^k_{it}$ and $\wh \eps^k_{it}$, respectively.
The proposed resampling procedure aims at approximating
the distributions of $\cT^\chi_{s, e}(k)$ and $\cT^\eps_{s, e}(k)$
under the null hypothesis of no change-point,
which then can be used for the selection of $\pi^\chi_{s, e}(k)$ and $\pi^\eps_{s, e}(k)$,
the corresponding, interval-specific test criteria for $\wh{\chi}^k_{it}$ and $\wh\eps^k_{it}$ over $[s, e]$,
an interval which is considered at some iteration of the DCBS algorithm.

Among the many block bootstrap procedures proposed in the literature 
for bootstrapping time series (see \cite{politis2004} for an overview), 
stationary bootstrap (SB) proposed in \cite{politis1994} generates bootstrap samples 
which are stationary conditional on the observed data (see Appendix \ref{app:SB} for details).
Based on the SB, our procedure derives $\pi^\chi_{s, e}(k)$ and $\pi^\eps_{s, e}(k)$.
Recall that $\wh\lambda_{ij}$ and $\wh{f}_{jt}$ denote the loadings and factors estimated via 
the capped PCA for $j=1,\ldots, k$, see Section \ref{sec:sep}.
\begin{description}
\item[Stationary bootstrap algorithm for factor models.]
\item[Step 1] \hfill
\begin{description}
\item[For the common components:]
For each $l \in \{1, \ldots, k\}$, 
produce the SB sample of $\{\wh f_{lt}, \, t=1,\ldots, T\}$ as $\{f^\bullet_{lt}, t=1,\ldots, T\}$. 
Compute $\chi^{k\bullet}_{it} = \sum_{l=1}^k \wh \lambda_{il} f^\bullet_{lt}$.

\item[For the idiosyncratic components:]
Produce the SB sample of 
$\{\wh{\bm\eps}^k_t = (\wh \eps^k_{1t}, \ldots, \wh \eps^k_{nt})^\top, \, t=1,\ldots, T\}$ as 
$\{\bm\eps^{k \bullet}_t = (\eps_{1t}^{k \bullet}, \ldots, \eps_{nt}^{k \bullet})^\top, \, t=1,\ldots, T\}$.
\end{description}

\item[Step 2:] Generate $y^\bullet_{\ell t}$ through transforming $\chi^{k\bullet}_{it}$ or $\eps^{k\bullet}_{it}$
using $g_j(\cdot)$ and $h_j(\cdot,\cdot)$ as described in Section \ref{sec:choice}.

\item[Step 3:] Compute $\cY^{\ell \bullet}_{s, b, e}$ on $y^\bullet_{\ell t}$
and generate the test statistic $\cT^\bullet_{s, e}$ according to \eqref{dc:test:stat}.

\item[Step 4:] Repeat Steps 1--3 $R$ times.
The critical value $\pi^\chi_{s, e}(k)$ or $\pi^\eps_{s, e}(k)$ for the segment $[s, e]$ 
is selected as the $(1-\alpha)$-quantile of the $R$ bootstrap test statistics $\cT^\bullet_{s, e}$ 
at given $\alpha \in (0, 1)$.
\end{description}

The bootstrap algorithm is designed to produce ${\bm\chi}^{k\bullet}_t$ and ${\bm\eps}^{k\bullet}_t$
that mimic the second-order structure of $\bm{\wh\chi}^k_t$ and $\wh{\bm\eps}^k_t$, respectively, 
when there is no change-point present,
and thus approximates the distributions of the test statistics under the null hypothesis.
In the algorithm, the treatment of the common and idiosyncratic components 
differs only in the application of SB in Step 1:
since the factors estimated from the PCA are uncorrelated by construction,
we generate the SB samples of $\wh f_{jt}$ for each $j$ separately,
while the $n$ elements of $\wh{\bm\eps}^k_t$ are resampled jointly 
in an attempt to preserve the cross-sectional dependence therein.
We discuss the choice of the bootstrap size $R$ and the level of quantile $\alpha$ in Section \ref{sec:tuningp}.

\subsection{Selection of tuning parameters}
\label{sec:tuningp}

The proposed methodology involves the choice of tuning parameters 
for the capped PCA, WT, DCBS algorithm and the bootstrap procedure. 
We here list the values used for the simulation studies (Section \ref{sec:sim})
and real data analysis (Section \ref{sec:real}). We provide further guidance on the implementation of the proposed methodology in Appendix \ref{sec:tricks}.

In the examples considered in Section \ref{sec:sim}, 
we have not observed unreasonably large contributions to $\wh\chi^k_{it}$ 
from spurious factors. 
As noticed in Section \ref{sec:sep}, 
selecting a sufficiently large constant as $c_w$ effectively 
disables the capping for the (unknown) $r$ leading principal components 
to $\wh\chi^k_{it}$.
For this reason, in the current implementation, we disable the capping.
However, this does not necessarily mean that capping will always be of no practical use.
Therefore, we recommend the data-driven choice of 
$c_w = \sqrt n \max_{1 \le j \le \underline{r}} \max_{1 \le i \le n} |\wh{w}_{x, ij}|$.
With such $c_w$, the capping is enabled for only those $\wh{w}_{x, ij}$
with $\underline{r} + 1 \le j \le \bar{r}$,
where $\underline{r}$ and $\bar{r}$ denote
the smallest and largest number of factors considered in Section \ref{sec:factor:number}.

For the WT, we propose to use $\jts = \lfloor C\log_2\log^\upsilon T \rfloor$ number of finest Haar wavelets
for some $\upsilon \in (0, 1]$ and $C > 0$, in order to control for any bias in change-point estimation arising from WT. 
Noting that the bias increases at the rate $2^C$ with increasing $C$,
we recommend the choice of $\jts = \lfloor\log_2\log_2\, T\rfloor$ in practice.

Although omitted from the description of the DCBS algorithm, we select a parameter $d_T$
controlling the minimum distance between two estimated change-points.
In light of Remark \ref{rem:beta} and Theorem \ref{thm:dcbs},
we choose to use $d_T= [\log^2 T \wedge 0.25T^{6/7}]$. 
Note that we can avoid using this parameter by replacing the binary segmentation procedure 
with wild binary segmentation \citep{fryzlewicz2014a}, such that
the Double CUSUM is applied over randomly drawn intervals to derive the test statistic.
It is conjectured that such a procedure will place a tighter bound 
on the bias in estimated change-point locations, as well as bypassing the need for the parameter $d_T$.
We leave the investigation in this direction for the future research.

Finally, for the proposed bootstrap procedure, 
we use the bootstrap sample size $R=200$ for the simulation studies and $R=500$ for the real data analysis.
As for the level of quantile, we select $\alpha = 0.05$; 
note that this choice does not indicate the significance level in hypothesis testing.

\section{Simulation studies}
\label{sec:sim}

In this section, we apply the proposed change-point detection methodology 
to both single and multiple change-point scenarios under factor modelling.
While our methodology is designed for multiple change-point detection,
Section \ref{sec:sim:single} gives us insight into its performance 
in the presence of a single change-point,
which is of the type, size and denseness that vary in a systematic manner.
Multiple change-point scenarios are considered in Section \ref{sec:sim:multi}. 

\subsection{Single change-point scenarios}
\label{sec:sim:single}

The following stationary $q$-factor model
allows for serial correlations in $f_{jt}$ and both serial and cross-sectional correlations in $\eps_{it}$:
\begin{align}
& x_{it} = \sum_{j=1}^q \lambda_{ij}f_{jt} + \sqrt{\vartheta}\eps_{it}, \quad \mbox{where} \quad 
\vartheta = \phi \cdot \frac{q}{1-\rho_f^2} \cdot \frac{1-\rho^2}{1+2H\beta^2}, \label{eq:sim:single:model}
\\
& f_{jt} = \rho_{f, j} f_{j, t-1} + u_{jt}, \quad u_{jt} \sim_{\iid} \cN(0, 1), \label{eq:sim:single:model:f}
\\
& \eps_{it} = \rho_{\eps,i} \eps_{i, t-1} + v_{it} + \beta_i \sum_{|k| \le H_i, \, k \ne 0} v_{i+k, t}, \quad v_{it} \sim_{\iid} \cN(0, 1), \label{eq:sim:single:model:e}
\end{align}
and $\lambda_{ij} \sim_{\iid} \cN(0, 1)$. 
The stationary factor model in \eqref{eq:sim:single:model} has been frequently adopted for 
empirical studies in the factor model literature including \cite{baing02}.
Throughout we set $q = 5$. T
he parameters $\rho_{f, j} = \rho_f-0.05(j-1)$ with $\rho_f = 0.4$, and 
$\rho_{\eps,i} \sim_{\iid} \cU(-\rho, \rho)$ with $\rho = 0.5$,
determine the autocorrelations in the factors and idiosyncratic components, 
while $H_i = H = \min(n/20, 10)$ and $\beta_i$ (randomly drawn from $\{-\beta, \beta\}$ with $\beta = 0.2$), 
determine the cross-sectional correlations in $\bm\eps_t$. 
The parameter $\phi$, or its inverse $\phi^{-1}$, is chosen from $\{1, 1.5, 2, 2.5\}$ 
depending on the change-point scenario, 
in order to investigate the impact of the ratio between the variance of the common and idiosyncratic components,
on the performance of the change-point detection methodology. 
We fix the number of observations at $T = 200$ and vary the dimensionality as $n \in \{100, 300\}$.

A single change-point is introduced to either $\chi_{it}$ or $\eps_{it}$ at $\eta_1 = [T/3] = 67$ as follows.
\begin{description}
\item[(S1) Change in the loadings.] For a randomly chosen index set $\mathcal{S} \subset \{1, \ldots, n\}$, 
the loadings $\lambda_{ij}, \, i \in \mathcal{S}$ are shifted by $\Delta_{ij} \sim_{\iid}\cN(0, \sigma^2)$, 
where  $\sigma \in \sqrt{2}\{1, 0.75, 0.5, 0.25\}$. 

\item[(S2) Change in $\rho_{f, j}$.] The signs of the AR parameters in (\ref{eq:sim:single:model:f}) are switched 
such that the autocorrelations of $f_{jt}$ change while their variance remains the same. 

\item[(S3) A new factor.] For a randomly chosen $\mathcal{S} \subset \{1, \ldots, n\}$, 
a new factor is introduced to $\chi_{it}, \, i \in \mathcal{S}$:
$f_{q+1, t} = \rho_f f_{q+1, t-1} + u_{q+1, t}$ with $u_{q+1, t} \sim_{\iid} \cN(0, 1)$
and $\lambda_{i, q+1} \sim_{\iid}\cN(0, \sigma^2)$, where $\sigma = \sqrt{2}$.

\item[(S4) Change in $\rho_{\eps,i}$.] For a randomly chosen $\mathcal{S} \subset \{1, \ldots, n\}$, 
the corresponding $\rho_{\eps,i}, \, i \in \mathcal{S}$ in (\ref{eq:sim:single:model:e}) 
have their signs switched so that the autocorrelations of such $\eps_{it}$ change
while their variance remains the same.

\item[(S5) Change in the covariance of $\bm\eps_{t}$.] 
For a randomly chosen $\mathcal{S} \subset \{1, \ldots, n\}$, 
the bandwidth $H_i, \, i \in \mathcal{S}$ in (\ref{eq:sim:single:model:e}) doubles.
\end{description}
In (S1) and (S3)--(S5), the size of the index set $\mc S$ is controlled by the parameter 
$\varrho \in \{1, 0.75, 0.5, 0.25\}$, as $|\mathcal{S}| = [\varrho n]$.


For each simulated dataset, we consider the range of possible factor numbers $k\in \cR$ 
selected as described in Section \ref{sec:factor:number}. 
For any given $k$, we estimate the common and idiosyncratic components and proceed with 
WT of Stage 1, to which {\it only the first iteration} of the DCBS algorithm is applied.
In this way, we test for the existence of at most a single change-point 
in the common and idiosyncratic components separately and, 
if its presence is detected, we identify its location in time. 
With a slight abuse of notation, we refer to the simultaneous testing and locating procedure as the DC test, 
and report its detection power and accuracy in change-point estimation. 

As noted in Introduction, existing methods for change-point analysis under factor modelling 
are not applicable to the scenarios other than (S1) and (S3). 
Hence, we compare the performance of the DC test to change-point tests based on 
two other high-dimensional CUSUM aggregation approaches, which are referred to as the MAX and AVG tests:
after Stage 1, the $N$-dimensional CUSUMs of the WT series
are aggregated via taking their point-wise maximum or average, respectively, 
which replaces the Double CUSUM statistics in Stage 2. 
In addition, we report the detection power of a variant of the DC test,
where the first iteration of the DCBS algorithm is applied to the panel data consisting 
of the WT of $x_{it}$, 
and compare its performance against DC, MAX and AVG tests applied to the WT of $\wh\chi^k_{it}$ under (S1)--(S3),
and that of $\wh\eps^k_{it}$ under (S4)--(S5).
We include this approach, termed the DC-NFA (no factor analysis) test, 
in order to demonstrate the advantage in linking the factor modelling and change-point detection 
as proposed in our methodology.

Figures \ref{sim:fig:100:chi:one:power:one}--\ref{sim:fig:100:vep:two:power} plot
the detection power of our change-point test as well as that of MAX, AVG and DC-NFA tests
under different scenarios over $100$ realisations when $n = 100$
(we only present the results from (S1) when $\sigma \in \{\sqrt{2}, 0.02\sqrt{2}\}$).
Since $\chi_{it}$ ($\eps_{it}$) do not contain any change under (S4)--(S5) ((S1)--(S3)),
testing for a change-point in $\wh{\eps}_{it}^*$ ($\wh{\chi}_{it}^*$) in these scenarios
offers insights into the size behaviour of DC, MAX and AVG tests.

Overall, change-point detection becomes more challenging as $\sigma$ (the size of changes) or 
$\varrho$ (the proportion of the coordinates with the change) decreases,
and also as $\Var(\chi_{it})/\Var(x_{it})$ decreases (with increasing $\phi$) under (S1)--(S3)
and increases (with increasing $\phi^{-1}$) under (S4)--(S5).
In all scenarios, the DC test shows superior performance to the DC-NFA test.
This confirms that the factor analysis prior to change-point analysis is an essential step in markedly improving 
the detection power, as well as in identifying the origins of the change-points; without a factor analysis,
a change-point that appears in either of $\chi_{it}$ or $\eps_{it}$ 
may be `masked' by the presence of the other component and thus escape detection.

The DC and AVG tests generally attain similar powers, while the MAX test tends to have 
considerably lower power in scenarios such as (S2), (S4) and (S5).
An exception is under (S3), where the MAX test attains larger power than the others 
in some settings with decreasing $\varrho$ (Figure \ref{sim:fig:100:chi:three:power}).
It may be explained by the fact that the smaller $\varrho$ is, 
the sparser the change-point becomes cross-sectionally,
which is a setting that favours the approach taken in MAX in aggregating the high-dimensional CUSUM series
(see the discussions in Section 2.1 of \cite{cho2015}).
Between the DC and AVG tests, the former outperforms the latter in the more challenging settings 
when the change is sparse cross-sectionally (with small $\varrho$),
particularly when the change is attributed to the introduction of a single factor to the existing five as in (S3).

Apart from (S5), the origin of the change-point is correctly identified in the sense that 
it is detected only from $\wh\chi_{it}^*$ under (S1)--(S3) or 
$\wh \eps_{it}^*$ under (S4) with power strictly above $\alpha = 0.05$.
In (S5), we observe spillage of the change-point in $\eps_{it}$:
the change-point is detected from both $\wh\chi_{it}^*$ and $\wh\eps_{it}^*$,
particularly when a large proportion of $\eps_{it}$ undergoes the change 
of an increase in the bandwidth $H_i$ in (\ref{eq:sim:single:model:e}), 
and when $\Var(\chi_{it})/\Var(x_{it})$ is small (with large $\phi$),
see Figure \ref{sim:fig:100:vep:two:power}.
This supports the notion of a pervasive change-point being a factor, 
i.e., a significant break that affects the majority of the idiosyncratic components in their second-order structure,
may be viewed as a {\it common} feature.
We also note that due to the relatively large variance of common component, 
the detection power of the DC-NFA test behaves as that of the DC test applied to $\wh\chi^*_{it}$ 
rather than $\wh\eps^*_{it}$ in this scenario.

We present in a supplementary document the results on 
other data generating processes taken from the literature on testing for a single change-point under factor models,
as well as tables and figures summarising the simulation results for (S1)--(S5) with $n = 300$,
along with box plots of the estimated change-points $\heta_1$. 
The performance of all tests under consideration generally improve when $n = 300$.

\begin{figure}[t!]
\centering
\includegraphics[scale=.5]{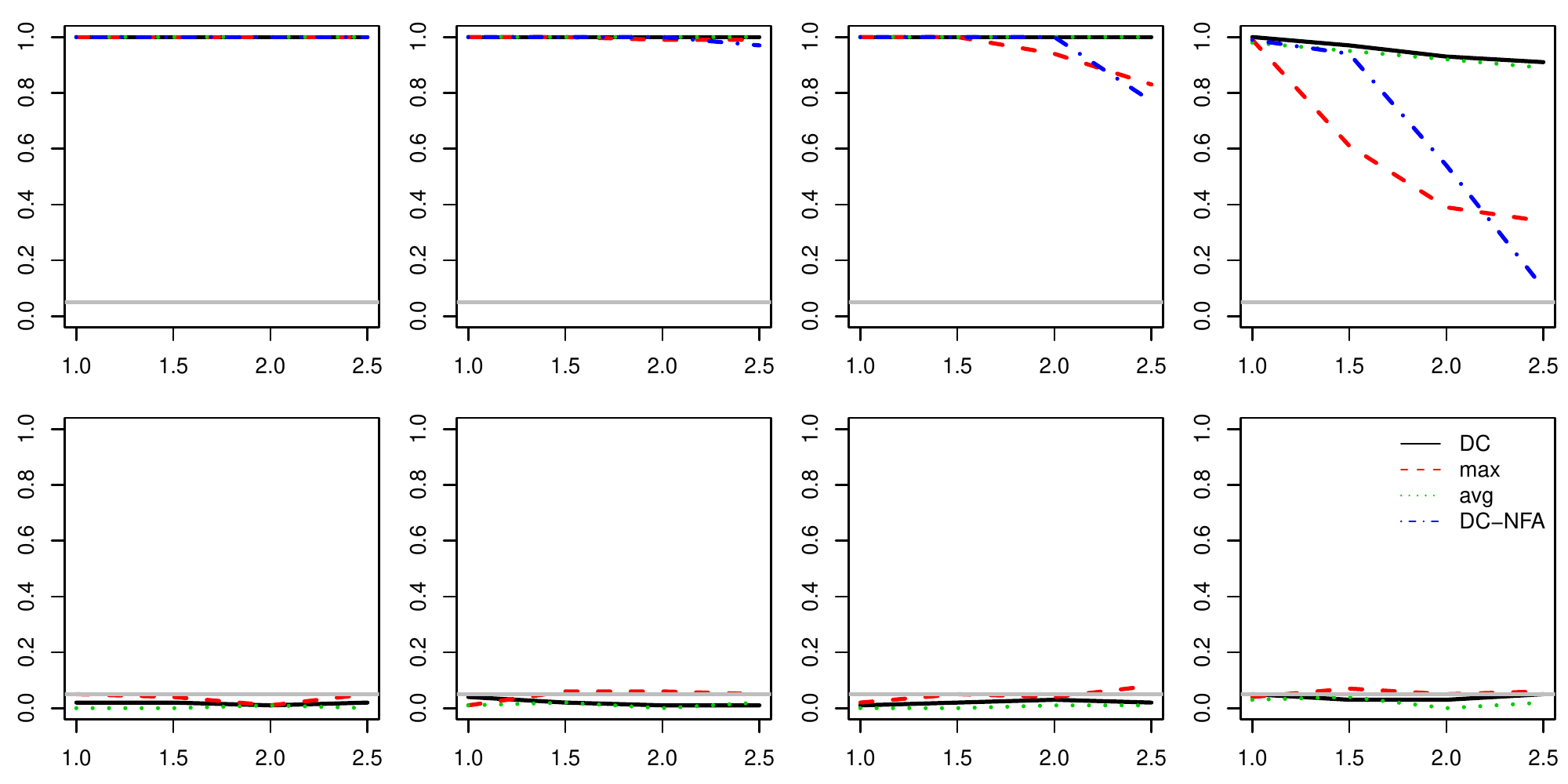}
\caption{\footnotesize(S1) Detection power ($y$-axis) of change-point tests on the common (top) and 
idiosyncratic (bottom) components with varying $\phi \in \{1, 1.5, 2, 2.5\}$ ($x$-axis)
when $n=100$, $T=200$, $\sigma = \sqrt{2}$ and $\varrho \in \{1, 0.75, 0.5, 0.25\}$ (left to right);
horizontal grey lines indicate the significance level $\alpha = 0.05$.}
\label{sim:fig:100:chi:one:power:one}
\end{figure}

%

\begin{figure}[t!]
\centering
\includegraphics[scale=.5]{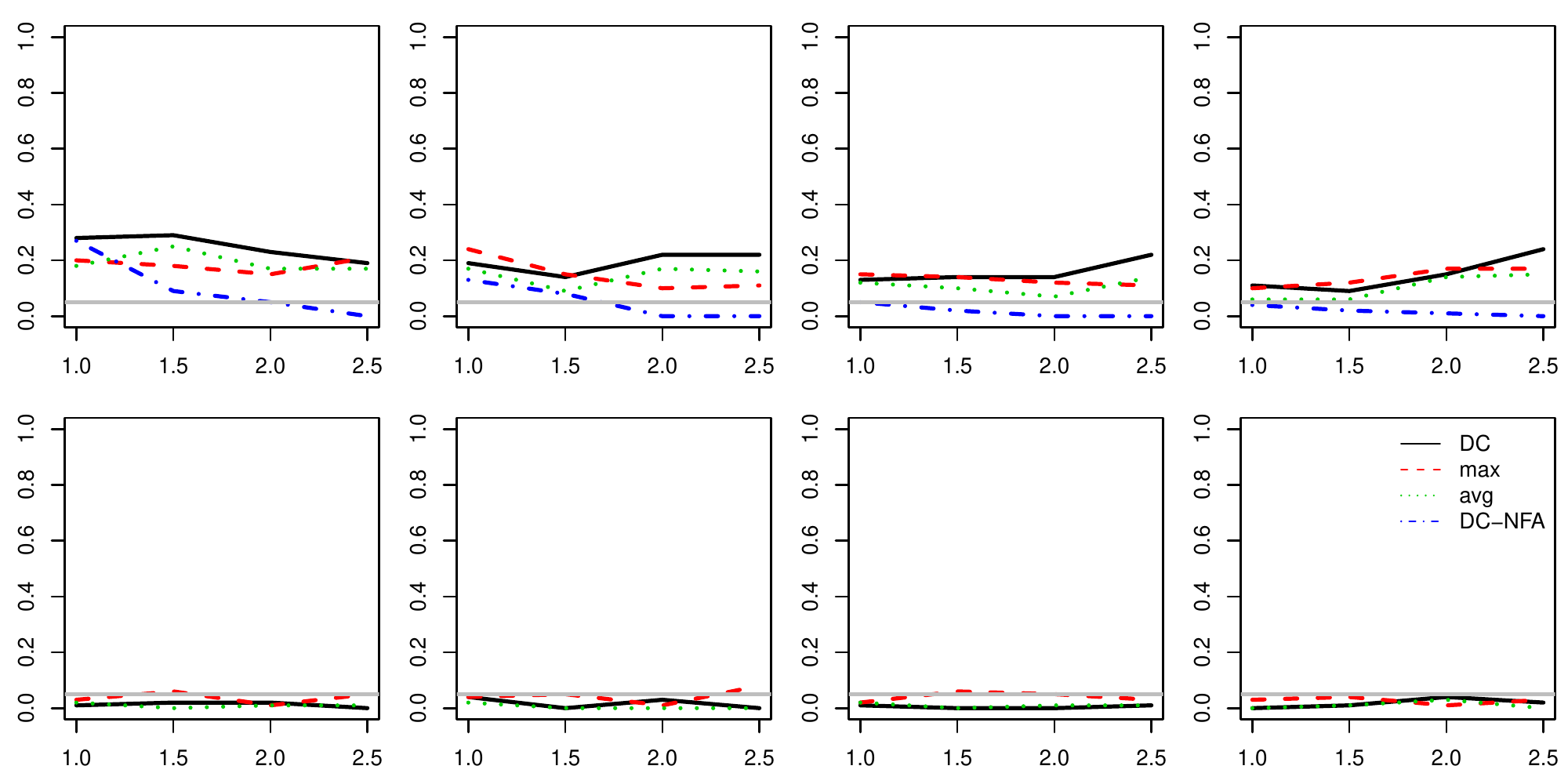}
\caption{\footnotesize(S1) Detection power of change-point tests when $\sigma = 0.25\sqrt{2}$.}
\label{sim:fig:100:chi:one:power:four}
\end{figure}

\begin{figure}[t!]
\centering
\includegraphics[scale=.65]{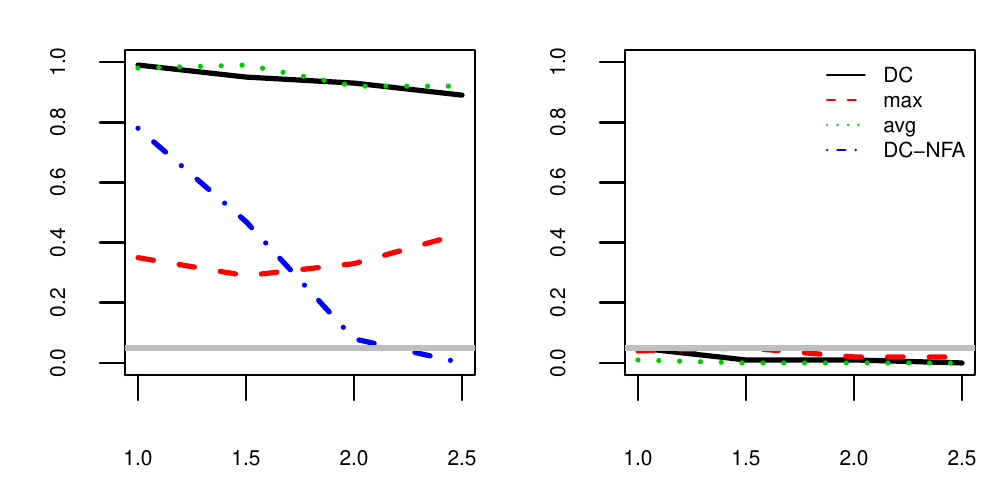}
\caption{\footnotesize(S2) Detection power of change-point tests on the common (left) and idiosyncratic (right) components 
with varying $\phi \in \{1, 1.5, 2, 2.5\}$ ($x$-axis)
when $n=100$ and $T=200$.}
\label{sim:fig:100:chi:two:power}
\end{figure}

\begin{figure}[t!]
\centering
\includegraphics[scale=.5]{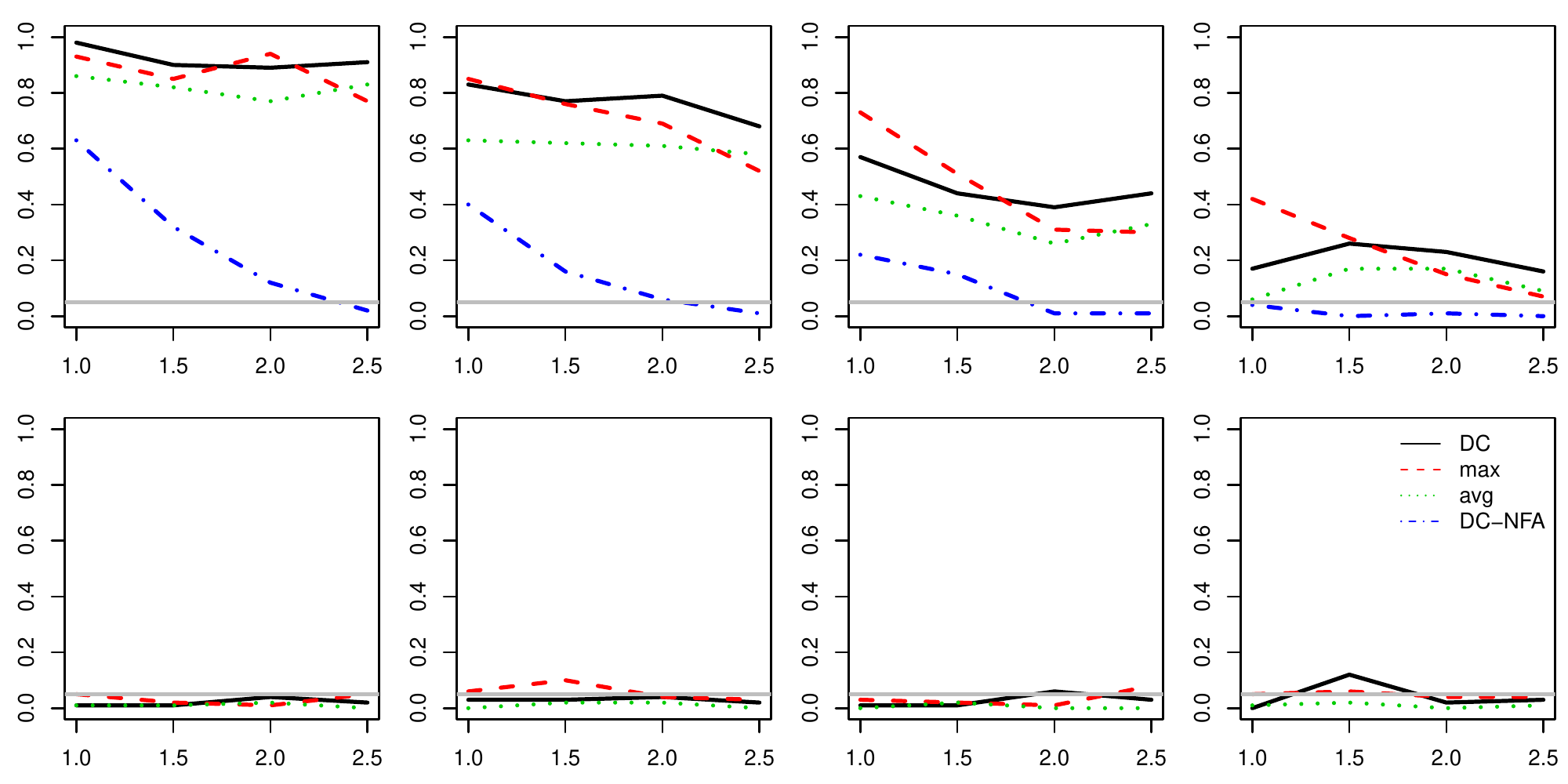}
\caption{\footnotesize(S3) Detection power of change-point tests on the common (top) and idiosyncratic (bottom) components 
with varying $\phi \in \{1, 1.5, 2, 2.5\}$ ($x$-axis)
when $n=100$, $T=200$, $\sigma = \sqrt{2}$ and $\varrho \in \{1, 0.75, 0.5, 0.25\}$ (left to right).}
\label{sim:fig:100:chi:three:power}
\end{figure}

\begin{figure}[t!]
\centering
\includegraphics[scale=.5]{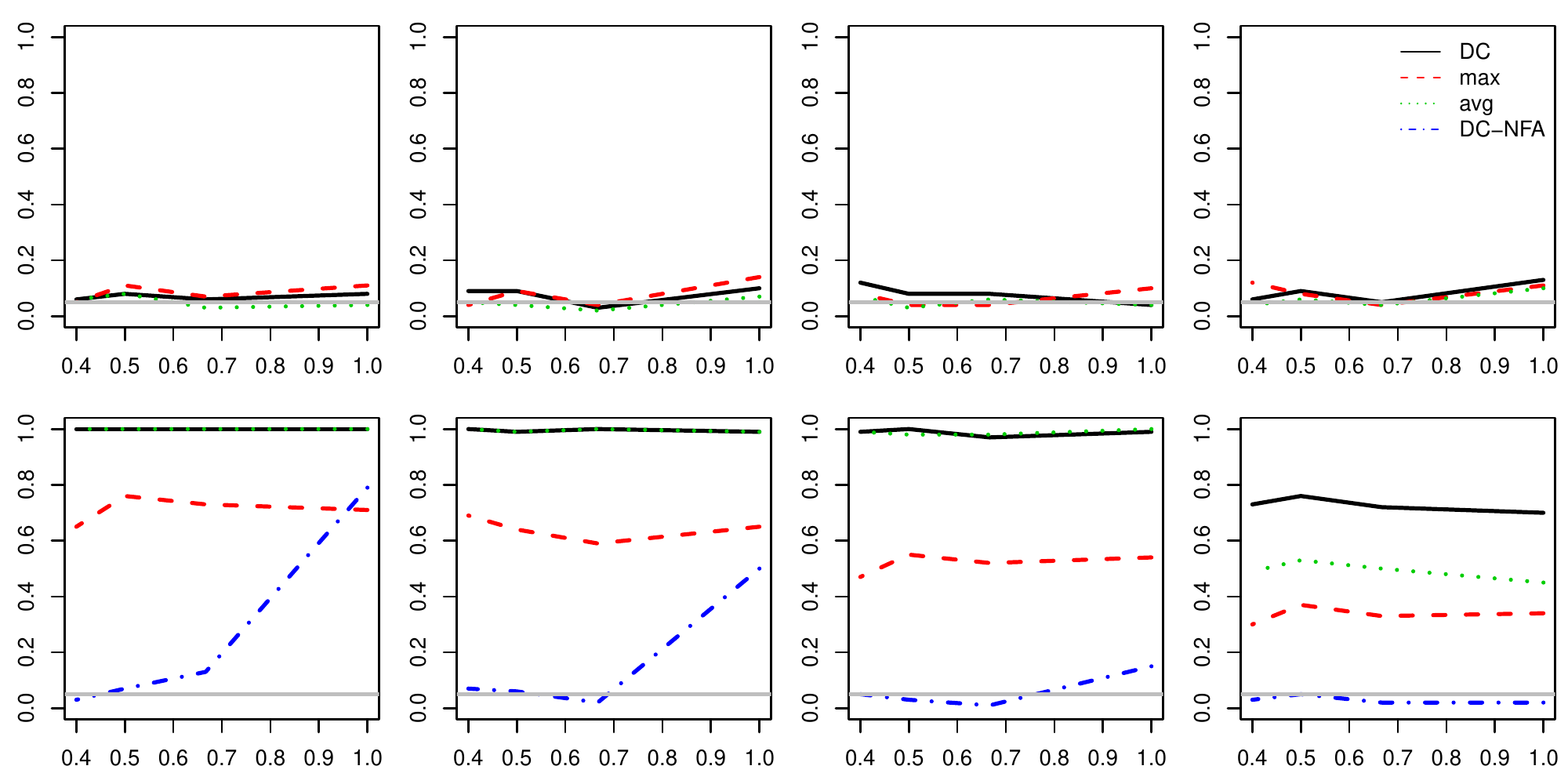}
\caption{\footnotesize(S4) Detection power of change-point tests on the common (top) and idiosyncratic (bottom) components 
with varying $\phi \in \{2.5^{-1}, 2^{-1}, 1.5^{-1}, 1\}$ ($x$-axis)
when $n=100$, $T=200$ and $\varrho \in \{1, 0.75, 0.5, 0.25\}$ (left to right).}
\label{sim:fig:100:vep:one:power}
\end{figure}

\begin{figure}[t!]
\centering
\includegraphics[scale=.5]{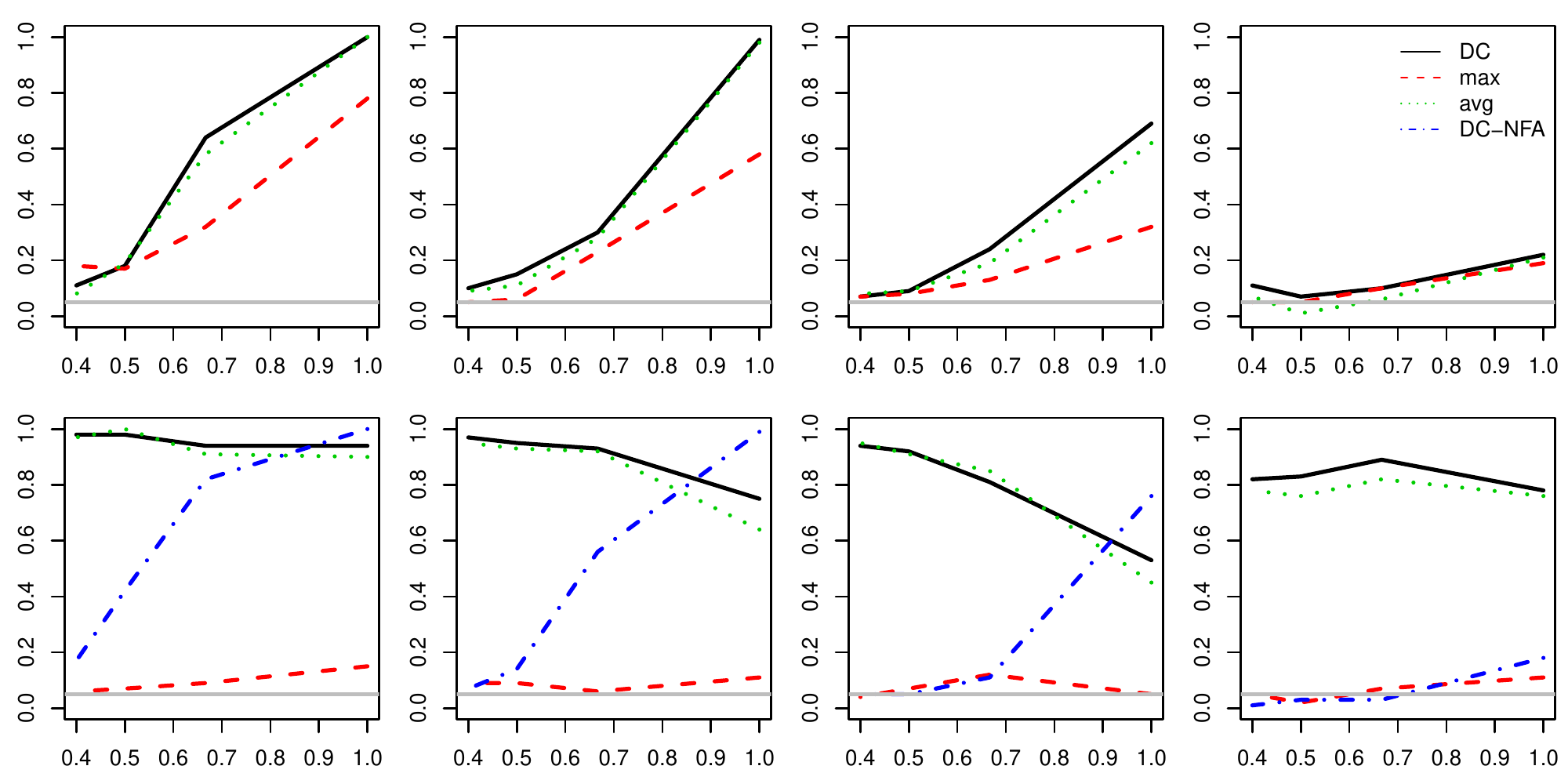}
\caption{\footnotesize (S5) Detection power of change-point tests on the common (top) and idiosyncratic (bottom) components 
with varying $\phi \in \{2.5^{-1}, 2^{-1}, 1.5^{-1}, 1\}$ ($x$-axis)
when $n=100$, $T=200$ and $\varrho \in \{1, 0.75, 0.5, 0.25\}$ (left to right).}
\label{sim:fig:100:vep:two:power}
\end{figure}

\subsection{Multiple change-point scenarios}
\label{sec:sim:multi}

\subsubsection{Model (M1)}
\label{sec:sim:multi:model:M1}

This model is intended to mimic the behaviour of the Standard and Poor's 100 log-return data, 
analysed in Section \ref{sec:sp100}. 
The information criterion of \cite{baing02} returned $q=4$ factors for these data 
and, imposing stability on the loadings, common factors four and idiosyncratic components 
were estimated by means of PCA as $\wh f_{jt}, \, j = 1, \ldots, q$ and $\wh \eps_{it}$.
Estimated factors exhibit little serial correlations, show 
the evidence of multiple change-points in their variance and heavy tails.
The same evidence holds for the estimated idiosyncratic components.
Based on these observations, we adopt the following data generating model:
\begin{align}
& x_{it} = \chi_{it} + \sqrt{\vartheta}\eps_{it} = \sum_{j=1}^q \wh\lambda_{ij} f_{jt} + \sqrt{\vartheta} \eps_{it}, 
\quad \mbox{where} \quad
\vartheta = \phi \cdot \sum_{i=1}^n \wh{\Var}(\chi_{it})\l\{\sum_{i=1}^n \wh{\Var}(\eps_{it})\r\}^{-1}, \nn
\\
& f_{jt} = \l\{\begin{array}{ll}
u_{j, t} & \mbox{for } 1 \le t \le \eta^\chi_1, \\
(\sigma \delta^f_{j, b})^{g_b} \cdot u_{jt} & 
\mbox{for } \eta^\chi_b+1 \le t \le \eta^\chi_{b+1} \mbox{ with } b \ge 1,
\end{array}\r. 
\nonumber
\\
& \eps_{it} = \l\{\begin{array}{ll}
v_{it} & \mbox{for } 1 \le t \le \eta^\eps_1, \\ 
(\sigma \delta^\eps_{i, b})^{g_b} \cdot v_{jt} & 
\mbox{for } \eta^\eps_b+1 \le t \le \eta^\eps_{b+1} \mbox{ with } b \ge 1,
\end{array}\r. 
\nonumber
\end{align}
where $\wh{\Var}(\cdot)$ denotes the sample variance operator,
$u_{jt}, v_{it} \sim_{\iid} t_7$ and $g_b = 1$ for even $b$ and $g_b = -1$ otherwise. 
We use the loadings estimated from the S\&P100 data without capping, denoted by $\wh\lambda_{ij}$,
and $\delta^f_{j, b}$ ($\delta^\eps_{i, b}$) is chosen from the same dataset
by contrasting pairs of intervals with visibly different $\wh{\Var}(\wh f_{jt})$ ($\wh{\Var}(\wh \eps_{it})$).
The change-points in the common components are introduced to $\Var(f_{jt})$
at $\eta^\chi_1 = [T/3]$ and $\eta^\chi_2 = [T/2]$,
and those in the idiosyncratic components to $\Var(\eps_{it})$ 
at $\eta^\eps_1 = [T/2]$ and $\eta^\eps_2 = [4T/5]$.
The magnitude of each change is controlled by $\sigma \in \{1, 0.75, 0.5, 0.25\}$, while
$\phi \in \{1, 1.5, 2, 2.5\}$ determines the ratio between the variance of the common and idiosyncratic components
(the larger $\phi$, the smaller $\Var(\chi_{it})/\Var(x_{it})$ is).  
We fix $T = 500$ and vary $n \in \{100, 300\}$.

We report the performance of our methodology over $100$ realisations when $n = 100$ 
in Figures \ref{sim:fig:m:one:100:one}--\ref{sim:fig:m:one:100:four} 
for the two extreme cases with $\phi \in \{1, 2.5\}$.
We also apply the DCBS algorithm to the WT of true common and idiosyncratic components generated under (M1)
and report the corresponding results,
which serve as a benchmark against which the efficacy of the PCA-based factor analysis
of the proposed methodology is assessed.
Additional simulation results with varying $\phi$ and $n$ 
are reported in the supplementary document.

\begin{figure}[t!]
\centering
\includegraphics[scale=.5]{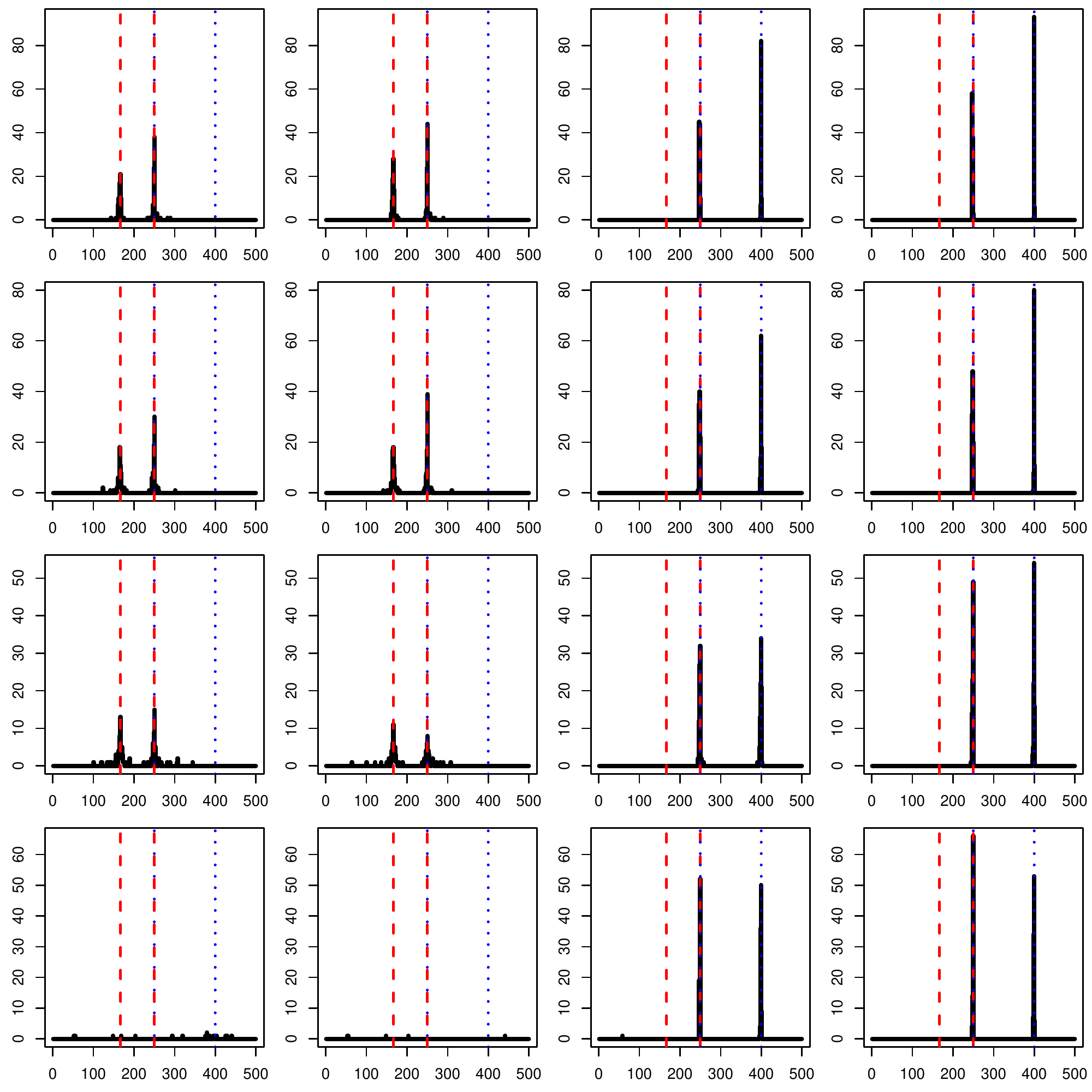}
\caption{\footnotesize (M1) Locations of the change-points estimated from $\wh{\chi}_{it}^*$, $\chi_{it}$ (oracle), $\wh{\eps}_{it}^*$ and $\eps_{it}$ (oracle)
by the DCBS algorithm (left to right) for $\sigma \in \{1, 0.75, 0.5, 0.25\}$ (top to bottom) when $n=100$, $T=500$ and $\phi=1$;
vertical lines indicate the locations of the true change-points $\eta^\chi_b, \, b=1, 2$ (dashed) and $\eta^\eps_b, \, b=1, 2$ (dotted); 
recall that $\eta^\chi_2 = \eta^\eps_1$.}
\label{sim:fig:m:one:100:one}
\end{figure}

\begin{figure}[t!]
\centering
\includegraphics[scale=.5]{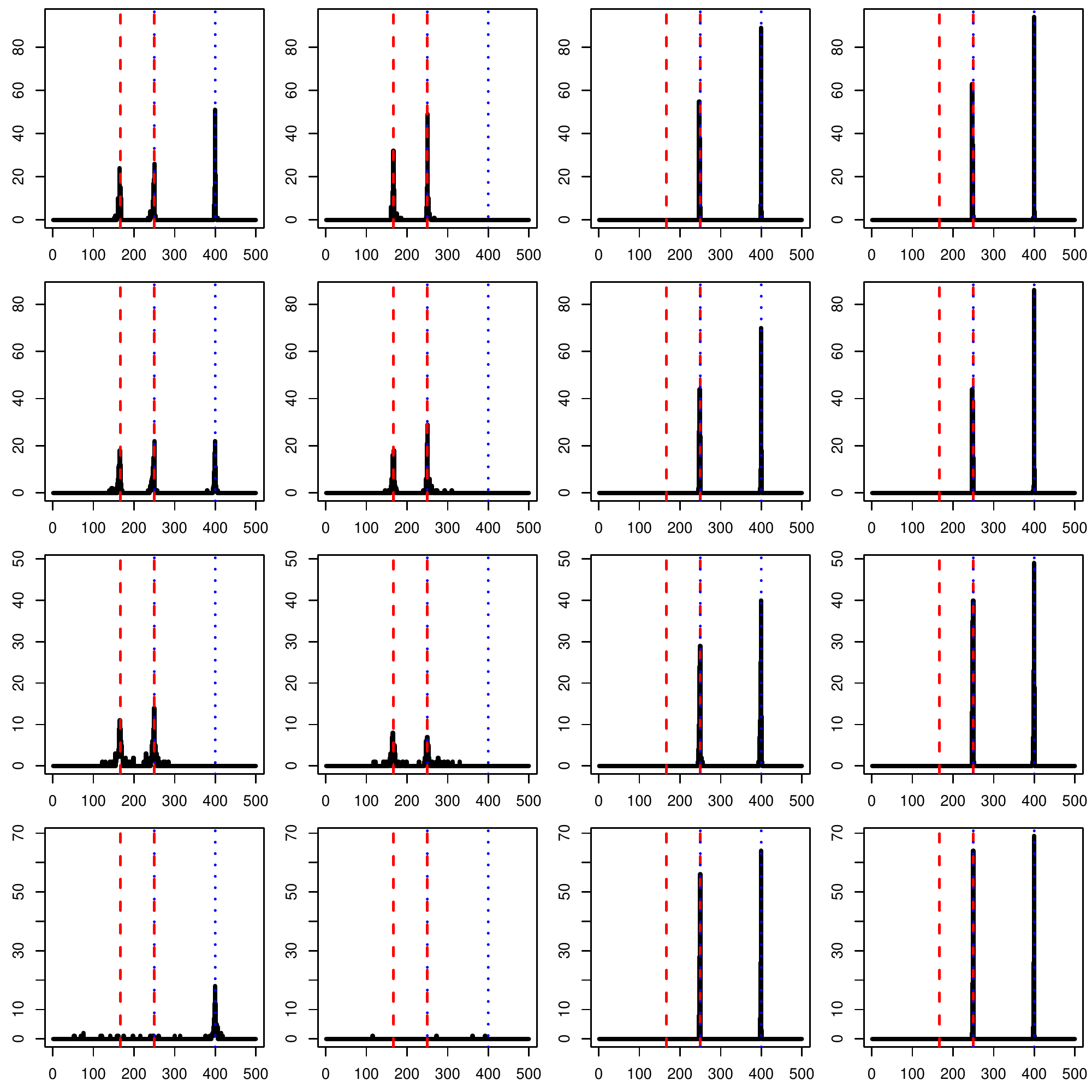}
\caption{\footnotesize (M1) Locations of the estimated change-points when $n=100$, $T=500$ and $\phi=2.5$.}
\label{sim:fig:m:one:100:four}
\end{figure}

The DCBS algorithm detects the two change-points in the idiosyncratic components 
equally well from $\wh \eps_{it}^*$ and $\eps_{it}$,
regardless of the values of $\phi$ and $\sigma$.
On the other hand, detection of $\eta^\chi_b, \, b=1, 2$ is highly variable with respect to these parameters.
When the size of the change is large ($\sigma \in \{1, 0.75\}$),
the panel data generated from transforming $\wh\chi_{it}^*$ serves as as good an input 
to the DCBS algorithm as that generated from transforming $\chi_{it}$
in terms of translating the presence and locations of both change-points.

With decreasing $\Var(\chi_{it})/\Var(x_{it})$, the change-point $\eta^\eps_2$,
which appears only in the idiosyncratic components,
is detected with increasing frequency as a change-point in the common components from $\wh \chi_{it}^*$, 
when (a) $\sigma \ge 0.75$ (the change in $\eps_{it}$ is large 
and thus all three change-points are detected from $\wh\chi_{it}^*$), or 
(b) $\sigma = 0.25$ (the changes in $\chi_{it}$ are ignored and 
a single change-point is detected at $\eta^\eps_2$ from $\wh\chi_{it}^*$).
This phenomenon is in line with the observation made for the model (S5) in Section \ref{sec:sim:single},
on the spillage of change-points in the idiosyncratic components over to the common components:
a significant co-movement in the dependence structure of $\eps_{it}$ 
may be regarded as being pervasive and common, 
and hence is captured as a change in the dependence structure of 
the common components by our proposed methodology. 

Lastly, we note that although we impose normality on the idiosyncratic components
for the theoretical development, these results show that
our methodology works well even when the data exhibits some deviations of normality, such as heavy-tails.

\subsubsection{Model (M2)}

In this model, change-points are introduced as in (S1)--(S4) of Section \ref{sec:sim:single}.
More specifically, 
\beas
\lambda_{ij, t} &=& \l\{\begin{array}{ll}
\lambda_{ij} & \mbox{for } 1 \le t \le \eta^\chi_1 = [T/3], \\
\lambda_{ij}+\Delta_{ij}\bbI(i \in \cS^\chi_1) & \mbox{for } \eta^\chi_1+1 \le t \le T, 
\end{array}
\r.
\\
f_{jt} &=& \l\{\begin{array}{ll}
\rho_{f, j} f_{j, t-1} + u_{jt}, & \mbox{for } 1 \le t \le \eta^\chi_2 = [T/2], \\
-\rho_{f, j} f_{j, t-1} + u_{jt}, & \mbox{for } \eta^\chi_2 = [T/2]+1 \le t \le T,
\end{array}
\r.
\\
\eps_{it} &=& \l\{\begin{array}{ll}
\rho_{\eps,i} \eps_{i, t-1} + v_{it} + \beta_i \sum_{\substack{|k| \le H \\ k \ne 0}} v_{i+k, t} & \mbox{for } 1 \le t \le \eta^\eps_1 = [3T/5], \\
\rho_{\eps,i}\{\bbI(i \notin \cS^\eps_1) - \bbI(i \in \cS^\eps_1)\}\eps_{i, t-1} + v_{it} + \beta_i \sum_{\substack{|k| \le H \\ k \ne 0}} v_{i+k, t} 
& \mbox{for } \eta^\eps_1+1 \le t \le T,
\end{array}
\r.
\\
x_{it} &=& \l\{\begin{array}{ll}
\sum_{j=1}^q \lambda_{ij, t}f_{jt} + \sqrt{\vartheta}\eps_{it} & \mbox{for } 1 \le t \le \eta^\chi_3 = [4T/5], \\
\sum_{j=1}^q \lambda_{ij, t}f_{jt} + \sqrt{2}\lambda_{i, q+1}f_{q+1, t}\bbI(i \in \cS^\chi_3) + \sqrt{\vartheta}\eps_{it} 
& \mbox{for } \eta^\chi_3+1 \le t \le T,
\end{array}
\r.
\eeas
with $q = 5$ and $\lambda_{ij}, u_{jt}, v_{it} \sim_{\iid} \cN(0, 1)$. 
The parameters $\rho_{f, j}$, $\rho_{\eps,i}$, $\beta_i$, $H$ and $\vartheta$ 
are chosen identically as in Section \ref{sec:sim:single}.
In summary, three change-points in the common components are introduced to the loadings ($\eta^\chi_1$),
autocorrelations of the factors ($\eta^\chi_2$) and the number of factors ($\eta^\chi_3$),
while a single change-point in the idiosyncratic components is introduced 
to their serial correlations ($\eta^\eps_1$).
The cardinality of the index sets $\cS^\chi_1$, $\cS^\chi_3$ and $\cS^\eps_1$ 
determines the sparsity of the change-points $\eta^\chi_1$, $\eta^\chi_3$ and $\eta^\eps_1$, respectively.
We randomly draw each index set from $\{1, \ldots, n\}$ 
with its cardinality set at $[\varrho n]$, where $\varrho \in \{1, 0.75, 0.5, 0.25\}$.
Also, in the case of $\eta^\chi_1$, the size of the shifts in the loadings is controlled by the parameter 
$\sigma \in \sqrt{2}\{1, 0.75, 0.5, 0.25\}$, as $\Delta_{ij} \sim_{\iid}\cN(0, \sigma^2)$.
Finally, we set $\phi \in \{1, 1.5, 2, 2.5\}$, a parameter that features in $\vartheta$,
in order to investigate the impact of $\Var(\chi_{it})/\Var(x_{it})$ 
on the performance of the change-point detection methodology. 
We fix the number of observations at $T = 500$ and the dimensionality at $n = 100$. 

We report the performance of our methodology over $100$ realisations 
in Figures \ref{sim:fig:m:two:100:one}--\ref{sim:fig:m:two:100:four} 
for the two extreme cases when $\phi \in \{1, 2.5\}$ with $\sigma = 0.75\sqrt{2}$.
Also, we include the results from applying the DCBS algorithm to the WT of 
true common and idiosyncratic components generated under (M2) as a benchmark case.
Additional simulation results with varying $\phi$ and $\sigma$ are reported in the supplementary document.

\begin{figure}[t!]
\centering
\includegraphics[scale=.5]{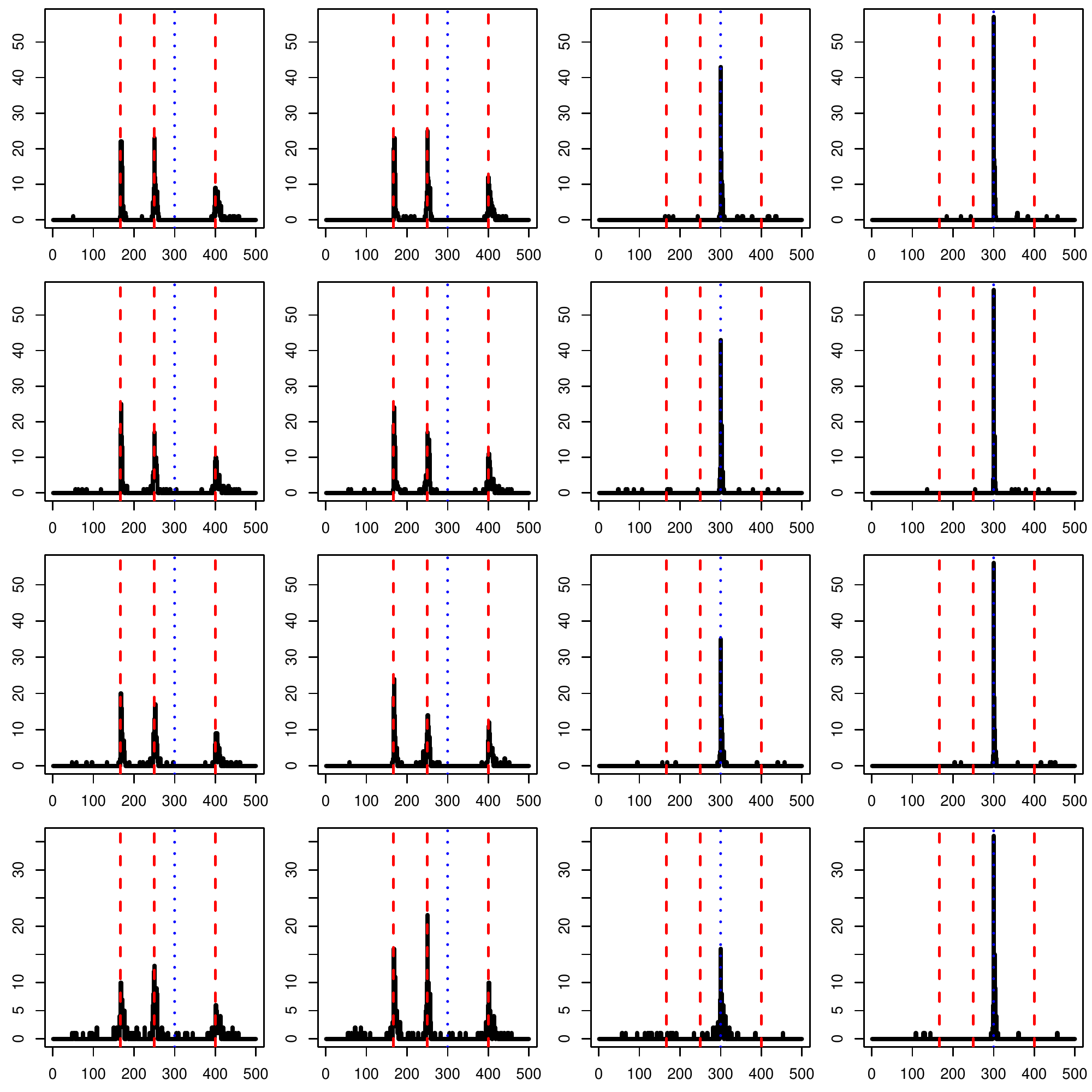}
\caption{\footnotesize(M2) Locations of the change-points estimated from $\wh{\chi}_{it}^*$, $\chi_{it}$ (oracle), $\wh{\eps}_{it}^*$ and $\eps_{it}$ (oracle)
by the DCBS algorithm (left to right) for $\varrho \in \{1, 0.75, 0.5, 0.25\}$ (top to bottom) 
when $n=100$, $T=500$, $\sigma = 0.75\sqrt{2}$ and $\phi=1$;
vertical lines indicate the locations of the true change-points $\eta^\chi_b, \, b=1, 2, 3$ (dashed) and $\eta^\eps_1$ (dotted).}
\label{sim:fig:m:two:100:one}
\end{figure}

\begin{figure}[t!]
\centering
\includegraphics[scale=.5]{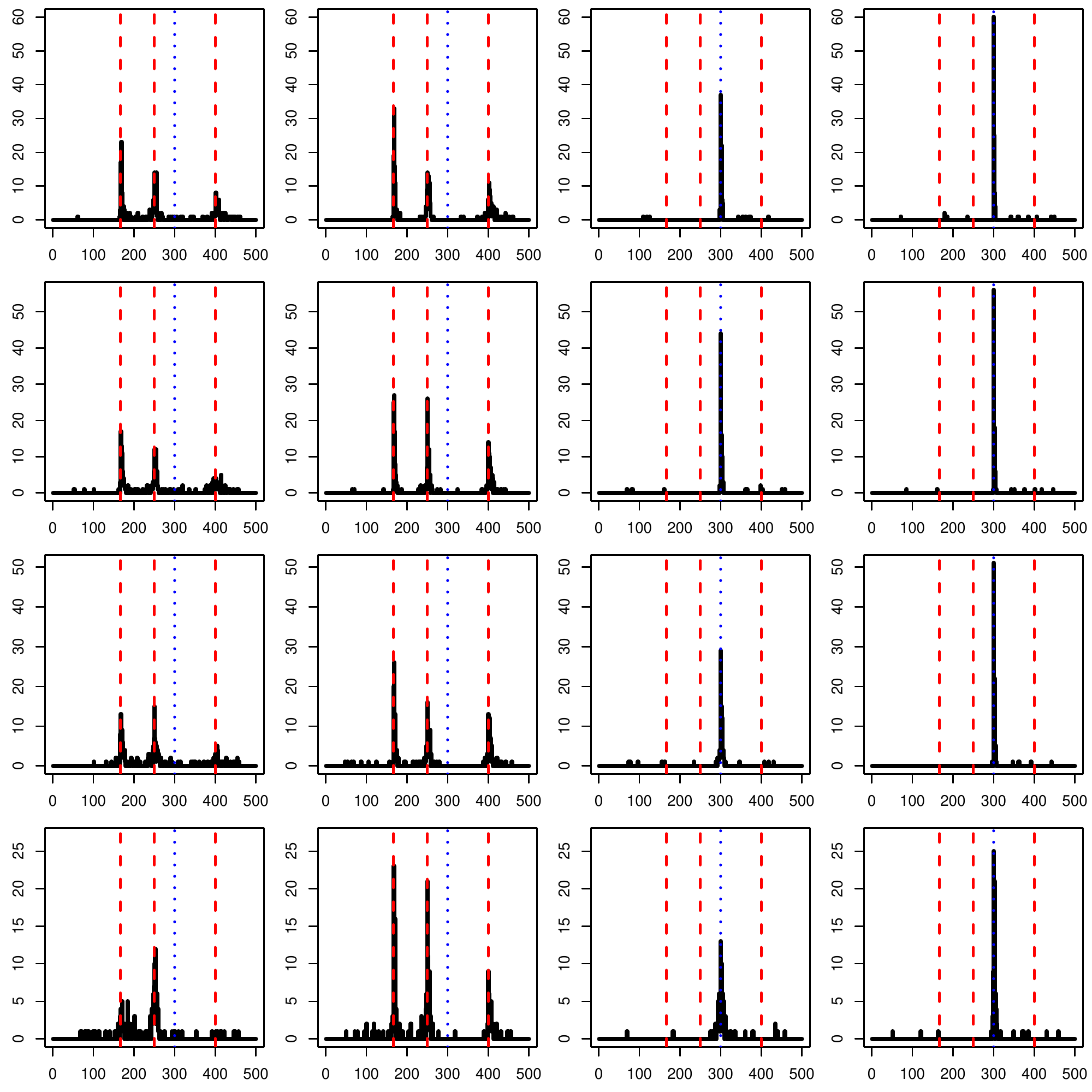}
\caption{\footnotesize(M2) Locations of the estimated change-points when $n=100$, $T=500$, $\sigma = 0.75\sqrt{2}$ and $\phi=2.5$.}
\label{sim:fig:m:two:100:four}
\end{figure}

In accordance with the observations made under single change-point scenarios (Section \ref{sec:sim:single}), 
detecting change-points in the common components, 
$\eta^\chi_1$ and $\eta^\chi_3$ in particular, becomes more challenging
as they grow sparse cross-sectionally (with decreasing $\varrho$)
and as $\Var(\chi_{it})/\Var(x_{it})$ decreases (with increasing $\phi$).
For the settings considered here, the DCBS algorithm applied to WT of $\wh\eps^*_{it}$
performs as well as that applied to the WT of the true $\eps_{it}$,
regardless of the model parameters $\phi$ and $\varrho$.
Not surprisingly, as the break in the loadings grows weaker with decreasing $\sigma$, 
the detection rate of $\eta^\chi_1$ deteriorates, especially when $\sigma = 0.25\sqrt{2}$.
Comparing the performance of the DCBS algorithm applied to $\wh\chi^*_{it}$ and $\chi_{it}$,
the gap is not so striking in the detection of the change-point $\eta^\chi_2$
in the autocorrelations of the factors. 
As for $\eta^\chi_1$ and $\eta^\chi_3$, provided that the breaks in the loadings
and the number of factors are moderately dense ($\varrho \ge 0.5$),
and the magnitude of the former reasonably large ($\sigma \ge 0.5\sqrt{2}$),
we can expect the common components estimated via PCA to recover the both change-points 
as those in their second-order structure, for a range of $\phi$.

\section{Real data analysis}
\label{sec:real}

\subsection{S\&P100 stock returns}
\label{sec:sp100}

In this section, we perform change-point analysis on
log returns of the daily closing values of the stocks composing the Standard and Poor's 100 (S\&P100) index,
observed between $4$ January $2000$ and $10$ August $2016$ ($n = 88$ and $T = 4177$). 
The dataset is available from Yahoo Finance.
With $\wh{r} = 4$ returned by the criterion in \eqref{eq:ic}, 
the set of factor number candidates is chosen as $\cR = \{4, \ldots, 20\}$.~Also, a constraint is imposed so that 
no two change-points are detected within the period of $20$ working days. 
The maximum number of change-points for the common components is attained 
with $k^* = 8$ ($\wh{B}_\chi = 15$), and 
we obtain $\wh{\cB}^\chi(k) \subset \wh{\cB}^\chi(k^*)$ for all $k \in \cR\setminus\{8\}$. 
Table \ref{table:sp:100} reports the change-points estimated from $\wh\chi_{it}^*$ and $\wh\eps_{it}^*$,
as well as their order of detection (represented by the level index $u$ 
of the nodes corresponding to the estimated change-points in the DCBS algorithm),
and Figure \ref{fig:sp100:ex} plots two representative daily log return series 
from the dataset along with $\wh\eta^\chi_b$.

Most of the change-points we find are in a neighbourhood of the events 
that characterise the financial market (some of which are not exactly dated). In particular:
\begin{enumerate}
\item the burst of the dot-com bubble which took place between March 2000 to October 2002;
\item the start of the second Iraq war in late March 2003;
\item Lehman Brothers bankruptcy in September 2008;
\item the first and second stages of the Greek and EU sovereign debt crisis in the summers of 2011 and 2015, respectively.
\end{enumerate}

\begin{table}[t!]
\caption{\footnotesize S\&P100 data: change-points estimated from the common and idiosyncratic components and their order of detection.}
\label{table:sp:100}
\centering
\scriptsize{
\begin{adjustwidth}{-.25cm}{-.25cm}
\begin{tabular}{c | cccccccc}
\hline\hline
\multirow{2}{*}{$\wh{\eta}^\chi_b$} 
& 43 & 89  & 197  & 331 & 613 & 656 & 816 & 1895
\\
& 06/03/2000 & 10/05/2000 & 12/10/2000 & 26/04/2001 & 14/06/2002 & 15/08/2002 & 04/04/2003 & 19/07/2007
\\
order & 3 & 2  & 5  & 4 & 3 & 4  & 1 & 2 \\
\cline{1-9}
\multirow{2}{*}{$\wh{\eta}^\chi_b$} & 2186 & 2249 & 2357 & 2397 & 2914 & 3020 & 3913
\\
& 12/09/2008 & 11/12/2008 & 19/05/2009 & 16/07/2009 & 03/08/2011 & 04/01/2012 & 24/07/2015
\\
order & 4 & 5 & 6 & 3 & 5 & 4 & 5
\\ \hline\hline
\multirow{2}{*}{$\wh{\eta}^\eps_b$} & 85 & 181 & 206 & 268 & 336 & 631 & 652 & 735
\\
& 04/05/2000 & 20/09/2000 & 25/10/2000 & 25/01/2001 & 03/05/2001 & 11/07/2002 & 09/08/2002 & 06/12/2002
 \\
order & 3 & 4 & 5 & 2 & 4 & 3 & 4 & 1
\\
 \cline{1-9}
\multirow{2}{*}{$\wh{\eta}^\eps_b$} & 914 & 1957 & 2184 & 2210 & 2253 & 2354 & 2537 & 3911
\\
& 25/08/2003 & 16/10/2007 & 10/09/2008 & 16/10/2008 & 17/12/2008 & 14/05/2009 & 04/02/2010 & 22/07/2015
\\
order & 4 & 5 & 3 & 5 & 6 & 4 & 2 & 3
\\
 \hline\hline   
\end{tabular}
\end{adjustwidth}
}
\end{table}

\begin{figure}[t!]
\centering
\includegraphics[scale=.6]{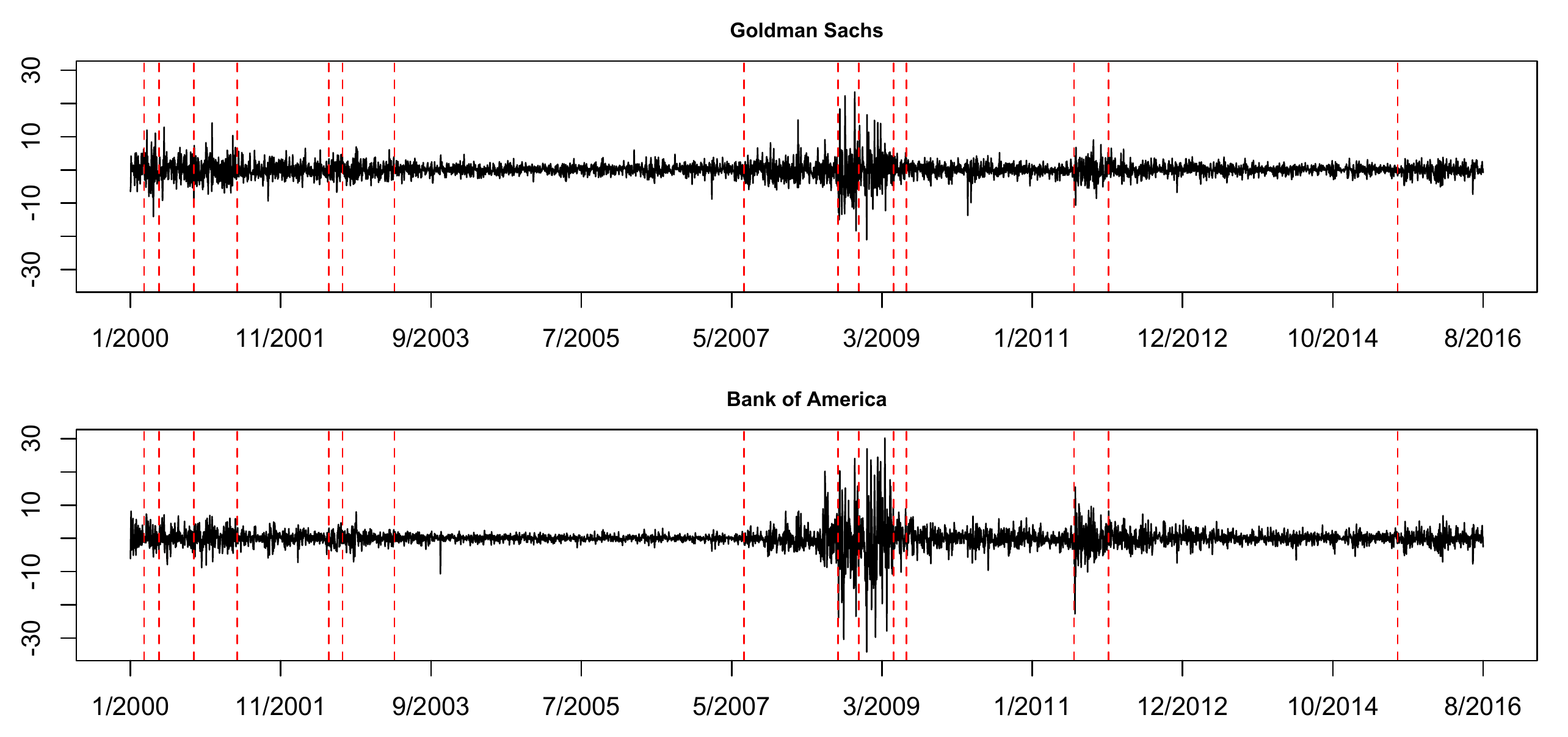}
\caption{\footnotesize S\&P100 data: daily log returns of Goldman Sachs (top) and Bank of America (bottom)
between $4$ January $2000$ and $10$ August $2016$, along with $\wh\eta^\chi_b, \, b = 1, \ldots, \wh{B}_\chi$
(vertical broken lines).}
\label{fig:sp100:ex}
\end{figure}

\begin{figure}[t!]
\centering
\includegraphics[scale=.5]{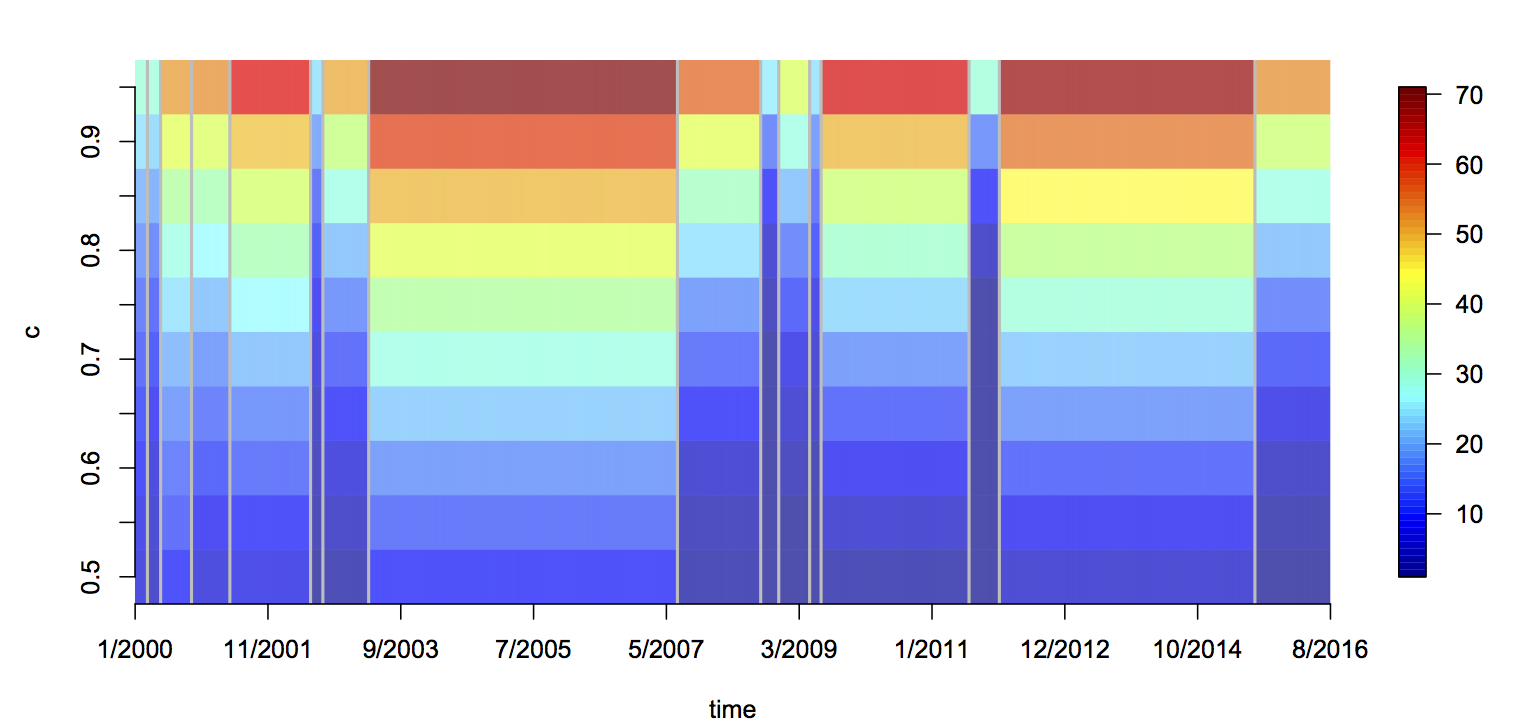}
\caption{\footnotesize S\&P100 data: $k_b(c)$ for each $[\wh\eta^\chi_b+1, \wh\eta^\chi_{b+1}]$, $b = 1, \ldots, \wh{B}_\chi$
according to the colour legend in the right;
$x$-axis denotes the time and $y$-axis denotes $c \in \{0.5, 0.55, \ldots, 0.95\}$.}
\label{fig:sp100:heat}
\end{figure}

By way of investigating the validity of $\wh\eta^\chi_b, \, b = 1, \ldots, \wh{B}_\chi$, 
we computed the following quantities over each segment defined by two neighbouring change-points,
$[\wh\eta^\chi_b+1, \wh\eta^\chi_{b+1}]$:
\beas
k_b(c) = \min\Big\{1 \le k \le q_b:\, \frac{\sum_{j = 1}^k \wh\mu^b_{x, j}}{\sum_{j = 1}^{q_b} \wh\mu^b_{x, j}} > c\Big\}
\mbox{ for some } c \in (0, 1),
\eeas
where $\wh\mu^b_{x, j}$ denotes the $j$th largest eigenvalue of $\wh{\bm\Gamma}^b_x$,
and $q_b = (n-1) \wedge (\wh\eta^\chi_{b+1}-\wh\eta^\chi_b - 1)$.
In short, $k_b(c)$ is the minimum number of eigenvalues required so that
the proportion of the variance of $\mbf x_t, \, t \in [\wh\eta^\chi_b+1, \wh\eta^\chi_{b+1}]$
accounted for by $\wh\mu^b_{x, j}, \, j = 1, \ldots, k_b(c)$ exceeds a given $c$.
Varying $c \in \{0.5, 0.55, \ldots, 0.95\}$, we plot $k_b(c)$ over the $\wh{B}_\chi+1$ segments in 
Figure \ref{fig:sp100:heat}.
We observe that over long stretches of stationarity, greater numbers of eigenvalues are required 
to account for the same proportion of variance, compared to shorter intervals which all tend to be characterised by high volatility. 
This finding is in accordance with the observation made in \cite{li2016},
that a small number of factors drive the majority of the cross-sectional correlations during
the periods of high volatility.

\subsection{US macroeconomic data}
\label{sec:macro}

We analyse the US representative macroeconomic dataset of $101$ time series, 
collected quarterly between $1960$:Q2 and $2012$:Q3 ($T = 210$), for change-points.
Similar datasets have been analysed frequently in the factor model literature, 
for example, in \cite{stockwatson02JASA}. 
The dataset is available from the St.~Louis Federal Reserve Bank website ({\url{https://fred.stlouisfed.org/}}). 
We impose a restriction in applying the DCBS algorithm so that 
no two change-points are detected within three quarters in analysing the quarterly observations.
Applying the information criterion of \cite{baing02},
$\wh{r} = 8$ is returned so we choose $\cR = \{8, \ldots, 20\}$.
All $k \ge 14$ lead to the identical change-point estimates for the common component 
with $\wh{B}_\chi = 5$ so that we select $k^* = 20$.
We also obtain $\wh{\cB}^\chi(k) \subset \wh{\cB}^\chi(k^*)$ for all $k < 14$.
Table \ref{table:macro:light} reports the change-points estimated from $\wh\chi^*_{it}$ and $\wh\eps^*_{it}$,
and we plot two representative series from the dataset, 
gross domestic product (GDP) growth rate and consumer price inflation (CPI), 
along with $\wh\eta^\chi_b$ in Figure \ref{fig:macro:ex}.

According to the change-points detected, the observations are divided into periods 
corresponding to different economic regimes characterised by high or low volatility. 
In particular, we highlight the following regimes 
(recessions are dated by the National Bureau of Economic Research, {\url {http://www.nber.org/cycles.html}}):
\ben
\item early 1970s to early 1980s marked by two major economic recessions,
which were characterised by high inflation due to the oil crisis and the level of interest rates;
\item the so-called Great Moderation period which, according to our analysis, started in late 1983 
and was characterised by low volatility of most economic indicators 
as a result of the implementation of new monetary policies, see also \cite{stockwatson03};
\item the period of the financial crisis that took place between 2007 and 2009 
and corresponds to the most recent (as of 2018) economic recession, 
with record low levels of GDP growth and inflation and associated high volatility;
\item the post-2009 years corresponding to the slow recovery of the US economy.
\een

\begin{table}[t!]
\caption{\footnotesize Macroeconomic data: change-points estimated from the common and idiosyncratic components
and their order of detection.}
\label{table:macro:light}
\centering
\scriptsize{
\begin{tabular}{c | cccccc}
\hline\hline
\multirow{2}{*}{$\wh{\eta}^\chi_b$} 
& 48 & 60 & 95 & 190 & 196
\\
&1972:Q1 & 1975:Q1 & 1983:Q4 & 2007:Q3 &2009:Q1
\\
order & 2 & 3 & 1 & 2 & 3 
\\ \hline\hline
\multirow{2}{*}{$\wh{\eta}^\eps_b$} 
& 49 & 58 & 92 & 189 & 194 & 200
\\
& 1972:Q2 & 1974:Q3 & 1983:Q1 &2007:Q2 & 2008:Q3 &2010:Q4
\\
order & 2 & 3 & 1 & 2 & 3 & 4
\\
 \hline\hline   
\end{tabular}}
\end{table}

\begin{figure}[t!]
\centering
\includegraphics[scale=.6]{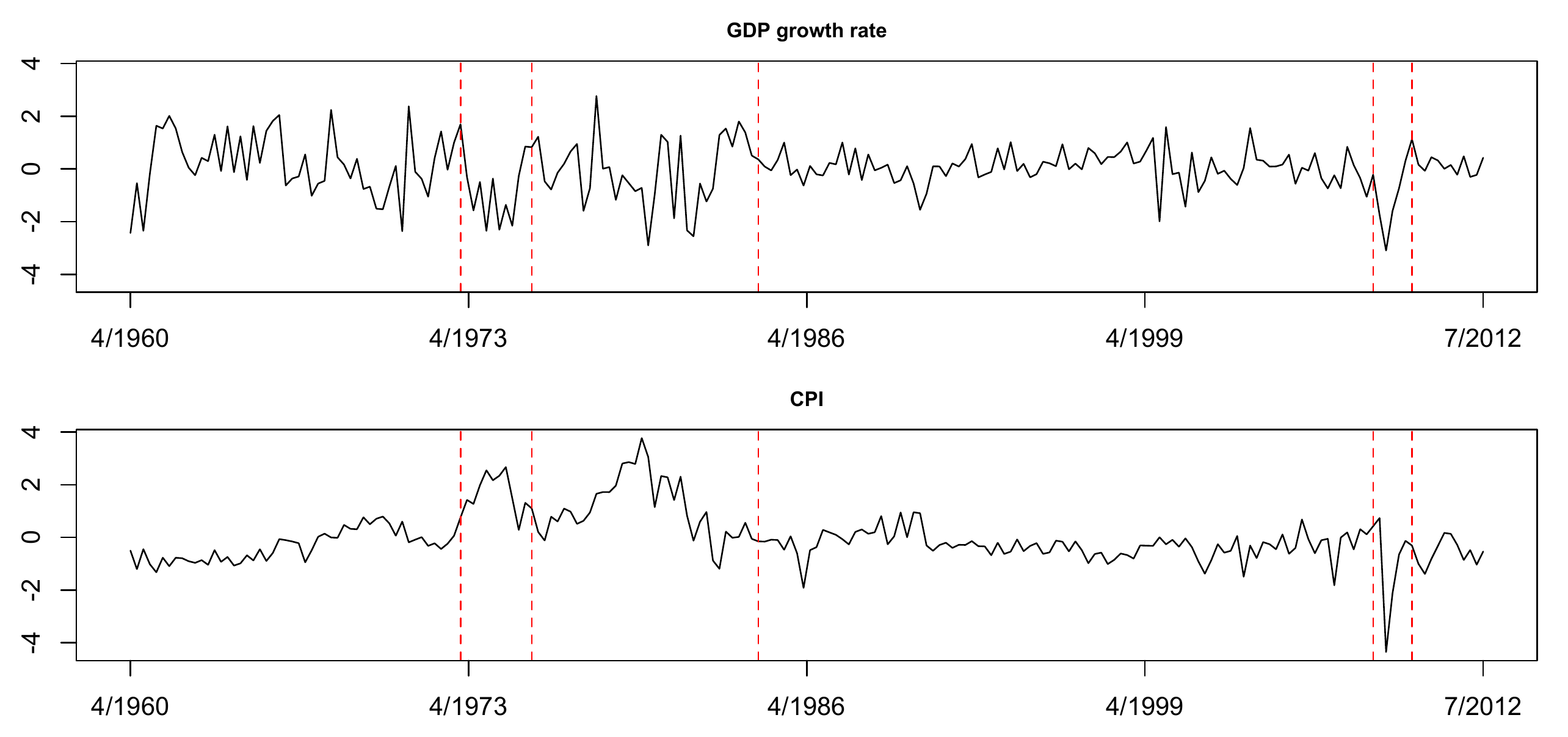}
\caption{\footnotesize Macroeconomic data: GDP growth rate (top) and CPI (bottom) 
between 1960:Q2 and 2012:Q3, along with $\wh\eta^\chi_b, \, b = 1, \ldots, \wh{B}_\chi$
(vertical broken lines).}
\label{fig:macro:ex}
\end{figure}

\begin{figure}[t!]
\centering
\includegraphics[scale=.5]{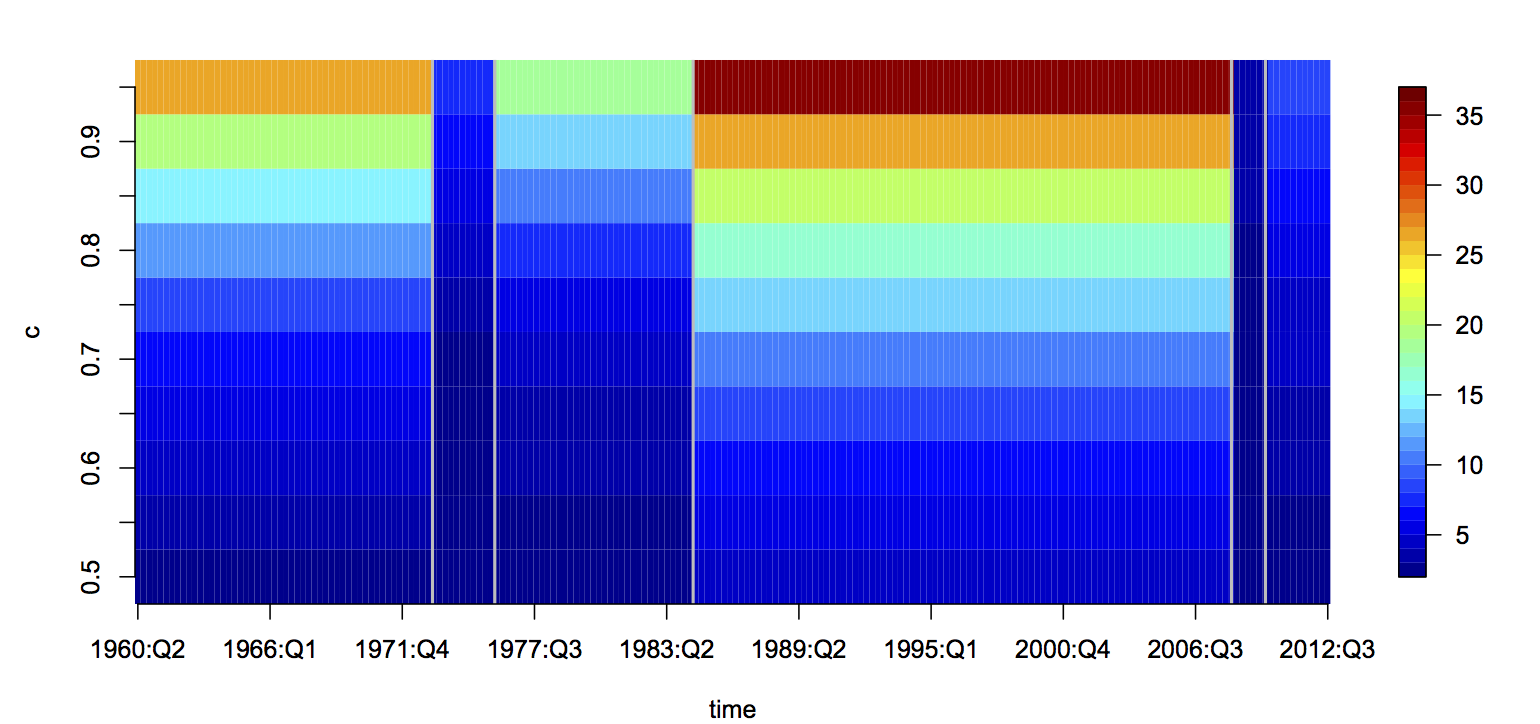}
\caption{\footnotesize Macroeconomic data: $k_b(c)$ for each $[\wh\eta^\chi_b+1, \wh\eta^\chi_{b+1}]$, $b = 1, \ldots, \wh{B}_\chi$
according to the colour legend in the right;
$x$-axis denotes the time and $y$-axis denotes $c \in \{0.5, 0.55, \ldots, 0.95\}$.}
\label{fig:macro:heat}
\end{figure}

\cite{cheng2016} performed change-point analysis on a similar set of macroeconomic and financial indicators,
observed monthly rather than quarterly, over a shorter span of period between January $1985$ and January $2013$.
Their focus was on verifying the existence of a structural break corresponding to the beginning of the financial crisis, 
which they estimated to be in December $2007$, 
a date which is close to our $\wh\eta^\chi_4$ (considering that we analyse quarterly observations).
As in Section \ref{sec:sp100}, we perform a post-change-point analysis 
by plotting $k_b(c)$ computed on each segment defined by $\wh\eta^\chi_b$, 
see Figure \ref{fig:macro:heat}, where we make similar observations about t
he contrast between the number of factors required 
over long stretches of stationarity and short intervals of volatility.

\section{Conclusions}

We have provided the first comprehensive treatment of 
high-dimensional time series factor models with multiple change-points 
in their second-order structure. 
We have proposed an estimation approach based on the capped PCA 
and wavelet transformations, 
first separating common and idiosyncratic components and then 
performing multiple change-point analysis on the levels of the transformed data. 
The number and locations of change-points are estimated consistently as $n, T \to \infty$
for both the common and idiosyncratic components. 
Our methodology is robust to the over-specification of the number of factors which, 
in the presence of multiple change-points, may not be accurately estimated 
by standard methods. 
Post change-point detection, we have proved the consistency of 
the common components estimated via PCA on each stationary segment.

An extensive numerical study has shown the good practical performance 
of our method and demonstrated that factor analysis prior to change-point detection 
improves the detectability of change-points. 
Two applications involving economic data have shown that we are able to pick up 
most of the structural changes in the economy, 
such as the recent financial crisis (2008--2009), 
economic recessions (mid 1970s and late 2000s)
or changes in the monetary policy regime 
(the start of the so-called Great Moderation in early 1908s). 
Our method is implemented in the R package  {\tt factorcpt}, available from CRAN.

\subsection*{Acknowledgements}

Haeran Cho's work was supported by the Engineering and Physical Sciences Research Council grant no. EP/N024435/1.
Piotr Fryzlewicz's work was supported by the Engineering and Physical Sciences Research Council grant no. EP/L014246/1.
We thank the Editor, Associate Editor and two referees for very helpful comments which
led to a substantial improvement of this paper.

\clearpage
\appendix

\section{An example of piecewise stationary factor model}
\label{sec:ex}

We illustrate with an example that the piecewise stationary factor model
\eqref{ps:fm} under Assumption \ref{assum:one} admits all possible change-point scenarios that arise 
under factor modelling for the common components.

Let $g_{1t}$ be an AR(1) process with white noise innovations $u_t \sim (0, 1)$, 
and let $g_{2t}$ be a white noise process.
Suppose that the common components, $\chi_{it}$, for all $i = 1, \ldots, n$, are written as 
\begin{align}
\chi_{it} = \l\{\begin{array}{ll}
a_i g_{1t} = a_i(\alpha g_{1, t-1} + u_t) & \mbox{for } \eta^\chi_0+1 = 1 \le t \le \eta^\chi_1, \\
a_i g_{1t} + b_i g_{2t} = a_i(\alpha g_{1, t-1} + u_t) + b_i g_{2t} & \mbox{for } \eta^\chi_1+1 \le t \le \eta^\chi_2, \\
c_i g_{1t} + d_i g_{2t} = c_i(\alpha g_{1, t-1} + u_t) + d_i g_{2t} & \mbox{for } \eta^\chi_2+1 \le t \le \eta^\chi_3, \\
c_i g_{1t} + d_i g_{2t} = c_i(\beta g_{1, t-1} + u_t) + d_i g_{2t} & \mbox{for } \eta^\chi_3+1 \le t \le \eta^\chi_4 = T, \\
\end{array}\r.
\label{eq:ex}
\end{align}
for some $|\alpha|, |\beta| < 1$.
In \eqref{eq:ex}, $\chi_{it}$ begins with a single factor $g_{1t}$ for $t \le \eta^\chi_1$, 
to which three change-points are introduced:
\begin{inparaenum}
\item [(a)] appearance of a new factor $g_{2t}$ at $t = \eta^\chi_1+1$,
\item [(b)] changes in both loadings at $t = \eta^\chi_2+1$, and
\item [(c)] changes in the autocorrelation structure of $g_{1t}$ at $t = \eta^\chi_3+1$.
\end{inparaenum}
We can re-write \eqref{eq:ex} into the piecewise stationary factor model in \eqref{ps:fm} 
with constant loadings $\bm\lambda_i = (a_i, b_i, c_i, d_i)^\top$, 
and a factor vector $\mbf f_t$ of dimension $r=4$ and defined over the four segments 
$[\eta^\chi_b+1, \eta^\chi_{b+1}], \, b = 0, \ldots, 3$ as
\bea
\mbf f_t = \l\{\begin{array}{ll}
(g_{1t}, 0, 0 ,0)^\top = (\alpha g_{1, t-1} + u_t, 0, 0 ,0)^\top & \mbox{for } 1 \le t \le \eta^\chi_1, \\ 
(g_{1t}, g_{2t}, 0 ,0)^\top = (\alpha g_{1, t-1} + u_t, g_{2t}, 0 ,0)^\top & \mbox{for } \eta^\chi_1+1 \le t \le \eta^\chi_2, \\
(0, 0, g_{1t}, g_{2t})^\top = (0, 0, \alpha g_{1, t-1} + u_t, g_{2t})^\top & \mbox{for } \eta^\chi_2+1 \le t \le \eta^\chi_3, \\
(0, 0, g_{1t}, g_{2t})^\top = (0, 0, \beta g_{1, t-1} + u_t, g_{2t})^\top & \mbox{for } \eta^\chi_3+1 \le t \le T. \\
\end{array}\r.
\label{eq:ex:one}
\eea
The following four comments help in understanding the properties of model \eqref{ps:fm}.

\ben
\item Note that $\mbf f_t$ in \eqref{eq:ex:one} meets condition Assumption \ref{assum:one} (i) with
\beas
\mbf f^0_t = (g^{(1)}_{1t}, 0, 0 ,0)^\top, \quad \mbf f^1_t = (g^{(1)}_{1t}, g_{2t}, 0 ,0)^\top, \quad
\mbf f^2_t = (0, 0, g^{(1)}_{1t}, g_{2t})^\top, \quad \mbf f^3_t = (0, 0, g^{(2)}_{1t}, g_{2t})^\top
\eeas
where $g^{(1)}_{1t} = \alpha g^{(1)}_{1, t-1} + u_t$ and 
$g^{(2)}_{1t} = \beta g^{(2)}_{1, t-1} + u_t$ are stationary AR($1$) processes. Indeed, it is clear that $\mbf f_t = \mbf f^b_t$ for $0 \le b \le 2$ and Assumption \ref{assum:one} (i) trivially holds in those segments.  Then, at $\eta \equiv \eta^\chi_2$, the AR parameter switches from $\alpha$ to $\beta$ and for $t \ge \eta+1$,
\beas
g_{1t} &=& \beta g_{1, t-1} + u_t = \beta(\beta g_{1, t-2} + u_{t-1}) + u_t = \cdots 
\\
&=& \beta^{t-\eta-1}g_{1, \eta+1} + \beta^{t-\eta-2}u_{\eta+2} + \ldots + \beta u_{t-1} + u_t,
\\
g^{(2)}_{1t} &=& \beta g^{(2)}_{1, t-1} + u_t  
= \beta^{t-\eta-1}g^{(2)}_{1, \eta+1} + \beta^{t-\eta-2}u_{\eta+2} + \ldots + \beta u_{t-1} + u_t.
\eeas
Since $|\alpha|, |\beta| < 1$, both $g_{1t}$ and $g^{(2)}_{1t}$ have finite variance. Therefore, by selecting $\rho_f = \beta^2 \in [0, 1)$, we have
\beas
\E(g_{1t} - g^{(2)}_{1t})^2 = \beta^{2(t-\eta-1)}\E(g_{1, \eta+1} - g^{(2)}_{1, \eta+1})^2
\le 2\beta^{2(t-\eta-1)}[\E(g_{1, \eta+1}^2) + \E\{(g^{(2)}_{1, \eta+1})^2\}]= O(\rho_f^{t-\eta}),
\eeas
and Assumption \ref{assum:one} (i) holds.


\item Recall that the number of factors in \eqref{ps:fm} satisfies $r \ge r_b$: in this example we have a single factor prior to $\eta^\chi_1$ ($r_0 = 1$) and then two factors in each of the subsequent segments ($r_1 = r_2 = r_3 = 2$), whereas $r = 4$.

\item The representation in \eqref{ps:fm} is not unique. We can for example re-write \eqref{eq:ex:one} with constant loadings $\bm\lambda_i = (a_i, b_i, c_i-a_i, d_i-b_i)^\top$ and $\mbf f_t = (g_{1t}, g_{2t}, g_{1t}, g_{2t})^\top$ for $\eta^\chi_2+1 \le t \le T$. 

\item Model \eqref{ps:fm} can also be written as a piecewise stationary version of the dynamic factor model 
introduced in \cite{fornilippi01} and \cite{hallinlippi13}. We can for example re-write \eqref{eq:ex:one} as the piecewise stationary version of a dynamic factor model with $r = 5$ and constant dynamic loadings $\bm\lambda_i(L) = (a_i(1-\alpha L)^{-1}, b_i, c_i(1-\alpha L)^{-1}, d_i, c_i(1-\beta L)^{-1})^\top$ ($L$ denoting the lag operator), and factors defined as
\beas
\mbf f^{0}_t = (u_{t}, 0, 0 ,0, 0)^\top, \quad
\mbf f^{1}_t = (u_{t}, g_{2t}, 0 ,0, 0)^\top, \quad
\mbf f^{2}_t =(0,0, u_{t}, g_{2t},0)^\top, \quad
\mbf f^{3}_t =(0,0,0, g_{2t},u_t)^\top.
\eeas
A change in the autocorrelation of the factors can therefore be equivalently represented by a change in the dynamic loadings. 
\een

\section{Proofs}
\label{sec:proofs}

\subsection{Preliminary results}

We denote by $\bm\varphi_i$ the $n$-dimensional vector with one as its $i$th element and zero elsewhere.

\begin{lem}\label{lem:cov} \hfill
\bit
\item [(i)] $n^{-1}\Vert \wh{\bm\Gamma}_x - \bm\Gamma_{\chi}\Vert =
O_p\Big(\sqrt{\frac{\log\,n}{T}} \vee \frac{1}{n}\Big)$;
\item [(ii)] $n^{-1/2}\Vert \bm\varphi_i^\top(\wh{\bm\Gamma}_x - \bm\Gamma_{\chi}) \Vert =
 O_p\l(\sqrt{\frac{\log\,n}{T}} \vee \frac{1}{\sqrt n}\r)$.
\eit
\end{lem}

\begin{proof}
Under Assumptions \ref{assum:two}--\ref{assum:five}, Lemmas A.3 and B.1 (ii) of \cite{fan2011b} show that
\bea
\max_{1 \le i, i' \le n} \l\vert \frac{1}{T}\sum_{t=1}^T x_{it}x_{i't} - 
\E\Big(\frac{1}{T}\sum_{t=1}^T x_{it}x_{i't}\Big)\r\vert
\le
\max_{1 \le j, j' \le r} r^2\bar{\lambda}^2 \l\vert \frac{1}{T}\sum_{t=1}^T f_{jt}f_{j't} - 
\E\Big(\frac{1}{T}\sum_{t=1}^T f_{jt}f_{j't}\Big)\r\vert
\nn \\
+
\max_{1 \le i, i' \le n} \l\vert \frac{1}{T}\sum_{t=1}^T \eps_{it}\eps_{i't} - 
\E\Big(\frac{1}{T}\sum_{t=1}^T \eps_{it}\eps_{i't}\Big)\r\vert
+
2\max_{\substack{1 \le j \le r \\ 1 \le i \le n}} r\bar{\lambda}
\l\vert \frac{1}{T}\sum_{t=1}^T f_{jt}\eps_{it}\r\vert
= O_p\l(\sqrt{\frac{\log\,n}{T}}\r).
\label{eq:consgamma0}
\eea
Therefore,
\beq\label{eq:consgamma}
\frac 1 n\Vert \wh{\bm\Gamma}_x - \bm\Gamma_{\chi}\Vert \leq  
\frac 1 n\Vert \wh{\bm\Gamma}_x - \bm\Gamma_{x}\Vert_F +\frac 1 n\Vert{\bm\Gamma}_\eps\Vert 
= O_p\l(\sqrt{\frac{\log\,n}{T}} \vee \frac{1}{n}\r),
\eeq
which follows from \eqref{eq:consgamma0}, Assumption \ref{assum:five} (i)
 and the observation under Assumption \ref{assum:four} (i) that 
$\mu_{\eps, 1}=\Vert \bm\Gamma_\eps\Vert<\infty$ for any $n$, see also the result in (C2). 
This proves part (i). 

For part (ii), now we deal with an $n$-dimensional vector and not a matrix. 
Hence, using the same approach as in part (i),
\beq
\frac 1{\sqrt n} \Vert \bm\varphi_i^\top(\wh{\bm\Gamma}_x - \bm\Gamma_{\chi}) \Vert\le 
\frac 1 {\sqrt n}\Vert\bm\varphi_i^\top( \wh{\bm\Gamma}_x - \bm\Gamma_{x})\Vert +\frac 1 {\sqrt n}\Vert{\bm\Gamma}_\eps\Vert  = O_p\l(\sqrt{\frac{\log\,n}{T}} \vee \frac{1}{\sqrt n}\r),\nn
\eeq
which completes the proof.
\end{proof}

\begin{lem}
\label{lem:evals}
{\it Let $r \times r$ diagonal matrices $\wh{\mbf M}_x$ and $\mbf M_{\chi}$ have
the $r$ largest eigenvalues of $\wh{\bm\Gamma}_x$ and of $\bm\Gamma_{\chi}$ in the decreasing order
as the diagonal elements, respectively. Then,
\beas
\Big\Vert\Big(\frac{\wh{\mbf M}_{x}}{n}\Big)^{-1}-\Big(\frac{\mbf M_{\chi}}{n}\Big)^{-1}\Big\Vert
= O_p\l(\sqrt{\frac{\log\,n}{T}} \vee \frac{1}{n}\r).
\eeas}
\end{lem}

\begin{proof}
As a consequence of Lemma \ref{lem:cov} (i) and Weyl's inequality, $\wh\mu_{x, j}$ satisfy
\beq\label{eq:eval}
\frac 1 n|\wh\mu_{x, j}-\mu_{\chi, j}|\leq \frac 1 n\Vert \wh{\bm\Gamma}_x-\bm\Gamma_{\chi}\Vert 
= O_p\l(\sqrt{\frac{\log\,n}{T}} \vee \frac{1}{n}\r), \quad j=1, \ldots, r.
\eeq
From \eqref{eq:eval}, there exists $\underline c_r \in (0, \infty)$ such that 
$\mu_{\chi, r}/n \ge \underline c_r$ and thus $\wh{\mu}_{x, r}/n \ge \underline c_r + O_p(\sqrt{\frac{\log n}{T}} \vee \frac{1}{n})$,
which implies that the matrix $n^{-1}{\mbf M_{\chi}}$ is invertible and the inverse of $n^{-1}{\wh{\mbf M}_{x}}$ 
exists with probability tending to one as $n,T\to\infty$. Therefore,
\beq
\Big\Vert\Big(\frac{\mbf M_{\chi}}{n}\Big)^{-1}\Big\Vert = \frac{n}{\mu_{\chi, r}} = O(1),
\qquad\Big\Vert\Big(\frac{\wh{\mbf M}_{x}}{n}\Big)^{-1}\Big\Vert = \frac{n}{\wh{\mu}_{x, r}} = O_p(1).\nn
\eeq
Finally, 
\begin{align}
& \Big\Vert\Big(\frac{\wh{\mbf M}_{x}}{n}\Big)^{-1 } - \Big(\frac{\mbf M_{\chi}}{n}\Big)^{-1}\Big\Vert 
\le \sqrt{\sum_{j=1}^r\Big(\frac{n}{\wh{\mu}_{x, j}}-\frac{n}{\mu_{\chi, j}}\Big)^2} 
\leq \sum_{j=1}^r n\Big|\frac{\wh{\mu}_{x, j}-\mu_{\chi, j}}{\wh{\mu}_{x, j}\mu_{\chi, j}}\Big| \nn
\\
\leq& \frac{r\max_{1 \le j \le r}|\wh{\mu}_{x, j}-\mu_{\chi, j}|}
{n\underline c_r^2+O_p\Big(n\sqrt{\frac{\log\,n}{T}} \vee 1\Big)}
= O_p\l(\sqrt{\frac{\log\,n}{T}} \vee \frac{1}{n}\r). \nn
\end{align}
\end{proof}

\begin{lem}
\label{lem:evecs}
{\it Denote the $n$-dimensional normalised eigenvectors corresponding to 
the $j$th largest eigenvalues of $\wh{\bm\Gamma}_x$ and ${\bm\Gamma}_{\chi}$,
by $\wh{\mbf w}_{x, j}$ and $\mbf w_{\chi, j}$, respectively.  
We further define the $n \times r$ matrices $\wh{\mbf W}_x = [\wh{\mbf w}_{x,1}, \ldots, \wh{\mbf w}_{x,r}]$ and 
$\mbf W_{\chi} = [\mbf w_{\chi,1}, \ldots, \mbf w_{\chi,r}]$. 
Then, there exists an orthonormal $r\times r$-matrix $\mbf S$ such that
\bit
\item [(i)] $\Vert \wh{\mbf W}_{x} - \mbf W_{\chi}\mbf S \Vert = 
O_p\l(\sqrt{\frac{\log\,n}{T}} \vee \frac{1}{n}\r)$;
\item [(ii)] $\sqrt n\,\Vert\bm\varphi_i^\top(\wh{\mbf W}_x-\mbf W_{\chi}\mbf S)\Vert
= O_p\l(\sqrt{\frac{\log\,n}{T}} \vee \frac{1}{\sqrt n}\r)$.
\eit}
\end{lem}

\begin{proof}
From Theorem 2 in \cite{yu15}, which is a generalisation of the $\sin\theta$ theorem in \cite{dk70}, we have
\begin{align}\label{eq:dk1}
\Vert \wh{\mbf W}_{x}-  \mbf W_{\chi}\mbf S \Vert \le 
\frac{2^{3/2}\sqrt{r}\Vert\wh{\bm\Gamma}_x-\bm\Gamma_{\chi}\Vert}
{\min\big(\mu_{\chi, 0}-\mu_{\chi, 1},\mu_{\chi, r}-\mu_{\chi, r+1}\big)}, 
\end{align}
where $\mu_{\chi, 0} = \infty$ and $\mu_{\chi, r+1} = 0$. 
From \eqref{eq:eval}, the denominator of \eqref{eq:dk1} is bounded from the below by $\underline c_r n$
and thus part (i) follows immediately from Lemma \ref{lem:cov} (i).

For part (ii), evoking Corollary 1 of \cite{yu15} and noticing that $\bm\Gamma_{\chi}$ has distinct eigenvalues given in (C1), we can further show that $\mbf S$ is a diagonal matrix with entries $\pm 1$. Then, using part (i) above, Lemmas \ref{lem:cov} (ii) and \ref{lem:evals}, and the fact that 
$\Vert\mbf W_{\chi}\mbf S\Vert=1$ and $\Vert\bm\varphi_i^{\top}\bm\Gamma_{\chi}\Vert=O(\sqrt n)$, we have
\begin{align}
&\sqrt n\Vert\bm\varphi_i^\top(\wh{\mbf W}_{x}- \mbf W_{\chi}\mbf S)\Vert=
\frac 1 {\sqrt n}\Big\Vert
 {\bm\varphi_i^\top} \Big\{\wh{\bm\Gamma}_{x}\wh{\mbf W}_{x}\Big(\frac{\wh{\mbf M}_{x}}{n}\Big)^{-1}
-\bm\Gamma_{\chi}\mbf W_{\chi}\mbf S\Big(\frac{\mbf M_{\chi}}{n}\Big)^{-1}\Big\}
\Big\Vert\nn\\
\leq& \frac 1 {\sqrt n}\Vert
\bm\varphi_i^\top\big(\wh{\bm\Gamma}_{x}-\bm\Gamma_{\chi}\big)\Vert\,
\Big\Vert\Big(\frac{\mbf M_{\chi}}{n}\Big)^{-1}\Big\Vert+
\frac 1 {\sqrt n}\Vert
\bm\varphi_i^\top\bm\Gamma_{\chi}\Vert\,\Big\Vert\Big(\frac{\wh{\mbf M}_{x}}{n}\Big)^{-1}-\Big(\frac{\mbf M_{\chi}}{n}\Big)^{-1}\Big\Vert\nn\\ 
+ &\frac 1 {\sqrt n}\Vert
\bm\varphi_i^\top\bm\Gamma_{\chi}\Vert\,
\Big\Vert\Big(\frac{\mbf M_{\chi}}{n}\Big)^{-1}\Big\Vert\,\Vert
\wh{\mbf W}_{x}-\mbf W_{\chi}\mbf S\big\Vert + o_p\l( \sqrt{\frac{\log\,n}{T}} \vee \frac 1 {\sqrt n}\r)
= O_p\l( \sqrt{\frac{\log\,n}{T}} \vee \frac{1}{\sqrt n}\r).\nn
\end{align} 
\end{proof}

\begin{lem}
\label{lem:x:n}
{\it For a fixed $\theta \ge 1+(\beta_f^{-1} \vee 1/2)$, we have
\beas
\max_{1 \le s \le e \le T}
\frac{1}{\sqrt{e-s+1}} \Big\Vert\sum_{t=s}^e \mbf x_t \Big\Vert = O_p(\sqrt{n}\log^\theta T).
\eeas}
\end{lem}

\begin{proof}
Note that
\beas
\frac{1}{\sqrt{e-s+1}} \Big\Vert\sum_{t=s}^e \mbf x_t \Big\Vert
\le 
\frac{1}{\sqrt{e-s+1}} \Big\Vert\sum_{t=s}^e \bm\chi_t \Big\Vert +
\frac{1}{\sqrt{e-s+1}} \Big\Vert\sum_{t=s}^e \bm\eps_t \Big\Vert,
\eeas
where
\beas
\max_{1 \le s \le e \le T} \frac{1}{\sqrt{e-s+1}} \Big\Vert\sum_{t=s}^e \bm\chi_t \Big\Vert
\le 
r\bar{\lambda}\sqrt{n} \max_{1 \le j \le r} \max_{1 \le s \le e \le T} \frac{1}{\sqrt{e-s+1}}
\Big\vert \sum_{t=s}^e f_{jt} \Big\vert,
\eeas
and
\beas
\max_{1 \le s \le e \le T} \frac{1}{\sqrt{e-s+1}} \Big\Vert\sum_{t=s}^e \bm\eps_t \Big\Vert
\le 
\sqrt{n} \max_{1 \le i \le n} \max_{1 \le s \le e \le T} \frac{1}{\sqrt{e-s+1}}
\Big\vert \sum_{t=s}^e \eps_{it} \Big\vert.
\eeas
When $\sqrt{e-s+1} \le \log\,T$, under Assumptions \ref{assum:two} (ii), \ref{assum:four} (ii) and \ref{assum:seven},
\beas
\max_{1 \le j \le r} \max_{1 \le s \le e \le T}\frac{1}{\sqrt{e-s+1}}
\Big\vert \sum_{t=s}^e f_{jt} \Big\vert \le \log\,T \max_{1 \le j \le r}\max_{1 \le t \le T}|f_{jt}| 
&=& O_p(\log^{1+1/\beta_f} T),
\\
\max_{1 \le i \le n} \max_{1 \le s \le e \le T}\frac{1}{\sqrt{e-s+1}}
\Big\vert \sum_{t=s}^e \eps_{it} \Big\vert \le \log\,T \max_{1 \le i \le n}\max_{1 \le t \le T}|\eps_{jt}| 
&=& O_p(\log^{3/2} T).
\eeas
When $\sqrt{e-s+1} > \log\,T$, under Assumptions \ref{assum:two} (ii), \ref{assum:four} (ii) and \ref{assum:five} (ii), 
the exponential inequality given in Theorem 1.4 of \cite{bosq1998} is applicable.
More specifically, setting $q = \lfloor (e-s+1)/(C_0\log\,T) \rfloor$ and $k=3$ in the statement of that theorem, 
we have
\beas
\p\l(\frac{1}{\sqrt{e-s+1}} \Big\vert \sum_{t=s}^e f_{jt}\Big\vert > C_1\log\,T \r)
\le C_2\log\,T\exp\Big(-\frac{C_3C_1^2}{C_0}\log\,T\Big) + 
\frac{C_4T^{3/2}}{\log\,T}\exp\{-C_5(C_0\log^\beta T)^{6/7}\}
\eeas
and similarly, 
\beas
\p\l(\frac{1}{\sqrt{e-s+1}} \Big\vert \sum_{t=s}^e \eps_{jt}\Big\vert > C_1\log\,T \r)
\le C'_2\log\,T\exp\Big(-\frac{C'_3C_1^2}{C_0}\log\,T\Big) + 
\frac{C'_4T^{3/2}}{\log\,T}\exp\{-C'_5(C_0\log^\beta T)^{6/7}\}
\eeas
for some fixed $C_0, C_1, C_2, \ldots, C_5, C_2', \ldots, C_5' > 0$ dependent on 
$\beta_f, c_f, \beta, c_\alpha$ and $\kappa$.
Applying Bonferroni correction,
\begin{align*}
& \p\l(\max_{1 \le j \le r}\max_{\substack{1 \le s \le e \le T \\ \sqrt{e-s+1} > \log\,T}} \frac{1}{\sqrt{e-s+1}} 
\Big\vert \sum_{t=s}^e f_{jt} \Big\vert > C_1\log\,T \r)
\\
& \le rT^2[C_2\log\,T\exp(-C_3C_0^{-1}C_1^2\log\,T) + 
C_4T^{3/2}/(\log\,T)\exp\{-C_5(C_0\log^\beta T)^{6/7}\}] \to 0,
\\
& \p\l(\max_{1 \le i \le n}\max_{\substack{1 \le s \le e \le T \\ \sqrt{e-s+1} > \log\,T}} \frac{1}{\sqrt{e-s+1}} 
\Big\vert \sum_{t=s}^e \eps_{it} \Big\vert > C_1\log\,T \r)
\\
& \le nT^2[C'_2\log\,T\exp(-C'_3C_0^{-1}C_1^2\log\,T) + 
C'_4T^{3/2}/(\log\,T)\exp\{-C'_5(C_0\log^\beta T)^{6/7}\}] \to 0
\end{align*}
as $T \to \infty$ for large enough $C_1$, which completes the proof. 
\end{proof}

\subsection{Proof of Theorem \ref{thm:common}}

Note that under the normalisation adopted in Assumption \ref{assum:two} (i),
$\wh{\chi}_{it} = \bm\varphi_i^{\top}\wh{\mbf W}_x\wh {\mbf W}_x^\top\mbf x_t$ and
$\chi_{it} 
= \bm\varphi_i^{\top} \mbf W_\chi \mbf W_\chi^\top\bm \chi_t$.
Hence,
\begin{align}
\max_{1 \le i \le n}\max_{1 \le t \le T} \vert \wh{\chi}_{it} - \chi_{it} \vert 
\le &
\max_{1 \le i \le n}\max_{1 \le t \le T} 
\vert \bm\varphi_i^{\top}\wh{\mbf W}_x\wh {\mbf W}_x^\top\mbf x_t - 
\bm\varphi_i^{\top}\mbf W_\chi\mbf W_\chi^\top\mbf x_t\vert 
\nn \\
& + 
\max_{1 \le i \le n}\max_{1 \le t \le T} 
\vert \bm\varphi_i^{\top}\mbf W_\chi\mbf W_\chi^\top\bm\eps_t\vert
= I + II.
\label{eq:chi1}
\end{align}
For $I$, we have 
\begin{align}
I \le &
\max_{1 \le i \le n}\max_{1 \le t \le T} 
\vert \bm\varphi_i^{\top}\wh{\mbf W}_x\wh {\mbf W}_x^\top\mbf x_t -
\bm\varphi_i^{\top}\mbf W_\chi\mbf S\wh{\mbf W}_x^\top\mbf x_t\vert
+ 
\max_{1 \le i \le n}\max_{1 \le t \le T} 
\vert \bm\varphi_i^{\top}\mbf W_\chi\mbf S\wh {\mbf W}_x^\top\mbf x_t -
\bm\varphi_i^{\top}\mbf W_\chi\mbf W_\chi^\top\mbf x_t\vert 
\nn \\
\le & \max_{1 \le i \le n} \Vert\bm\varphi_i^\top(\wh{\mbf W}_x - \mbf W_\chi\mbf S)\Vert\,
\Vert \wh{\mbf W}_x\Vert \, \max_{1 \le t \le T} \Vert \mbf x_t\Vert
+
\max_{1 \le i \le n} \Vert\bm\varphi_i^\top\mbf W_\chi\Vert \,
\Vert\wh{\mbf W}_x\mbf S-\mbf W_\chi \Vert \, \max_{1 \le t \le T} \Vert \mbf x_t\Vert 
\nn \\
= &  O_p\l\{\l(\sqrt{\frac{\log\,n}{T}} \vee \frac{1}{\sqrt n}\r)\log^\theta T\r\},
\label{eq:chi3}
\end{align}
from Lemma \ref{lem:evecs} (i)--(ii),
Lemma \ref{lem:x:n} and the result in (C1) lead to
\bea
\label{eq:evec:size}
\max_{1 \le i \le n} \Vert\bm\varphi_i^\top\mbf W_\chi\Vert \le
\max_{1 \le i \le n} \Vert\bm\varphi_i^\top\bm\Gamma_{\chi}\Vert\, \Vert \mbf W_\chi\Vert \, 
\Vert \mbf M_\chi^{-1}\Vert
= O\l(\frac{1}{\sqrt n}\r).
\eea

As for $II$, due to normalisation of the eigenvectors, we invoke Assumption \ref{assum:four} (i):
\begin{align}
\label{eq:chi4}
\E(\Vert{\mbf W}_{\chi}^{\top}\bm\eps_t\Vert^2) = \sum_{j=1}^r \E\{ (\mbf w_{\chi, j}^\top \bm\eps_t)^2 \} 
= \sum_{j=1}^r\sum_{i, i'=1}^n w_{\chi, ij} w_{\chi,i'j}\E(  \eps_{it}\eps_{i't} ) < rC_\eps.
\end{align}
Thereby,  due to Assumption \ref{assum:four} (ii) and Bonferroni correction,
$\max_{1 \le t \le T} \Vert{\mbf W}_{\chi}^{\top}\bm\eps_t\Vert = O_p(\sqrt{\log\,T})$ and using \eqref{eq:evec:size}
\begin{align}
\max_{1 \le i \le n} \max_{1 \le t \le T} \vert \bm\varphi_i^{\top}\mbf W_\chi\mbf W_\chi^\top\bm\eps_t\vert
\le 
\max_{1 \le i \le n} \Vert \bm\varphi_i^{\top}\mbf W_\chi\Vert \,
\max_{1 \le t \le T}\Vert{\mbf W}_{\chi}^{\top}\bm\eps_t\Vert = O_p\l(\sqrt{\frac{\log\,T}{n}}\r).
\label{eq:chi5}
\end{align}
Substituting \eqref{eq:chi3} and \eqref{eq:chi5} into \eqref{eq:chi1} completes the proof.\hfill $\Box$

\subsection{Proof of Theorem \ref{thm:overestimation}}
\label{sec:thm:overestimation}

Note that
\beas
\wh{\chi}_{it}^k =\sum_{j=1}^r \wt{w}_{x,ij} \wt{\mbf w}_{x, j}^{\top}\mbf x_t
+
\sum_{j=r+1}^k \wt{w}_{x,ij} \wt{\mbf w}_{x, j}^{\top} \mbf x_t.
\eeas
Recall \eqref{eq:evec:size}, from which 
\beas
\max_{1 \le i \le n} \Vert \bm\varphi_i^\top\wh{\mbf W}_x \Vert
\le \max_{1 \le i \le n} \Vert \bm\varphi_i^\top\mbf W_\chi \Vert + 
\max_{1 \le i \le n} \Vert \bm\varphi_i^\top(\wh{\mbf W}_x - \mbf{W}_\chi\mbf S) \Vert 
= O_p\Big(\frac{1}{\sqrt n}\Big)
\eeas
due to Lemma \ref{lem:evecs} (ii),
and therefore $\max_{1 \le j \le r}\max_{1 \le i \le n} \wh{w}_{x, ij} = O_p(1/\sqrt{n})$.
Therefore, for $c_w$ chosen large enough, we have 
$\wh\chi_{it}^r=\sum_{j=1}^r \wt{w}_{x,ij} \wt{\mbf w}_{x, j}^{\top}\mbf x_t = 
\sum_{j=1}^r \wh{w}_{x,ij} \wh{\mbf w}_{x, j}^{\top}\mbf x_t = \wh\chi_{it}$ for all $i$ and $t$ 
with probability tending to one,
and we prove this theorem conditioning on such an event;
once this is done, it then implies the unconditional arguments.

Consider the scaled partial sums
\begin{align}
& \max_{1 \le i \le n} \max_{1 \le s \le e \le T} 
\frac 1 {\sqrt {e-s+1}} \Big\vert\sum_{t=s}^e (\wh{\chi}^k_{it}-\chi_{it})\Big\vert  
\nn \\
\le & \max_{1 \le i \le n} \max_{1 \le s \le e \le T} 
\frac 1 {\sqrt {e-s+1}} \Big\vert\sum_{t=s}^e (\wh{\chi}^k_{it}-\wh{\chi}_{it})\Big\vert +
\max_{1 \le i \le n} \max_{1 \le s \le e \le T} 
\frac 1 {\sqrt {e-s+1}} \Big\vert\sum_{t=s}^e (\wh{\chi}_{it}-\chi_{it})\Big\vert
= I + II.
\label{eq:psums1}
\end{align}
Starting from $II$,
\begin{align*}
II &=
\max_{1 \le i \le n} \max_{1 \le s \le e \le T} \frac 1 {\sqrt{e-s+1}} \Big\vert \sum_{t=s}^e 
\bm\varphi_i^{\top}(\wh{\mbf W}_x\wh {\mbf W}_x^\top\mbf x_t - 
\mbf W_\chi\mbf W_\chi^\top\bm\chi_t)\Big\vert
\\
\le& 
\max_{1 \le i \le n} \max_{1 \le s \le e \le T} \frac 1 {\sqrt{e-s+1}} \Big\vert \sum_{t=s}^e 
\bm\varphi_i^{\top}(\wh{\mbf W}_x\wh {\mbf W}_x^\top- \mbf W_\chi\mbf W_\chi^\top)\mbf x_t\Big\vert
+
\\
& \max_{1 \le i \le n} \max_{1 \le s \le e \le T} \frac 1 {\sqrt{e-s+1}} \Big\vert \sum_{t=s}^e 
\bm\varphi_i^{\top}\mbf W_\chi\mbf W_\chi^\top \bm\eps_t \Big\vert
= III + IV
\end{align*}
where, following \eqref{eq:chi3},
\begin{align*}
III \le& 
\max_{1 \le i \le n} \Vert\bm\varphi_i^\top(\wh{\mbf W}_x - \mbf W_\chi\mbf S)\Vert
\, \Vert \wh{\mbf W}_x\Vert \, 
\max_{1 \le s \le e \le T} \frac{1}{\sqrt{e-s+1}} \Big\Vert \sum_{t=s}^e \mbf x_t \Big\Vert
\\
+ &
\max_{1 \le i \le n} \Vert\bm\varphi_i^\top\mbf W_\chi\Vert \, \Vert\wh{\mbf W}_x\mbf S-\mbf W_\chi \Vert \,
\max_{1 \le s \le e \le T} \frac{1}{\sqrt{e-s+1}} \Big\Vert \sum_{t=s}^e \mbf x_t \Big\Vert 
=  O_p\l\{\Big(\sqrt{\frac{\log\,n}{T}} \vee \frac{1}{\sqrt n}\Big)\log^\theta T\r\},
\end{align*}
from Lemma \ref{lem:evecs} (i)--(ii) and Lemma \ref{lem:x:n}. 
For $IV$, we first invoke Assumption \ref{assum:four} (i) as in \eqref{eq:chi4}:
\beas
\max_{1 \le s \le e \le T} \frac{1}{e-s+1} \E\l(\Big\Vert 
\sum_{t=s}^e {\mbf W}_{\chi}^{\top}\bm\eps_t \Big\Vert^2\r) 
= \sum_{j=1}^r \max_{1 \le s \le e \le T} \frac{1}{e-s+1}
\sum_{t, t'=s}^e \sum_{i, i'=1}^n w_{\chi, ij} w_{\chi,i'j} \E(\eps_{it}\eps_{i't'}) 
< rC_\eps.
\eeas
Hence, under Assumption \ref{assum:four} (ii),
$(e-s+1)^{-1/2}\sum_{t=s}^e \mbf W_\chi^\top\bm\eps_t$ is 
an $r$-vector of zero-mean normally distributed random variables
with finite variance, and thus
\begin{align*}
& \max_{1 \le s \le e \le T} \frac{1}{\sqrt{e-s+1}} \Big\Vert \sum_{t=s}^e \mbf W_\chi^\top\bm\eps_t \Big\Vert 
= O_p(\sqrt{\log\,T}), \quad \mbox{and}
\\
& IV \le \max_{1 \le i \le n} \Vert \bm\varphi_i^{\top}\mbf W_\chi\Vert\,
\max_{1 \le s \le e \le T} 
\frac{1}{\sqrt{e-s+1}}\Big\Vert \sum_{t=s}^e {\mbf W}_{\chi}^{\top}\bm\eps_t \Big\Vert 
= O_p\l(\sqrt{\frac{\log\,T}{n}}\r).
\end{align*}
Therefore, we have
\beas
\max_{1 \le i \le n}\max_{1 \le s \le e \le T}
\frac{1}{\sqrt{e-s+1}} \Big\vert \sum_{t=s}^e (\wh\chi_{it} - \chi_{it}) \Big\vert =
O_p\l\{\Big(\sqrt{\frac{\log\,n}{T}} \vee \frac{1}{\sqrt{n}}\Big)\log^\theta T\r\},
\eeas
which proves (i).
Next, 
\beas
I &\le& \max_{1 \le i \le n} \sum_{j = r+1}^k |\wt{ w}_{x, ij}| \max_{1 \le s \le e \le T}
\frac{1}{\sqrt{e-s+1}}\Big\vert \sum_{t=s}^e \wt{\mbf w}_{x, j}^{\top} \mbf x_t \Big\vert 
\\
&=& \frac{(k-r)c_w}{\sqrt n} \Vert \wt{\mbf w}_{x, j} \Vert
\max_{1 \le s \le e \le T} \frac{1}{\sqrt{e-s+1}} \Big\Vert \sum_{t=s}^e \mbf x_t \Big\Vert
= O_p(\log^\theta T),
\eeas
thanks to Lemma \ref{lem:x:n}, thus proving (ii).\hfill $\Box$

\subsection{Proof of Proposition \ref{prop:chi:additive}}
\label{sec:proof:prop:additive}

$g_j(\wh\chi^k_{it})$ and $h_j(\wh\chi^k_{it}, \wh\chi^k_{i't})$ admit the following decompositions
\begin{align}
g_j(\wh\chi^k_{it}) =& \, \E\{g_j(\chi^{\beta(t)}_{it})\} + 
[\E\{g_j(\chi_{it})\} - \E\{g_j(\chi^{\beta(t)}_{it})\}] + [g_j(\chi_{it}) - \E\{g_j(\chi_{it})\}] 
+ \{g_j(\wh\chi^k_{it}) - g_j(\chi_{it})\} \nn
\\
=& \, \E\{g_j(\chi^{\beta(t)}_{it})\} + I + II + III,
\nn 
\\
h_j(\wh\chi^k_{it}, \wh\chi^k_{i't}) =& \, \E\{h_j(\chi^{\beta(t)}_{it}, \chi^{\beta(t)}_{i't})\} 
+ [\E\{h_j(\chi_{it}, \chi_{i't})\} - \E\{h_j(\chi^{\beta(t)}_{it}, \chi^{\beta(t)}_{i't})\}] \nn
\\
& + [h_j(\chi_{it}, \chi_{i't}) - \E\{h_j(\chi_{it}, \chi_{i't})\}] + 
\{h_j(\wh\chi^k_{it}, \wh\chi^k_{i't}) - h_j(\chi_{it}, \chi_{i't})\} \nn
\\
=& \, \E\{h_j(\chi^{\beta(t)}_{it}, \chi^{\beta(t)}_{i't})\} + IV + V + VI. \nn 
\end{align}
By definition, $\E\{g_j(\chi^{\beta(t)}_{it})\}$ and 
$\E\{h_j(\chi^{\beta(t)}_{it}, \chi^{\beta(t)}_{i't})\}$ are piecewise constant 
with their change-points belonging to $\cB^\chi$.
Moreover, under Assumption \ref{assum:one}, 
all change-points in $\bm\Gamma_\chi^{\beta(t)}(\tau), \, |\tau| \le \bar{\tau}_\chi$, 
i.e., all $\eta^\chi_b \in \cB^\chi$, appear as change-points in the panel
$\{\E\{g_j(\chi^{\beta(t)}_{it})\}, \, 1 \le i \le n, \, 
\E\{h_j(\chi^{\beta(t)}_{it}, \chi^{\beta(t)}_{i't})\}, \, 1 \le i < i' \le n\}$ for $j \ge -\jts$,
with the choice of $\jts = \lfloor \log_2\log^\upsilon T \rfloor$ as discussed in Section \ref{sec:wave}.

Next, we turn our attention to scaled sums of $I$ and $IV$ over any given interval $[s, e]$, 
which is bounded as in the following Lemma \ref{lem:chi:i} (see Appendix \ref{pf:thm:i} for a proof).
\begin{lem}
\label{lem:chi:i}
{\it Suppose that the conditions of Theorem \ref{thm:overestimation} are met. At scales $j \ge -\jts$,
\begin{align*}
\max_{1 \le s < e \le T} \frac{1}{\sqrt{e-s+1}} & \Big\{
\max_{1 \le i \le n} \Big\vert \sum_{t=s}^e \E\{g_j(\chi_{it})\} - \E\{g_j(\chi^{\beta(t)}_{it})\} \Big\vert \vee \\
& \max_{1 \le i < i' \le n} \Big\vert \sum_{t=s}^e \E\{h_j(\chi_{it}, \chi_{i't})\} - 
\E\{h_j(\chi^{\beta(t)}_{it}, \chi^{\beta(t)}_{i't})\} \Big\vert \Big\}
= O(\log^\upsilon T).
\end{align*}}
\end{lem}

For bounding $III$ and $VI$, we investigate the behaviour of  
$\xi_j(\wh\chi^k_{it}) = g_j(\wh\chi^k_{it}) - g_j(\chi_{it})$ 
and $\zeta_j(\wh\chi^k_{it}, \wh\chi^k_{i't}) = h_j(\wh\chi^k_{it}, \wh\chi^k_{i't}) - h_j(\chi_{it}, \chi_{i't})$,
the errors arising from replacing the unobservable $\chi_{it}$ by its estimate $\wh\chi^k_{it}$, 
in the following Lemma \ref{lem:chi:xi} (see Appendix \ref{pf:thm:xi} for a proof).

\begin{lem}
\label{lem:chi:xi}
{\it Suppose that the conditions of Theorem \ref{thm:overestimation} are met. At scales $j \ge -\jts$ and $k \ge r$,
\beas
\max_{1 \le s < e \le T} \frac{1}{\sqrt{e-s+1}} \Big\{
\max_{1 \le i \le n} \Big\vert \sum_{t=s}^e \xi_j(\wh\chi^k_{it}) \Big\vert \vee 
\max_{1 \le i < i' \le n} \Big\vert \sum_{t=s}^e \zeta_j(\wh\chi^k_{it}, \wh\chi^k_{i't}) \Big\vert \Big\}
= O_p(\log^{\theta+\upsilon} T).
\eeas}
\end{lem}

Finally, the scaled partial sums of $II$ and $V$ are handled by the following Lemma \ref{lem:chi:a9} (see Appendix \ref{pf:lem:chi:a9} for a proof).
\begin{lem}
\label{lem:chi:a9}
{\it Suppose that the conditions of Theorem \ref{thm:overestimation} are met. 
At scales $j \ge -\jts$,
\begin{align*}
& \max_{1 \le s < e \le T} \frac{1}{\sqrt{e-s+1}} \Big\{
\max_{1 \le i \le n} \Big\vert \sum_{t=s}^e g_j(\chi_{it}) - \E\{g_j(\chi_{it})\} \Big\vert 
\\
&   \qquad \qquad \qquad \vee \max_{1 \le i < i' \le n}
\Big\vert \sum_{t=s}^e h_j(\chi_{it}, \chi_{i't}) - \E\{h_j(\chi_{it}, \chi_{i't})\} \Big\vert\Big\} 
= O_p(\log^{\theta+\upsilon}T).
\end{align*}}
\end{lem}

From Lemmas \ref{lem:chi:i}--\ref{lem:chi:a9} Proposition \ref{prop:chi:additive} follows.\hfill $\Box$

\subsubsection{Proof of Lemma \ref{lem:chi:i}}
\label{pf:thm:i}

Let $\tau = e-s+1$  when there is no confusion. Note that
\begin{align}
& \frac{1}{\sqrt{\tau}} \Big\vert \sum_{t=s}^e  \E\{g_j(\chi_{it}) - g_j(\chi^{\beta(t)}_{it})\} \Big\vert
\le \sum_{h=0}^{\cL_j-1} \frac{1}{\sqrt{\tau}} \Big\vert \sum_{t=s}^e 
\E\Big\{g_j(\chi^{\beta(t)}_{it}, \ldots, \chi^{\beta(t)}_{i, t-h+1}, \chi_{i, t-h}, \ldots, \chi_{i, t-\cL_j+1})
\nn
\\
& \qquad \qquad \qquad \qquad \qquad 
- g_j(\chi^{\beta(t)}_{it}, \ldots, \chi^{\beta(t)}_{i, t-h}, \chi_{i, t-h-1}, \ldots, \chi_{i, t-\cL_j+1})\Big\}\Big\vert.
\label{eq:pf:i:one}
\end{align}
Restricting our attention to the first summand in the RHS of (\ref{eq:pf:i:one}) (when $h=0$), by the definition of $g_j$,
\begin{align*}
& \frac{1}{\sqrt{\tau}} \Big\vert \E\Big\{\sum_{t=s}^e 
|\underbrace{\sum_{l=0}^{\cL_j-1} \psi_{j, l}\chi_{i, t-l}}_{w_{it}}| - 
|\underbrace{\psi_{j, 0}\chi^{\beta(t)}_{it} + \sum_{l=1}^{\cL_j-1}\psi_{j, l}\chi_{i, t-l}}_{v_{it}}| \Big\}\Big\vert
\le \frac{1}{\sqrt{\tau}} \Big\vert \sum_{t: \, w_{it} \cdot v_{it} \ge 0}  \E(w_{it} - v_{it}) \Big\vert
\\
+ &
\frac{1}{\sqrt{\tau}} \Big\vert \sum_{\substack{t: \, w_{it} \cdot v_{it} < 0 \\ |w_{it}| \ge |v_{it}|}} 
\E(w_{it} + v_{it}) \Big\vert
 +
\frac{1}{\sqrt{\tau}} \Big\vert \sum_{\substack{t: \, w_{it} \cdot v_{it} < 0 \\ |w_{it}| < |v_{it}|}} \E(w_{it} + v_{it})\Big\vert
= I+II+III.
\end{align*}
Starting from $I$, let $b_1 = \beta(s)$ and $b_2 = \beta(e)+1$. Then, 
\beas
I^2 \le 
\psi_{j, 0}^2 \sum_{t=s}^e \E(\chi_{it} - \chi^{\beta(t)}_{it})^2 
\le \psi_{j, 0}^2 \sum_{b=b_1}^{b_2} \sum_{t = s \vee (\eta_b+1)}^{e \wedge \eta_{b+1}} O(r^2\bar{\lambda}^2\rho_f^{t-\eta_b})
= O\Big(\frac{\psi_{j, 0}^2B_\chi r^2\bar{\lambda}^2\rho_f}{1-\rho_f}\Big) = O(1)
\eeas
uniformly in $1 \le s \le e \le T$ for any $i = 1, \ldots, n$, from Assumptions \ref{assum:one} (i) and
\ref{assum:three}.
As for $II$, note that for all $t$ satisfying $w_{it} \cdot v_{it} < 0$ and $|w_{it}| \ge |v_{it}|$,
we have $w_{it} + v_{it} = c_t(w_{it} - v_{it})$ for some $c_t \in [0, 1)$,
hence we can similarly show that $II = O(1)$ and the same arguments apply to $III$.
Then, (\ref{eq:pf:i:one}) involves summation of such summands over $h=0, \ldots, \cL_j-1$ and thus is bounded by $O(\log^\upsilon T)$.
Similar arguments can be employed to derive a bound on 
the scaled partial sums of $\E\{h_j(\chi_{it}, \chi_{i't})\} - \E\{h_j(\chi^{\beta(t)}_{it}, \chi^{\beta(t)}_{i't})\}$,
which concludes the proof.\hfill $\Box$

\subsubsection{Proof of Lemma \ref{lem:chi:xi}}
\label{pf:thm:xi}

The proof of Theorem \ref{thm:overestimation} indicates that
\bea
\max_{1 \le i \le n} \max_{I \subset T}\frac{1}{\sqrt{|I|}} \Big\vert \sum_{t \in I} (\wh\chi^k_{it} - \chi_{it}) \Big\vert = O_p(\log^\theta T).
\label{eq:pf:xi:one} 
\eea
Let $\tau = e-s+1$ when there is no confusion.
Adopting the similar arguments as in Appendix \ref{pf:thm:i}, we need to derive a bound on the following:
\begin{align*}
& \frac{1}{\sqrt{\tau}} \Big\vert \sum_{t=s}^e 
|\underbrace{\sum_{l=0}^{\cL_j-1}\psi_{j, l}\wh\chi^k_{i, t-l}}_{w_{it}}| - 
|\underbrace{\psi_{j, 0}\chi_{it} + \sum_{l=1}^{\cL_j-1}\psi_{j, l}\wh\chi^k_{i, t-l}}_{v_{it}}|\Big\vert
\le \frac{1}{\sqrt{\tau}} \Big\vert \sum_{t: \, w_{it} \cdot v_{it} \ge 0} (w_{it} - v_{it}) \Big\vert
\\
+ &
\frac{1}{\sqrt{\tau}} \Big\vert \sum_{\substack{t: \, w_{it} \cdot v_{it} < 0 \\ |w_{it}| \ge |v_{it}|}} (w_{it} + v_{it}) \Big\vert
+
\frac{1}{\sqrt{\tau}} \Big\vert \sum_{\substack{t: \, w_{it} \cdot v_{it} < 0 \\ |w_{it}| < |v_{it}|}} (w_{it} + v_{it}) \Big\vert
= I+II+III.
\end{align*}

Since $w_{it} - v_{it} = \psi_{j, 0}(\wh\chi^k_{it} - \chi_{it})$, 
we have $I = O_p(\log^\theta T)$ from \eqref{eq:pf:xi:one}.
As for $II$, note that for all $t$ satisfying $w_{it} \cdot v_{it} < 0$ and $|w_{it}| \ge |v_{it}|$,
we have $w_{it} + v_{it} = c_t(w_{it} - v_{it})$ for some $c_t \in [0, 1)$,
which as for $I$ it leads to $II =O_p(\log^\theta T)$. Similar arguments apply to $III$.
Then, similarly to \eqref{eq:pf:i:one}, $\tau^{-1/2} |\sum_{t=s}^e \xi_j(\wh\chi^k_{it})|$ involves summation of $\cL_j$ such summands,
and thus is $O_p(\log^{\theta+\upsilon} T)$.
Analogous arguments can be applied to bound the scaled partial sums of $\zeta_j(\wh\chi^k_{it}, \wh\chi^k_{i't})$.\hfill $\Box$

\subsubsection{Proof of Lemma \ref{lem:chi:a9}}
\label{pf:lem:chi:a9}

The proof of Lemma \ref{lem:x:n} implies that
\beas
\max_{1 \le i \le n} \max_{I \subset T}
\frac{1}{\sqrt{| I |}} \Big\vert \sum_{t \in I} \{\chi_{it} - \E(\chi_{it})\} \Big\vert = O_p(\log^\theta T)
\eeas
(with $\E(\chi_{it}) = 0$),
from which Lemma \ref{lem:chi:a9} follows 
since $g_j(\chi_{it})$ and $h_j(\chi_{it}, \chi_{i't})$
involves at most $O(2^{\jts}) = O(\log^\upsilon T)$ (lagged) terms of $\chi_{it}$;
for details of the arguments, see the proof of Lemma \ref{lem:chi:xi} in Appendix \ref{pf:thm:xi}.\hfill $\Box$

\subsection{Proof of Theorem \ref{thm:dcbs}}
\label{pf:dcbs}

The DC operator in this paper is identical to $\cD^\varphi_m(\cdot)$ defined in \cite{cho2016} with $\varphi = 1/2$.
Let the additive panel data considered therein be $y_{\ell t}' = z_{\ell t}' + \vep_{\ell t}'$.
Then Assumptions \ref{assum:b:one}--\ref{assum:b:two} along with 
Assumption \ref{assum:six} imposed on the change-points in $z_{\ell t}$,
are sufficient for the conditions imposed on $z_{\ell t}'$.
On the other hand, their noise term satisfies $\E(\vep_{\ell t}') = 0$,
while it is generally expected that $\E(\vep_{\ell t}) \ne 0$.
Assuming that $\vep_{\ell t}'$ is strong mixing with bounded moments,
it was shown that $(e-s+1)^{-1/2}|\sum_{t=s}^e \vep_{\ell t}'| \le \log\,T$ uniformly in $\ell \in \{1, \ldots, N\}$
and $1 \le s < e \le T$ (their Lemma 1), which is comparable to the bound of $\log^{\theta+\upsilon} T$ 
on $(e-s+1)^{-1/2}|\sum_{t=s}^e \vep_{\ell t}|$ as shown in 
Propositions \ref{prop:chi:additive} and \ref{prop:e:additive}.
This enables us to directly employ the arguments used in the proofs of Theorem 3.3 of \cite{cho2016} 
for the proof of Theorem \ref{thm:dcbs}.\hfill$\Box$

\subsection{Proof of Proposition \ref{prop:bn}}

In order to prove Proposition \ref{prop:bn}, 
we first introduce the following lemmas (see Appendix \ref{lem:pw:proof}--\ref{lem:lam:proof} for the proofs).
\begin{lem} 
\label{lem:pw}
{\it Suppose that all the conditions in Theorem \ref{thm:dcbs} hold.
Let
$\bm\Gamma^b_x = (\eta^\chi_{b+1} - \eta^\chi_b)^{-1} 
\E(\sum_{t=\eta^\chi_b+1}^{\eta^\chi_{b+1}} \mbf x_t\mbf x_t^\top)$,
and define $\wh{\bm\Gamma}^b_\chi$ analogously as $\wh{\bm\Gamma}^b_x$. 
Then, for all $b = 0, \ldots, \wh{B}_\chi$, 
\bea
\frac{1}{n} \Vert \wh{\bm\Gamma}^b_x - \bm\Gamma^b_x \Vert_F = 
O_p\l(\frac{\omega_{n, T}\log^{\theta-1/2} T}{T} \vee \sqrt{\frac{\log\,n}{T}} \r),
\label{lem:pw:eq:one}
\\
\frac{1}{n} \Vert \wh{\bm\Gamma}^b_\chi - \bm\Gamma^b_\chi \Vert_F = 
O_p\l(\frac{\omega_{n, T}\log^{\theta-1/2} T}{T} \vee \sqrt{\frac{\log\,n}{T}} \r).
\label{lem:pw:eq:two}
\eea}
\end{lem}

\begin{lem}
\label{lem:pw:evec}
{\it Suppose that all the conditions in Theorem \ref{thm:dcbs} hold.
Let $\mbf W^b_{\chi} = [\mbf w^b_{\chi,1}, \ldots, \mbf w^b_{\chi, r_b}]$,
with $\mbf w^b_{\chi, j}$ denoting the normalised eigenvectors corresponding to $\mu^b_{\chi, j}$,
the $j$th largest eigenvalues of ${\bm\Gamma}^b_{\chi}$.
For some $0 \le k \le r_b-1$, let $\wh{\mbf V} = [\wh{\mbf w}^b_{x, k+1}, \ldots, \wh{\mbf w}^b_{x, r_b}]$. 
Then, there exists an orthonormal $r_b \times (r_b-k)$ matrix $\wt{\mbf S}$ such that
\beas
\Vert \wh{\mbf V} - \mbf W^b_\chi\wt{\mbf S} \Vert = O_p\l(\sqrt{\frac{\log\,n}{T}} \vee \frac{1}{n}\r). 
\eeas}
\end{lem}

\begin{lem}
\label{lem:lam}
{\it Suppose that all the conditions in Theorem \ref{thm:dcbs} and Assumption \ref{assum:three:new} hold.
For a fixed $k > r$, let $\wh{\mbf V} = [\wh{\mbf w}^b_{x, r_b+1}, \ldots, \wh{\mbf w}^b_{x, k}]$. Then,
\beas
\Vert \wh{\mbf V}^\top \bm\Lambda_b \Vert = O_p\l(\sqrt{\frac{n\log\,n}{T}}\vee \frac 1 {\sqrt n}\r).
\eeas}
\end{lem}

From Theorem \ref{thm:dcbs},
we have $\p(\mc C_{n, T}) \to 1$ where
\beas
\mc C_{n, T} = \Big\{ 
\wh{B}_\chi =B_\chi; \, \max_{1 \le b \le \wh{B}_\chi} 
|\wh\eta^\chi_b-\eta^\chi_b| < c_1\omega_{N, T}\Big\}
\eeas
with for some $c_1 > 0$ and 
$\omega_{n, T} = \jts \min_{1 \le b \le B_\chi} \wt{\delta}_b^{-2} \log^{2\theta+2\upsilon} T$.
We first show that 
\bea
\label{ineq:prop:bn}
\p\{r_b = \arg\min_{1 \le k \le \bar{r}} V_b(k)\}
\ge
\p\{r_b = \arg\min_{1 \le k \le \bar{r}} V_b(k) \, | \, \mc C_{n, T}\} \p(\mc C_{n, T}) \to 1, 
\eea
where, denoting by 
$\wh{\mbf W}^b_{j:l} = [\wh{\mbf w}^b_{x, j}, \ldots, \wh{\mbf w}^b_{x, l}]$ for $1 \le j \le l \le n$,
\beas
V_b(k) 
= \frac{1}{n(\wh\eta^\chi_{b+1}-\wh\eta^\chi_b)}\sum_{t=\wh\eta^\chi_b+1}^{\wh\eta^\chi_{b+1}} 
\Vert \wh{\mbf W}^b_{1:k} (\wh{\mbf W}^b_{1:k})^\top \mbf x_t \Vert^2 + k \, p(n, T).
\eeas
Thanks to Theorem \ref{thm:dcbs}, it is sufficient to show that 
$\p\{r_b = \arg\min_{1 \le k \le \bar{r}} V_b(k) | \mc C_{n, T}\} \to 1$
for the proof of \eqref{ineq:prop:bn}.
Firstly, let $k > r_b$. Due to the orthonormality of $\wh{\mbf w}^b_{x, j}$,
\beas
V_b(k) - V_b(r_b) 
= \frac{1}{n(\wh\eta^\chi_{b+1}-\wh\eta^\chi_b)} \sum_{t=\wh\eta^\chi_b+1}^{\wh\eta^\chi_{b+1}} 
\Vert \wh{\mbf V} \wh{\mbf V}^\top \mbf x_t \Vert^2
+ (k-r_b) p(n, T)
\eeas
where $\wh{\mbf V} = \wh{\mbf W}^b_{(r_b+1):k}$. Note that
\begin{align*}
\frac{1}{n(\wh\eta^\chi_{b+1}-\wh\eta^\chi_b)} \sum_{t=\wh\eta^\chi_b+1}^{\wh\eta^\chi_{b+1}} 
\Vert \wh{\mbf V} \wh{\mbf V}^\top \mbf x_t \Vert^2
\le & \frac{2}{n(\wh\eta^\chi_{b+1}-\wh\eta^\chi_b)} \sum_{t=\wh\eta^\chi_b+1}^{\wh\eta^\chi_{b+1}} 
\Vert \wh{\mbf V} \wh{\mbf V}^\top \bm\chi_t \Vert^2
\\
& + \frac{2}{n(\wh\eta^\chi_{b+1}-\wh\eta^\chi_b)} \sum_{t=\wh\eta^\chi_b+1}^{\wh\eta^\chi_{b+1}} 
\Vert \wh{\mbf V} \wh{\mbf V}^\top \bm\eps_t \Vert^2 
= I + II.
\end{align*}
Then, using Lemma \ref{lem:lam},
\beas
I &\le& \frac{2}{n(\wh\eta^\chi_{b+1}-\wh\eta^\chi_b)} \sum_{t=\wh\eta^\chi_b+1}^{\wh\eta^\chi_{b+1}} 
\Vert \wh{\mbf V} \Vert^2 \Vert \wh{\mbf V}^\top\bm\Lambda_b \Vert^2 
\Vert \mbf f_t \Vert^2
= \frac{2}{n} \cdot O_p\l( \frac{n\log\,n}{T}\vee \frac 1n\r) \cdot O_p(\log^{2/\beta_f} T) 
\\
&=& O_p\l\{\Big(\frac{\log\, n}{T} \vee \frac{1}{n^2}\Big)\log^{2/\beta_f} T\r\},
\eeas
due to Assumptions \ref{assum:two} (ii), \ref{assum:five} (ii) and \ref{assum:six}. Also,
\begin{align*}
II =& \frac{2}{n} 
\tr\l[\wh{\mbf V}\wh{\mbf V}^\top
\frac{1}{|I^\chi_b|}\sum_{t \in I^\chi_b}
\Big\{\bm\eps_t\bm\eps_t^\top - \E(\bm\eps_t\bm\eps_t^\top)\Big\}\r]
 +
\frac{2}{n} 
\tr\l\{\wh{\mbf V}\wh{\mbf V}^\top
\frac{1}{|I^\chi_b|}\sum_{t \in I^\chi_b}\E(\bm\eps_t\bm\eps_t^\top)\r\}
\\
\le &
\frac{2(k-r_b)}{|I^\chi_b|n}\l\Vert \sum_{t \in I^\chi_b}
\Big\{\bm\eps_t\bm\eps_t^\top - \E(\bm\eps_t\bm\eps_t^\top)\Big\} \r\Vert + 
\frac{2(k-r_b)}{|I^\chi_b|n}\l\Vert \sum_{t \in I^\chi_b}\E(\bm\eps_t\bm\eps_t^\top) \r\Vert
= O_p\Big(\sqrt{\frac{\log\,n}{T}}\Big)
\end{align*}
invoking Assumptions \ref{assum:four})--\ref{assum:five} and Lemma A.3 of \cite{fan2011b}.
Hence, under the conditions imposed on $p(n, T)$,
we conclude that $V_b(k) > V_b(r_b)$ for any fixed $k > r_b$ 
with probability tending to one as $n, T \to \infty$.

Next, let $k < r_b$. Recalling the definition of $\bm\Gamma^b_\chi$,
denote the $r_b \times r_b$ diagonal matrix with $\mu^b_{\chi, j}, \, j=1, \ldots, r_b$ 
in its diagonal by $\mbf M^b_\chi$.
Note that
\beas
V_b(k) - V_b(r_b) = 
\frac{1}{n(\wh\eta^\chi_{b+1}-\wh\eta^\chi_b)}\sum_{t=\wh\eta^\chi_b+1}^{\wh\eta^\chi_{b+1}}
\Vert \wh{\mbf V} \wh{\mbf V}^\top \mbf x_t \Vert^2
+ (r_b-k) p(n, T)
\eeas
where $\wh{\mbf V} = \wh{\mbf W}^b_{(k+1):r_b}$. Further,
\begin{align*}
& \frac{1}{n(\wh\eta^\chi_{b+1}-\wh\eta^\chi_b)}\sum_{t=\wh\eta^\chi_b+1}^{\wh\eta^\chi_{b+1}} 
\Vert \wh{\mbf V} \wh{\mbf V}^\top \mbf x_t \Vert^2
= \frac{1}{n(\wh\eta^\chi_{b+1}-\wh\eta^\chi_b)}\sum_{t=\wh\eta^\chi_b+1}^{\wh\eta^\chi_{b+1}}
\Vert \wh{\mbf V} \wh{\mbf V}^\top \bm\chi_t \Vert^2
\\
& + \frac{2}{n(\wh\eta^\chi_{b+1}-\wh\eta^\chi_b)}\sum_{t=\wh\eta^\chi_b+1}^{\wh\eta^\chi_{b+1}}
\bm\chi_t^\top  \wh{\mbf V} \wh{\mbf V}^\top \bm\eps_t
+ \frac{1}{n(\wh\eta^\chi_{b+1}-\wh\eta^\chi_b)}\sum_{t=\wh\eta^\chi_b+1}^{\wh\eta^\chi_{b+1}}
\Vert \wh{\mbf V} \wh{\mbf V}^\top \bm\eps_t \Vert^2 
= III + IV + V.
\end{align*}
Then, we can bound $V = O_p(\sqrt{\log\,n/T})$ similarly as $II$.
Also, thanks to Lemma \ref{lem:pw:evec}, 
there exists an $r_b \times (r_b-k)$ matrix $\wt{\mbf S}$ with orthonormal columns
so that
\begin{align}
\Vert \wh{\mbf V}\wh{\mbf V}^\top - \mbf W^b_\chi\wt{\mbf S}\wt{\mbf S}^\top(\mbf W^b_\chi)^\top \Vert
&\le \Vert \wh{\mbf V}(\wh{\mbf V} - \mbf W^b_\chi\wt{\mbf S})^\top \Vert
+ \Vert (\wh{\mbf V} - \mbf W^b_\chi\wt{\mbf S})\mbf W^b_\chi\wt{\mbf S}^\top \Vert
= O_p\l(\sqrt{\frac{\log\,n}{T}} \vee \frac{1}{n}\r). \label{eq:vv:ww}
\end{align}
Note that
\begin{align*}
III &\ge 
\overbrace{\frac{1}{n(\wh\eta^\chi_{b+1}-\wh\eta^\chi_b)}
\sum_{t=\wh\eta^\chi_b+1}^{\wh\eta^\chi_{b+1}} \Vert \mbf W^b_\chi\wt{\mbf S}
\wt{\mbf S}^\top(\mbf W^b_\chi)^\top\bm\chi_t \Vert^2}^{VI}
\\
&- \overbrace{\frac{2}{n(\wh\eta^\chi_{b+1}-\wh\eta^\chi_b)}\sum_{t=\wh\eta^\chi_b+1}^{\wh\eta^\chi_{b+1}}
\Vert \mbf W^b_\chi\wt{\mbf S}\wt{\mbf S}^\top(\mbf W^b_\chi)^\top\bm\chi_t \Vert \;
\Vert \{\wh{\mbf V}\wh{\mbf V}^\top - 
\mbf W^b_\chi\wt{\mbf S}\wt{\mbf S}^\top(\mbf W^b_\chi)^\top\}\bm\chi_t \Vert}^{VII}
\\
&+ \underbrace{\frac{1}{n(\wh\eta^\chi_{b+1}-\wh\eta^\chi_b)}
\sum_{t=\wh\eta^\chi_b+1}^{\wh\eta^\chi_{b+1}}
\Vert \{\wh{\mbf V}\wh{\mbf V}^\top - 
\mbf W^b_\chi\wt{\mbf S}\wt{\mbf S}^\top(\mbf W^b_\chi)^\top\}\bm\chi_t \Vert^2}_{VIII}.
\end{align*}
Then,
\begin{align*}
VI &= \frac{1}{n(\wh\eta^\chi_{b+1}-\wh\eta^\chi_b)}\sum_{t=\wh\eta^\chi_b+1}^{\wh\eta^\chi_{b+1}}
\bm\chi_t^\top\mbf W^b_\chi\wt{\mbf S}\wt{\mbf S}^\top(\mbf W^b_\chi)^\top\bm\chi_t
= \frac{1}{n} \tr(\mbf W^b_\chi\wt{\mbf S}\wt{\mbf S}^\top(\mbf W^b_\chi)^\top \wh{\bm\Gamma}^b_\chi)
\\
&=\frac{1}{n} \tr(\mbf W^b_\chi\wt{\mbf S}\wt{\mbf S}^\top(\mbf W^b_\chi)^\top \bm\Gamma^b_\chi) +
\frac{1}{n} \tr\{\mbf W^b_\chi\wt{\mbf S}\wt{\mbf S}^\top(\mbf W^b_\chi)^\top 
(\wh{\bm\Gamma}^b_\chi-\bm\Gamma^b_\chi)\}
\\
&\le \frac{1}{n} \tr\{\mbf W^b_\chi\wt{\mbf S}\wt{\mbf S}^\top(\mbf W^b_\chi)^\top \bm\Gamma^b_\chi\} +
\Vert\mbf W^b_\chi\wt{\mbf S}\wt{\mbf S}^\top(\mbf W^b_\chi)^\top\Vert_F \cdot \frac{1}{n} 
\Vert \wh{\bm\Gamma}^b_\chi-\bm\Gamma^b_\chi \Vert_F
\\
&= \frac{1}{n} \tr(\mbf W^b_\chi\wt{\mbf S}\wt{\mbf S}^\top(\mbf W^b_\chi)^\top \bm\Gamma^b_\chi) + 
O_p\l(\sqrt{\frac{\log\,n}{T}}\r)
= \frac{1}{n}\tr(\wt{\mbf S}\wt{\mbf S}^\top\mbf M^b_\chi) + O_p\l(\sqrt{\frac{\log\,n}{T}}\r),
\end{align*}
which follows from Lemma \ref{lem:pw}
and that $\wt{\mbf S}\wt{\mbf S}^\top$ is a rank $r_b-k$ projection matrix,
and hence $VI > 0$. 
Also, using \eqref{eq:vv:ww} and Assumption \ref{assum:two} (ii),
\beas
\Vert \{\wh{\mbf V}\wh{\mbf V}^\top - 
\mbf W^b_\chi\wt{\mbf S}\wt{\mbf S}^\top(\mbf W^b_\chi)^\top\}\bm\chi_t \Vert
&=& O_p\l\{\Big(\sqrt{\frac{n\log\,n}{T}} \vee \frac{1}{\sqrt n}\Big)\log^{1/\beta_f} T\r\}, \quad \mbox{and} 
\\
\Vert \mbf W^b_\chi\wt{\mbf S}\wt{\mbf S}^\top(\mbf W^b_\chi)^\top\bm\chi_t \Vert 
&=& O_p(\sqrt n\log^{1/\beta_f} T)
\eeas
uniformly in $t$, and therefore $VII = O_p\{(\sqrt{\log\,n/T} \vee 1/n)\log^{2/\beta_f}T\}$.
Besides, 
\beas
VIII \le \frac{1}{n(\wh\eta^\chi_{b+1}-\wh\eta^\chi_b)}\sum_{t=\wh\eta^\chi_b+1}^{\wh\eta^\chi_{b+1}}
\Vert \wh{\mbf V}\wh{\mbf V}^\top - \mbf W^b_\chi\wt{\mbf S}\wt{\mbf S}^\top(\mbf W^b_\chi)^\top \Vert^2
\Vert \bm\chi_t \Vert^2
= O_p\l\{\Big(\frac{\log\,n}{T} \vee \frac{1}{n^2}\Big)\log^{2/\beta_f}T\r\}.
\eeas
Combining the bounds on $VI$, $VII$ and $VIII$, we conclude that $III$ is bounded away from zero with probability
tending to one.
Finally, under Assumptions \ref{assum:two} (ii), \ref{assum:four} (ii) and \ref{assum:five},
Lemma B.1 (ii) of \cite{fan2011b} leads to
\beas
IV = \frac{2}{n|I^\chi_b|} \tr\l(\wh{\mbf V}\wh{\mbf V}^\top 
\sum_{t \in I^\chi_b} \bm\chi_t\bm\eps_t^\top\r)
\le \frac{2(r_b-k)}{n|I^\chi_b|}\l\Vert 
\sum_{t \in I^\chi_b} \bm\chi_t\bm\eps_t^\top \r\Vert
= O_p\l(\sqrt{\frac{\log\,n}{T}}\r),
\eeas
which leads to $V_b(k) > V_b(r_b)$ with probability converging to one.
Having shown that $V_b(k)$ is minimised at $r_b$, we proved \eqref{ineq:prop:bn}. 
Then we can adopt the arguments used in the proof of Corollary 1 in \cite{baing02} verbatim to complete the proof. 
\hfill $\Box$

\subsubsection{Proof of Lemma \ref{lem:pw}}
\label{lem:pw:proof}

For a given $b$, without loss of generality, assume that $\wh\eta_b \le \eta_b < \wh\eta_{b+1} \le \eta_{b+1}$. 
Define
$\check{\bm\Gamma}^b_x=(\eta^\chi_{b+1} - \eta^\chi_b)^{-1} 
\sum_{t=\eta^\chi_b+1}^{\eta^\chi_{b+1}} \mbf x_t\mbf x_t^\top$ and 
let $\check\gamma^b_{x, ii'} = [\check{\bm\Gamma}^b_x]_{i, i'}$ and 
$\wh\gamma^b_{x, ii'} = [\wh{\bm\Gamma}^b_x]_{i, i'}$. Then,
\beas
\Vert \wh{\bm\Gamma}^b_x - \bm\Gamma^b_x \Vert_F^2 
= \sum_{i, i'=1}^n |\wh\gamma^b_{x, ii'} - \gamma^b_{x, ii'}|^2
\le 2\sum_{i, i'=1}^n |\wh\gamma^b_{x, ii'} - \check\gamma^b_{x, ii'}|^2 +
2\sum_{i, i'=1}^n |\check\gamma^b_{x, ii'} - \gamma^b_{x, ii'}|^2 = I + II.
\eeas
Under Assumption \ref{assum:six}, Lemmas A.3 and B.1 (ii) of \cite{fan2011b} can be adopted to 
show that $II = O_p(n^2\log\,n/T)$.
Recall the definition of $\mc C_{n, T}$ from the proof of Proposition \ref{prop:bn}.
Then, for some $c_2 > 0$, consider
\bea
\p\Big(\max_{1 \le i, i' \le n} |\wh\gamma^b_{x, ii'} - \check\gamma^b_{x, ii'}| > \frac{c_2\omega_{n, T}}{T}\Big)
\le 
\p\Big(\max_{1 \le i, i' \le n} |\wh\gamma^b_{x, ii'} - \check\gamma^b_{x, ii'}| > \frac{c_2\omega_{n, T}}{T}
\, \cap \, \mc C_{n, T}\Big) + \p(\mc C_{n, T}^c),
\label{lem:pw:one}
\eea
where the second probability in the RHS of \eqref{lem:pw:one} tends to zero as $n, T \to \infty$.
Note that
\begin{align*}
|\wh\gamma^b_{x, ii'} - \check\gamma^b_{x, ii'}| \le & 
\frac{1}{\wh\eta^\chi_{b+1} - \wh\eta^\chi_b} 
\sum_{t=\wh\eta^\chi_b+1}^{\eta^\chi_b} \vert x_{it}x_{i't} \vert
+
\l\vert \frac{1}{\wh\eta^\chi_{b+1} - \wh\eta^\chi_b} - \frac{1}{\eta^\chi_{b+1} - \eta^\chi_b} \r\vert
\sum_{t=\eta^\chi_b+1}^{\wh\eta^\chi_{b+1}} \vert x_{it}x_{i't} \vert
\\
& +
\frac{1}{\eta^\chi_{b+1} - \eta^\chi_b} 
\sum_{t=\wh\eta^\chi_{b+1}+1}^{\eta^\chi_{b+1}} \vert x_{it}x_{i't}\vert = III + IV + V.
\end{align*}
Lemma A.2 of \cite{fan2011b} shows that the exponential tail bound
carries over to $x_{it}x_{i't}$ with parameter $2\beta_f/(2+\beta_f)$, from which we derive that
$\max_{1 \le i, i' \le n}\max_{1 \le t \le T} |x_{it}x_{i't}| = O_p(\log^{\theta - 1/2} T)$ 
under Assumption \ref{assum:seven}.
Then, we have $III, IV, V = O_p(\omega_{n, T}\log^{\theta - 1/2} T/T)$ 
uniformly in $i, i' = 1, \ldots, n$ in the event of $\mc C_{n, T}$ under Assumption \ref{assum:six}.
Hence, the RHS of \eqref{lem:pw:one} tends to zero with $n, T \to \infty$, 
which leads to $I = O_p(n^2\omega_{n, T}^2\log^{2\theta - 1} T/T^2)$ and 
concludes the proof of \eqref{lem:pw:eq:one}. The proof of \eqref{lem:pw:eq:two} follows analogously. \hfill $\Box$

\subsubsection{Proof of Lemma \ref{lem:pw:evec}}
\label{lem:pw:evec:proof}

Recalling the discussion on Assumption \ref{assum:b:two} and 
the definition of $\omega_{n, T}$ in Theorem \ref{thm:dcbs},
we have $\omega_{n, T} \ll T^{1/2}$ 
for estimating the change-points in the common components, 
and thus $\sqrt{\log\,n/T}$ dominates the RHS of \eqref{lem:pw:eq:one}. 
Then, we apply the variation of sin$\,\theta$ theorem in \cite{yu15} as in Lemma \ref{lem:evecs} 
in combination with Lemma \ref{lem:pw}, and show that 
$\Vert \wh{\mbf V} - \mbf W^b_\chi\wt{\mbf S} \Vert = O_p\l(\sqrt{\frac{\log\,n}{T}} \vee \frac{1}{n}\r)$.
\hfill$\Box$

\subsubsection{Proof of Lemma \ref{lem:lam}}
\label{lem:lam:proof}

Let $\wh{\mbf W}^b_x = [\wh{\mbf w}^b_{x, 1}, \ldots, \wh{\mbf w}^b_{x, r_b}]$.
Since $\wh{\mbf V}^\top \wh{\mbf W}^b_x = \mbf O_{(k-r_b)\times r_b}$ and 
$\Vert\wh{\mbf V}\Vert=1$, we have
\bea
\Vert \wh{\mbf V}^\top \mbf W^b_\chi \Vert = \Vert \wh{\mbf V}^\top \mbf W^b_\chi \wt{\mbf S} \Vert
= \Vert \wh{\mbf V}^\top (\wh{\mbf W}^b_x - \mbf W^b_\chi\wt{\mbf S}) \Vert
\leq \Vert \wh{\mbf W}^b_x - \mbf W^b_\chi\wt{\mbf S} \Vert
= O_p\l(\sqrt{\frac{\log\,n}{T}} \vee \frac{1}{n}\r)
\label{eq:v:w}
\eea
from Lemma \ref{lem:pw:evec}.
Note that $\bm\Gamma^b_\chi = \mbf W^b_\chi \mbf M^b_\chi \mbf W_\chi^{b\, \top}
= \bm\Lambda_b \bm\Gamma^b_f \bm\Lambda_b^\top$.
Then, 
\beas
\Vert \wh{\mbf V}^\top\bm\Lambda_b \Vert 
&=& \Vert \wh{\mbf V}^\top \mbf W^b_\chi \mbf M^b_\chi \mbf W_\chi^{b\, \top}
\bm\Lambda _b(\bm\Lambda_b^\top\bm\Lambda)_b^{-1} (\bm\Gamma_f^{b})^{-1} \Vert
\le \Vert \wh{\mbf V}^\top \mbf W^b_\chi \Vert\, \Vert \mbf M^b_\chi \Vert\, 
\Vert \bm\Lambda_b (\bm\Lambda_b^\top\bm\Lambda_b)^{-1} \Vert\,
\Vert (\bm\Gamma^b_f)^{-1} \Vert
\\
&=& O_p\l\{\Big(\sqrt{\frac{\log\,n}{T}} \vee \frac{1}{{n}}\Big) \cdot n \cdot \frac{1}{\sqrt n}\r\} 
= O_p\l(\sqrt{\frac{n\log\,n}{T}} \vee \frac1{\sqrt n}\r).
\eeas
using \eqref{eq:v:w} and Assumption \ref{assum:three:new}.\hfill$\Box$

\subsection{Proof of Proposition \ref{thm:pca}}

We can show the consistency of the PCA-based estimator of the common components
within each segment (in the sense of Theorem \ref{thm:common}),
by establishing consistency of the $r_b$ leading eigenvectors of $\wh{\bm\Gamma}^b_x$
in estimating the leading eigenvectors of $\bm\Gamma^b_\chi$ up to a rotation.
see the proof of Lemma \ref{lem:cov} (i). 



Lemma \ref{lem:pw} shows the element-wise consistency of $\wh{\bm\Gamma}_x^b$ over each $I^\chi_b$
defined by the estimated change-points $\wh\eta^\chi_b, \, b = 1, \ldots, \wh{B}_\chi$. 
Recalling that $\omega_{n, T} \ll T^{1/2}$ for estimating the change-points in the common components, 
we have $\sqrt{\log\,n/T}$ dominate the RHS of \eqref{lem:pw:eq:one}. 
Therefore,
\[
\frac 1 n\Vert  \wh{\bm\Gamma}^b_x- {\bm\Gamma}^b_{\chi}\Vert\le \frac{1}{n} \Vert \wh{\bm\Gamma}^b_x - \bm\Gamma^b_x \Vert_F +\frac 1 n\Vert\bm\Gamma^b_{\eps} \Vert= 
O_p\l(\sqrt{\frac{\log\,n}{T}}\vee \frac 1 n \r),
\]
as $\mu_{\eps,1}^b<C_{\eps}$  from Assumption \ref{assum:four} (i). 
Given this result, the arguments adopted in the proof of Theorem \ref{thm:common} are applicable
verbatim for the proof of Proposition \ref{thm:pca} and details are therefore omitted.
\hfill$\Box$


\section{WT for change-point analysis of the idiosyncratic components}
\label{sec:when:y:e}

The WT proposed in Section \ref{sec:choice}
is directly applicable to $\wh\eps^k_{it}$ for change-point analysis in the idiosyncratic components.
For completeness, we state the corresponding arguments below.

Recall the notation $\gamma(t) = \max\{0 \le b \le B_\eps: \, \eta^\eps_b+1 \le t\}$. 
Then, $g_j(\wh \eps^k_{it})$ and $h_j(\wh \eps^k_{it}, \wh \eps^k_{i't})$ admit the following decompositions
\begin{align}
g_j(\wh \eps^k_{it}) &= \E\{g_j(\eps^{\gamma(t)}_{it})\} + 
[\E\{g_j(\eps_{it})\} - \E\{g_j(\eps^{\gamma(t)}_{it})\}] + [g_j(\eps_{it}) - \E\{g_j(\eps_{it})\}] 
+ \xi_j(\wh \eps^k_{it}), 
\label{eq:eps:g:decom}
\\
h_j(\wh \eps^k_{it}, \wh \eps^k_{i't}) &= \E\{h_j(\eps^{\gamma(t)}_{it}, \eps^{\gamma(t)}_{i't})\} 
+ [\E\{h_j(\eps_{it}, \eps_{i't})\} - \E\{h_j(\eps^{\gamma(t)}_{it}, \eps^{\gamma(t)}_{i't})\}] 
\nn \\
& + [h_j(\eps_{it}, \eps_{i't}) - \E\{h_j(\eps_{it}, \eps_{i't})\}] +  \zeta_j(\wh \eps^k_{it}, \wh \eps^k_{i't}).
\label{eq:eps:h:decom}
\end{align}
Under Assumption \ref{assum:one} (iii)--(iv), 
all change-points in $\bm\Gamma_\eps^{\gamma(t)}(\tau)$ at lags $|\tau| \le \bar{\tau}_\eps$, 
namely all $\eta^\eps_b \in \cB^\eps$, appear as change-points in the piecewise constant signals
$\{\E\{g_j(\eps^{\gamma(t)}_{it})\}, \, 1 \le i \le n, \, 
\E\{h_j(\eps^{\gamma(t)}_{it}, \eps^{\gamma(t)}_{i't})\}, \, 1 \le i < i' \le n\}$ 
at scales $j \ge -\jts$.
Analogous to Lemmas \ref{lem:chi:i}--\ref{lem:chi:a9},
we establish the bound on the scaled partial terms in \eqref{eq:eps:g:decom}--\eqref{eq:eps:h:decom}.

\begin{lem}
\label{lem:eps:i}
{\it Suppose that the conditions of Theorem \ref{thm:overestimation} are met. At $j \ge -\jts$ and $k \ge r$,
\bit
\item[(i)]
\begin{align*}
\max_{1 \le s < e \le T} \frac{1}{\sqrt{e-s+1}} & \Big\{
\max_{1 \le i \le n} \Big\vert \sum_{t=s}^e \E\{g_j(\eps_{it})\} - \E\{g_j(\eps^{\gamma(t)}_{it})\} \Big\vert \vee \\
& \max_{1 \le i < i' \le n} \Big\vert \sum_{t=s}^e \E\{h_j(\eps_{it}, \eps_{i't})\} - 
\E\{h_j(\eps^{\gamma(t)}_{it}, \eps^{\gamma(t)}_{i't})\} \Big\vert \Big\}
= O(\log^\upsilon T).
\end{align*}
\item[(ii)]
\beas
 \max_{1 \le s < e \le T} \frac{1}{\sqrt{e-s+1}} \Big\{
\max_{1 \le i \le n} \Big\vert \sum_{t=s}^e \xi_j(\wh \eps^k_{it}) \Big\vert \vee 
\max_{1 \le i < i' \le n} \Big\vert \sum_{t=s}^e \zeta_j(\wh \eps^k_{it}, \wh \eps^k_{i't}) \Big\vert \Big\}
= O(\log^{\theta+\upsilon}T).
\eeas
\item[(iii)]
\begin{align*}
& \max_{1 \le s < e \le T} \frac{1}{\sqrt{e-s+1}} \Big[
\max_{1 \le i \le n} \Big\vert \sum_{t=s}^e g_j(\eps_{it}) - \E\{g_j(\eps_{it})\} \Big\vert 
\\
& \qquad \qquad \qquad \vee \max_{1 \le i < i' \le n}
\Big\vert \sum_{t=s}^e h_j(\eps_{it}, \eps_{i't}) - \E\{h_j(\eps_{it}, \eps_{i't})\} \Big\vert\Big] 
= O_p(\log^{\theta+\upsilon}T).
\end{align*}
\eit}
\end{lem}

The proof of Lemma \ref{lem:eps:i} take the analogous steps 
as the proofs of Lemmas \ref{lem:chi:i}--\ref{lem:chi:a9} and thus is omitted.
Then Proposition \ref{prop:e:additive} below holds for the WT of $\wh\eps^k_{it}$.
Again, its proof takes the identical arguments as the proof of Proposition \ref{prop:chi:additive} thus is omitted.
\begin{prop}
\label{prop:e:additive}
{\it Suppose that all the conditions in Theorem \ref{thm:overestimation} hold. 
For some fixed $k \ge r$ and $\jts = \lfloor C\log_2\log^\upsilon T\rfloor$, 
consider the $N = \jts\, n(n+1)/2$-dimensional panel in \eqref{eq:eps:panel0},
and denote as $y_{\ell t}$ a generic element of the panel. Then, we have the following decomposition:
\beas
y_{\ell t} = z_{\ell t} + \vep_{\ell t}, \quad \ell = 1, \ldots, N,  \quad  t=1, \ldots, T.
\eeas
\bit
\item[(i)] $z_{\ell t}$ are piecewise constant as the corresponding elements of \eqref{eq:eps:z}.
That is, all change-points in $z_{\ell t}$ belong to $\cB^\eps = \{\eta^\eps_1, \ldots, \eta^\eps_{B_\eps}\}$ and
for each $b \in \{1, \ldots, B_\eps\}$, there exists at least a single index $\ell \in \{1, \ldots, N\}$ for which 
$|z_{\ell, \eta^\eps_b+1} - z_{\ell \eta^\eps_b}| \ne 0$.
\item[(ii)]
$\max_{1 \le \ell \le N} \max_{1 \le s < e \le T} (e-s+1)^{-1/2} \;
\vert \sum_{t=s}^e \vep_{\ell t} \vert = O_p(\log^{\theta+\upsilon}T)$.
\eit}
\end{prop}


\section{Implementation of the proposed methodology}
\label{sec:tricks}

\subsection{Generation of the additive panel data from $\wh\chi^k_{it}$}

In the implementation of our proposed methodology,
we use the input panel data of reduced dimension $\wt{N} = \jts \, n$ instead of $N = \jts\,n(n+1)/2$:
\bea
\label{eq:red:panel}
\{g_j(\wh\chi^k_{it}), \, 1 \le i \le n; \, -\jts \le j \le -1; \, 1 \le t \le T\}.
\eea
The motivation behind such a modification is two-fold.
Firstly, it substantially reduces the computational cost by speeding up
not only the computation of the $\wt{N} \times (e-s)$-dimensional array of DC statistics over the segment $[s, e]$, 
but also the resampling procedure adopted in the bootstrap algorithm for threshold selection.

Besides, suppose that a change-point, say $\eta^\chi_b$,
appears as jumps only in $h_j(\chi^{\beta(t)}_{it}, \chi^{\beta(t)}_{i't})$ at the population level, i.e.,
\beas
\bm\lambda_i^\top\l[\sum_{|\tau| < \cL_j} \Big\{\wt{\bm\Gamma}^{b}_f(\tau) - 
\wt{\bm\Gamma}^{b-1}_f(\tau)\Big\}\Psi_j(\tau)\r]\bm\lambda_i = 0
\quad \mbox{for all } i=1, \ldots, n \mbox{ and } j \ge -\jts,
\eeas
where $\wt{\bm\Gamma}^b_f(\tau) = \E\{\mbf f^b_{t+\tau}(\mbf f^b_t)^\top\}$.
Such a constraint may be viewed as an over-determined system
consisting of $\jts\,n$ equations and $\cL_{-\jts}r(r+1)/2$ variables.
With $n \to \infty$ and $r$ fixed, there exists only a trivial solution, namely
$\wt{\bm\Gamma}^{b}_f(\tau) - \wt{\bm\Gamma}^{b-1}_f(\tau) = \mbf O$ 
for all $|\tau| < \cL_{-\jts}$.
Therefore, it is reasonable to expect that
every change-point in the second-order structure of the common components is detectable from (\ref{eq:red:panel}).

\subsection{Choice of the number of wavelet scales}

In the paper, we have investigated the DCBS algorithm applied to the panel that consists of  
wavelet-transformed $\wh\chi_{it}^k$ (or $\wh\eps_{it}^k$) at multiple scales simultaneously.
However, it is possible to apply the algorithm sequentially:
\begin{itemize}
\item[(a)] first apply the DCBS algorithm to the panel consisting of wavelet-transformed 
$\wh\chi^k_{it}$ ($\wh\eps_{it}^k$) at scale $-1$ only,
and denote the set of estimated change-points by $\wh{\mc B}$;
\item[(b)] apply the DCBS algorithm to wavelet-transformed $\wh\chi^k_{it}$ ($\wh\eps_{it}^k$) at scale $-2$
over each segment defined by two adjacent change-points in $\wh{\mc B}$,
and add the estimated change-points to $\wh{\mc B}$;
\item[(c)] repeat (b) with wavelet scales $j = -3, -4, \ldots$ until $\wh{\mc B}$ does not change
from one scale to another.
\end{itemize}
The above procedure is motivated by the fact that the finer wavelet scales are preferable 
for the purpose of change-point estimation,
and allows for the data-driven choice of $\jts$.

\subsection{Binary segmentation algorithm}

\subsubsection{Trimming off of the intervals for change-point analysis}
\label{sec:dt}

In practice, we introduce an additional parameter $d_T$ in order to
ensure that the interval of interest $[s, e]$ is of sufficient length ($e-s+1 > 4d_T$); and
account for possible bias in the previously detected change-points,
by trimming off the short intervals of length $d_T$ around the previously identified change-points in 
Steps 1.2--1.3 of the DCBS algorithm:
\beas
\cT_{s, e} = \max_{b\in[s+d_T, e-d_T]}\max_{1 \le m \le N} \cD_{s, b, e}(m) \quad \mbox{and} \quad
\heta = \arg\max_{b\in[s+d_T, e-d_T]}\max_{1 \le m \le N} \cD_{s, b, e}(m).
\eeas
Due to the condition on the spread of change-points in Assumption \ref{assum:six}, 
this adjustment does not affect the theoretical consistency in the detected change-points,
while ensuring that the bias in estimated change-points 
do not hinder the subsequent search for change-points empirically.
With regards to the discussion in Remark \ref{rem:beta} and the bias presented in Theorem \ref{thm:dcbs}, 
we propose to use $d_T = [\log^2 T \wedge 0.25T^{6/7}]$.

Though our methodology requires the selection of a multitude of parameters as listed in this section,
we observe that its empirical performance is relatively robust to their choices.
The parameter that exerts the most influence is $d_T$:
smaller values of $d_T$ tend to return a larger set of estimated change-points, some of which may be spurious 
as a result of bias associated with previously detected change-points.
On the other hand, if $d_T$ is set to be too large, some change-points may be left undetected due to the restrictions imposed by its choice.
The default choice of $d_T$ recommended above worked well for the simulated datasets.
Also, with some prior knowledge on the dataset, $d_T$ may be selected accordingly, as we have done in real data analysis.

\subsubsection{Implementation of the DCBS algorithm}
\label{sec:rule}

Steps 1--2 of the bootstrap algorithm proposed in Section \ref{sec:bootstrap}
generate panel data $y^\bullet_{\ell t}$ of the same dimensions as the original 
$\{y_{\ell t}, \, \ell = 1, \ldots, N; \, t = 1, \ldots, T\}$.
Therefore, $\cT_{s, e}^\bullet$ over different intervals $[s_{u, v}, e_{u, v}]$ (with $s = s_{u, v}$ and $e = e_{u, v}$),
which are examined at certain iterations of the DCBS algorithm,
can be computed from the same bootstrap panel data $y^\bullet_{\ell t}$.
Since storing $R$ copies of such high-dimensional panel data ($N \times T$) is costly, 
we propose to apply the DCBS algorithm with the bootstrap algorithm as proposed below.

Firstly, a binary tree of a given height, say $\hbar$, is grown from $y_{\ell t}$ with $(s_{u, v}, e_{u, v})$ as its nodes,
by omitting the testing procedure of Step 1.3 when applying the DCBS algorithm until $u = \hbar$.
We impose a minimum length constraint that requires $e_{u, v} - s_{u, v} + 1 \ge 4d_T$; 
if not, the interval $[s_{u, v}, e_{u, v}]$ is not split further.

Let $\cI = \{(s_{u, v}, e_{u, v}):\, 1 \le v \le I_u; \, 1 \le u \le \hbar\}$
be the collection of the nodes in the thus-generated binary tree,
with $I_u$ denoting the number of nodes at level $u$.
Then, the following is repeated for $m = 1, \ldots, R$:
we generate the bootstrap panel data $y^{m\bullet}_{\ell t}$ using the bootstrap algorithm,
and compute $\cT_{s, e}^{m\bullet}$ for all $(s, e) \in \cI$, from which the corresponding thresholds are drawn.
Once it is complete, we perform Step 1.3 by setting $(s, e) = (s_{u, v}, e_{u, v}) \in \cI$; 
starting from $(u, v) = (1, 1)$, we progress in the order $(u, v) = (1, 1), (2, 1), \ldots, (2, I_2), (3, 1), \ldots$.
If $\cT_{s, e} \le \pi_{N, T}$ for some $(s, e)$,
all the nodes $(s', e')$ that are dependent on $(s, e)$ (in the sense that $[s', e'] \subset [s, e]$)
are removed from $\cI$, which ensures that we are indeed performing a binary segmentation procedure.
$\hbar = [\log_2\,T/2]$ is used in simulation studies and real data analysis,
which grows a sufficiently large tree
considering that we can grow a binary tree of height at most $\lfloor \log_2\,T \rfloor$.

\subsection{Stationary bootstrap and choice of parameters}
\label{app:SB}

Stationary bootstrap (SB) generates bootstrap samples $X^\bullet_1, \ldots, X^\bullet_T$ as below. 
Let $J_1, J_2, \ldots$ be i.i.d. geometric random variables independent from the data, 
where $\p(J_1 = L) = p(1-p)^{L-1}$ with some $p\in (0, 1)$,
 $I_1, \ldots, I_Q$ be i.i.d. uniform variables on $\{1, \ldots, T\}$,
and denote a data block by $B(t, L) = (X_t, \ldots, X_{t+L-1})$ for $t, L \ge 1$ with a periodic extension 
($X_t = X_{t-uT}$ for some $u \in \Z$ such that $1 \le t - uT \le T$). 
The SB sample $X^\bullet_1, \ldots, X^\bullet_T$ is generated as the first $T$ observations in the sequence
$B(I_1, J_1), \ldots, B(I_Q, J_Q)$ for $Q = \min\{q:\, \sum_{l=1}^q J_q \ge T\}$.

We used the bootstrap sample size $R = 200$ for simulation studies, and $R = 500$ for real data analysis.
Naturally, the larger bootstrap sample size is expected to return the better choice for the threshold,
but we did not observe any sensitivity due to the choice of the size of bootstrap samples in our simulation studies.
The binary segmentation procedure implicitly performs multiple testing,
with the number of tests determined by the choice of $\hbar$ described in Appendix \ref{sec:rule}.
However, the test statistics computed at different iterations of the DCBS algorithm are correlated.
Therefore, with the added difficulty arising from the hierarchy inherent in the procedure,
the task of controlling for multiple testing is highly challenging.
Instead, noticing that the bootstrap test statistics tend to have heavy tails,
we chose to adopt the same $\alpha = 0.05$ at all iterations of the DCBS algorithm.

In the SB, $p$ stands for the inverse of average block length.
Typically, large $p$ leads to large bias and small variance and vice versa.
For approximating the finite sample distribution of the sample mean of a univariate series $\{X_t\}_{t=1}^T$,
\cite{politis2004} (accompanying corrections in \cite{patton2009}) proposed to select $p$ as 
$p^{-1} = (\wh{G}^2/\wh{g}^2(0))^{1/3}T^{1/3}$,
where $\wh{G} = \sum_{k=-\Lambda}^\Lambda \lambda(k/\Lambda)|k|\wh{R}(k)$,
$\wh{g}(0) = \sum_{k=-\Lambda}^\Lambda \lambda(k/\Lambda)\wh{R}(k)$, $\wh{R}(k) = T^{-1}\sum_{t=1}^{T-|k|}(X_t-\bar{X})(X_{t+|k|}-\bar{X})$
and $\lambda(z)$ is a trapezoidal shape symmetric taper around zero:
\begin{eqnarray*}
\lambda(z) = \l\{\begin{array}{ll}
1 & \mbox{for } |z| \in [0, 1/2), \\
2(1-|z|) & \mbox{for } |z| \in [1/2, 1], \\
0 & \mbox{otherwise}.
\end{array}
\r.
\end{eqnarray*}
For the stationarity testing in multivariate (finite-dimensional) time series, \cite{jentsch2015} proposed to modify the above $p$ 
as per the theory developed for the consistency of their stationary bootstrap,
by (a) replacing $T^{1/3}$ with $T^{1/5}$, and
(b) taking the average of $p^{-1}$ associated with each univariate series.
Taking their approach, we choose to select $p_j$ for individual $\wh f_{jt}, \, j = 1, \ldots, k$ 
according to (a) in Step 1 for the common components,
and select $p_i$ for individual $\wh{\eps}_{it}^{k}, \, i=1, \ldots, n$ according to (a) and
use $(n^{-1}\sum_{i=1}^n p_i^{-1})^{-1}$ according to (b) in the same step for the idiosyncratic components.
As for $\Lambda$, we plug in the automatically chosen bandwidth 
based on the autocorrelation structure of $\wh f_{jt}$ and $\wh \eps^k_{it}$
as suggested in \cite{politis2004}.

\end{document}